\documentclass[
  aps,
  pra,
  onecolumn,
  superscriptaddress,
  nofootinbib,
  longbibliography
]{revtex4-2}

\usepackage[T1]{fontenc}
\usepackage[utf8]{inputenc}
\usepackage{graphicx}
\usepackage{bm}
\usepackage{bbm}
\usepackage{amsmath,amssymb,amsfonts,mathtools}
\usepackage{amsthm}
\usepackage{booktabs}
\usepackage{multirow}
\usepackage{array}
\usepackage{makecell}
\usepackage{siunitx}
\usepackage{dcolumn}
\usepackage{braket}
\usepackage{qcircuit}
\usepackage{xcolor}
\usepackage{microtype}
\usepackage{url}
\definecolor{NavyBlue}{HTML}{000080}
\usepackage[colorlinks=true,linkcolor=red,citecolor=NavyBlue,urlcolor=magenta]{hyperref}

\theoremstyle{plain}
\newtheorem{theorem}{Theorem}

\newtheorem{lemma}{Lemma}
\newtheorem{corollary}{Corollary}
\theoremstyle{definition}

\theoremstyle{remark}

\allowdisplaybreaks
\newcommand{\orange}[1]{\textcolor{orange}{#1}}

\newcommand{\blue}[1]{\textcolor{blue}{#1}}

\def\O{\mathcal{O}}

\def\R{\mathbb{R}}

\def\Tr{{\rm Tr}}
\newcommand{\E}{\mathop{\mathbb{E}}}
\newcommand{\Var}{\mathop{\rm Var}}

\def\0{\bm{0}}
\def\1{\bm{1}}
\def\2{\bm{2}}
\def\>{\rangle}
\def\<{\langle}

\def\bmt{\bm{\theta}}
\def\bmo{\bm{\omega}}

\begin{document}

\title{Stochastic Pauli-path simulator for large-scale quantum optimization}

\author{Kaining Zhang}
\affiliation{Generative AI Lab, College of Computing and Data Science, Nanyang Technological University, Singapore 639798, Singapore}

\author{Xinbiao Wang}
\affiliation{Generative AI Lab, College of Computing and Data Science, Nanyang Technological University, Singapore 639798, Singapore}

\author{Kunsheng Li}
\affiliation{Generative AI Lab, College of Computing and Data Science, Nanyang Technological University, Singapore 639798, Singapore}

\author{Qixin Zhang}
\affiliation{Generative AI Lab, College of Computing and Data Science, Nanyang Technological University, Singapore 639798, Singapore}

\author{Yuxuan Du}
\email{yuxuan.du@ntu.edu.sg}
\affiliation{Generative AI Lab, College of Computing and Data Science, Nanyang Technological University, Singapore 639798, Singapore}
\affiliation{School of Physical and Mathematical Sciences, Nanyang Technological University, Singapore 639798, Singapore}

\author{Min-Hsiu Hsieh}
\email{min-hsiu.hsieh@foxconn.com}
\affiliation{Hon Hai (Foxconn) Research Institute, Taipei, Taiwan}

\author{Dacheng Tao}
\email{dacheng.tao@ntu.edu.sg}
\affiliation{Generative AI Lab, College of Computing and Data Science, Nanyang Technological University, Singapore 639798, Singapore}

\date{\today}

\begin{abstract}
Pauli-based simulators offer a promising route to large-scale classical simulation of quantum circuits in the low-magic regime. Yet their applicability remains largely limited to forward simulation, making them inadequate for optimization-driven quantum tasks such as variational state preparation and parameter initialization. Existing approaches either lack native support for gradient-based optimization or suffer from severe gradient bias. Here we propose the stochastic Pauli-path simulator (\texttt{SPPS}), a computational framework for large-scale quantum optimization that enables unbiased stochastic gradient estimation via Pauli-path sampling across optimization iterations. Our theoretical analysis shows that the proposed simulator yields unbiased gradient estimates and admits provable convergence guarantees. We systematically evaluate our proposal, including quantum eigensolver benchmarks with up to 100 qubits and quantum neural network benchmarks with up to 40 qubits. Across these tasks, \texttt{SPPS} faithfully tracks optimization dynamics, converges within minutes, and broadens the role of Pauli-based simulation from forward estimation to large-scale quantum optimization.

\end{abstract}

\maketitle

\section{Introduction}
\label{intro}

Advances in classical simulation are indispensable to the progress of quantum computing~\cite{doi:10.1137/050644756,PhysRevLett.116.250501,XU20254104,10.1145/3762672}. Over the years, continued progress in classical simulation has not only provided practical platforms~\cite{pastaq,broughton2021tensorflowquantumsoftwareframework,Zhang2023tensorcircuit,Gidney2021stimfaststabilizer,10313722,asadi2024hybridquantumprogrammingpennylane} for the design and development of quantum algorithms in the absence of large-scale quantum machines, but also imposed increasingly stringent standards on claims of quantum advantage~\cite{Arute2019,doi:10.1126/science.abe8770,Kim2023,RevModPhys.95.035001}. A notable example is random circuit sampling~\cite{Arute2019}, where newly developed tensor-network methods have repeatedly revisited and in some cases narrowed the claimed quantum advantage~\cite{huang2020classical,10.1145/3458817.3487399,PhysRevLett.128.030501,10.1093/nsr/nwae317}. More recently, Pauli-based simulation ($\texttt{PBS}$) methods~\cite{gao2018efficient,PhysRevA.99.062337,xct1-7kf2,PRXQuantum.6.020302,10.1063/5.0269149,Fontana2025,GonzalezGarcia2025paulipath,j1gg-s6zb,PhysRevLett.133.120603,lh6x-7rc3,Lin2026utilityscalequantum,monaco2025symbolic,rudolph2025pauli,bermejo2024quantum,10.21468_SciPostPhysCodeb,d2025circuit,miller2025simulation,upreti2025interplay,upreti2025quantum,facelli2026fast} have opened a complementary route to classically simulating quantum dynamics and estimating their \textit{mean values}. By propagating Pauli representations of operators during the evolution, they remain effective even in regimes of highly entangling quantum dynamics that are challenging for tensor-network methods. Attributed to this capability, \texttt{PBS} methods redefine the classical boundary for utility-scale experiments~\cite{Kim2023,doi:10.1126/sciadv.adk4321}, and advance fundamental science, e.g.,  exploring thermal states and imaginary-time evolution~\cite{gomez2026pauli,rudolph2026thermal}.

Despite the progress, prior \texttt{PBS} methods remain limited in optimization functionality. To be concrete, most of \texttt{PBS} methods~\cite{gao2018efficient,PhysRevA.99.062337,xct1-7kf2} are designed mainly for fixed-circuit simulation. The only exception is truncation-based \texttt{PBS} (\texttt{Tb-PBS}) methods, a.k.a, Pauli propagation simulators~\cite{j1gg-s6zb,PhysRevLett.133.120603,lh6x-7rc3}. They construct a differentiable map from circuit parameters to target observables, while using truncation strategies to control the proliferation of Pauli terms during propagation. Yet as empirically observed in Ref.~\cite{4kyq-n8jb}, the obtained gradients are often \textbf{systematically biased} from the exact gradients. This fundamental limitation restricts the applicability of \texttt{Tb-PBS} in various quantum optimization tasks, where circuit parameters are updated by gradient-based optimizers to minimize a task-specific objective. 
Typical examples cover finding high-overlap initial states for quantum phase estimation~\cite{PhysRevLett.133.250601}, optimizing reference states for eigenstate filtering~\cite{Lin2020optimalpolynomial}, constructing warm-start parameters for state preparation~\cite{Kuzmin2020variationalquantum,PhysRevApplied.16.054035} and circuit compiling~\cite{Khatri2019quantumassisted,Jones2022robustquantum}, and pre-training variational quantum algorithms~\cite{Cerezo2021}.
The severe limitation of \texttt{PBS} towards such wide applications raises a critical challenge:
\par\noindent{\centering
   \textit{Is there any \texttt{PBS} that can advance quantum optimization tasks at scale?} \par
}
Addressing this challenge would have two important implications: (\textbf{\textit{i}}) it would redefine the classical bar for quantum advantage by bringing optimization dynamics into the scope of efficient simulation; (\textbf{\textit{ii}}) it would open a practical route to classical-first optimization for quantum applications with high experimental overhead, substantially reducing the quantum resources required for deployment.

Here we provide a positive answer to this question by proposing the \textit{stochastic Pauli-Path simulator} (\texttt{SPPS}), an efficient and scalable \texttt{PBS} for large-scale quantum optimization. In contrast to prior \texttt{Tb-PBS} methods, the key idea of \texttt{SPPS} is to dynamically sample propagation paths from the full path space during Heisenberg evolution and correct each sampled contribution by importance reweighting. To this end,  we devise a \textit{path automatic differentiation method}, which effectively converts path contribution into unbiased gradient estimates through simple algebraic score factors. 
On the theoretical side, we prove that unlike \texttt{Tb-PBS}s that suffer from biased gradients (Theorem~\ref{thm_toy_bias}) and sub-optimal convergence (Corollary~\ref{cor_opt_gap}), \texttt{SPPS} enables unbiased gradients with guaranteed accuracy (Theorem~\ref{thm_spp_estimator}), which leads to provably convergent optimization trajectories (Corollary~\ref{cor_spp_sgd}). These theoretical analyses reveal that \texttt{SPPS} can faithfully simulate quantum optimization driven by gradient-based optimizers, with sample complexity depending on the circuit structure and parameters.

To validate the effectiveness of \texttt{SPPS}, we conduct systematic quantum optimization experiments, including pre-training variational quantum eigen-solvers (VQEs)~\cite{Peruzzo2014,rudolph2023synergistic,PhysRevResearch.5.043217,khan2023pre}, pre-training quantum neural networks (QNNs)~\cite{bermejo2024quantum,Benedetti_2019,verdon2019learning}, and preparing quantum state encoding circuits~\cite{PhysRevResearch.4.023136,PhysRevA.109.052423,zhang2026aqer}. Across these standard benchmarks, \texttt{SPPS} consistently achieves better convergence than \texttt{Tb-PBS} methods with substantially less runtime and maintains stable optimization dynamics. Specifically, \texttt{SPPS} completes the pre-training of VQE on 100-qubit Ising model in about one minute and the pre-training of QNN on a synthetic dataset with 40 qubits in less than ten minutes. These results demonstrate \texttt{SPPS} as a route to faithful classical simulation of large-scale quantum optimization.

In summary, our primary contributions are threefold. 
(i) We theoretically characterize the limitation of \texttt{Tb-PBS} methods in quantum optimization, showing the dilemma between the accurate mean-value estimation and the faithful gradient estimation. 
(ii) We develop \texttt{SPPS} to construct unbiased stochastic gradients, with both sample-complexity and convergence guarantees. 
(iii) We conduct systematic experiments up to $100$ qubits, validating that \texttt{SPPS} faithfully tracks optimization dynamics while improving accuracy and runtime over \texttt{Tb-PBS} methods.
The corresponding code is available at \href{https://anonymous.4open.science/r/SPPS-371D/README.md}{\blue{GitHub}} for reproducibility purposes.

\section{Preliminary}
\label{pre}

Here, we present the basics of quantum computing, quantum tasks with optimization, variational quantum algorithms, and \texttt{PBS} methods, followed by related works. Refer to App.~\ref{spp_app_pre} for more details.

\textbf{Basics of quantum computing.} 
The pure state of a single qubit is represented by a normalized vector \(\ket{\psi}\in\mathbb{C}^2\)~\cite{feynman1982simulating}, which admits the computational-basis expansion \(\ket{\psi}=a\ket{0}+b\ket{1}\) with \(a,b\in\mathbb{C}\) and \(|a|^2+|b|^2=1\). More generally, quantum states (including mixed states) are described by a density operator \(\rho\), with pure states as the special case \(\rho=\ket{\psi}\!\bra{\psi}\) and \(\bra{\psi}=\ket{\psi}^\dagger\). An \(n\)-qubit state lives in the tensor product of single-qubit spaces as \((\mathbb{C}^2)^{\otimes n}\).
Quantum operations can be implemented by {gates}~\cite{doi:10.1126/science.273.5278.1073}, i.e., unitary operators acting on one or more qubits, and a {quantum circuit}~\cite{10.1098/rspa.1989.0099} is the composition of a sequence of such gates. Common gates include the Hadamard gate $H = \frac{1}{\sqrt{2}}(\begin{smallmatrix} 1 & 1 \\ 1 & -1 \end{smallmatrix})$ and CNOT gate $ {\rm diag} [ (\begin{smallmatrix} 1 & 0 \\ 0 & 1 \end{smallmatrix}), (\begin{smallmatrix} 0 & 1 \\ 1 & 0 \end{smallmatrix}) ]$. Beyond fixed gates, quantum circuits also employ variational unitaries, e.g., Pauli rotations $R_{P}(\theta)=\exp{- \imath \theta P/2}$, where $P \in \mathcal{P}_n=\{\pm 1\} \times \{I,X,Y,Z\}^{\otimes n}$ is a Pauli operator with $I = (\begin{smallmatrix} 1 & 0 \\ 0 & 1 \end{smallmatrix})$, $X = (\begin{smallmatrix} 0 & 1 \\ 1 & 0 \end{smallmatrix})$, $Y = (\begin{smallmatrix} 0 & -\imath \\ \imath & 0 \end{smallmatrix})$, $Z = (\begin{smallmatrix} 1 & 0 \\ 0 & -1 \end{smallmatrix})$ being Pauli-I, -X, -Y, -Z matrices. A quantum measurement~\cite{Helstrom1969} is the process of extracting classical information from a quantum state, which is specified by a Hermitian operator $O$ known as the observable. Given a state \(\rho\), the measurement of $O$ yields a random variable whose expectation value or mean value is $\Tr[O\rho]$.

\begin{figure*}[t]
    \centering
    \includegraphics[width=0.95\linewidth]{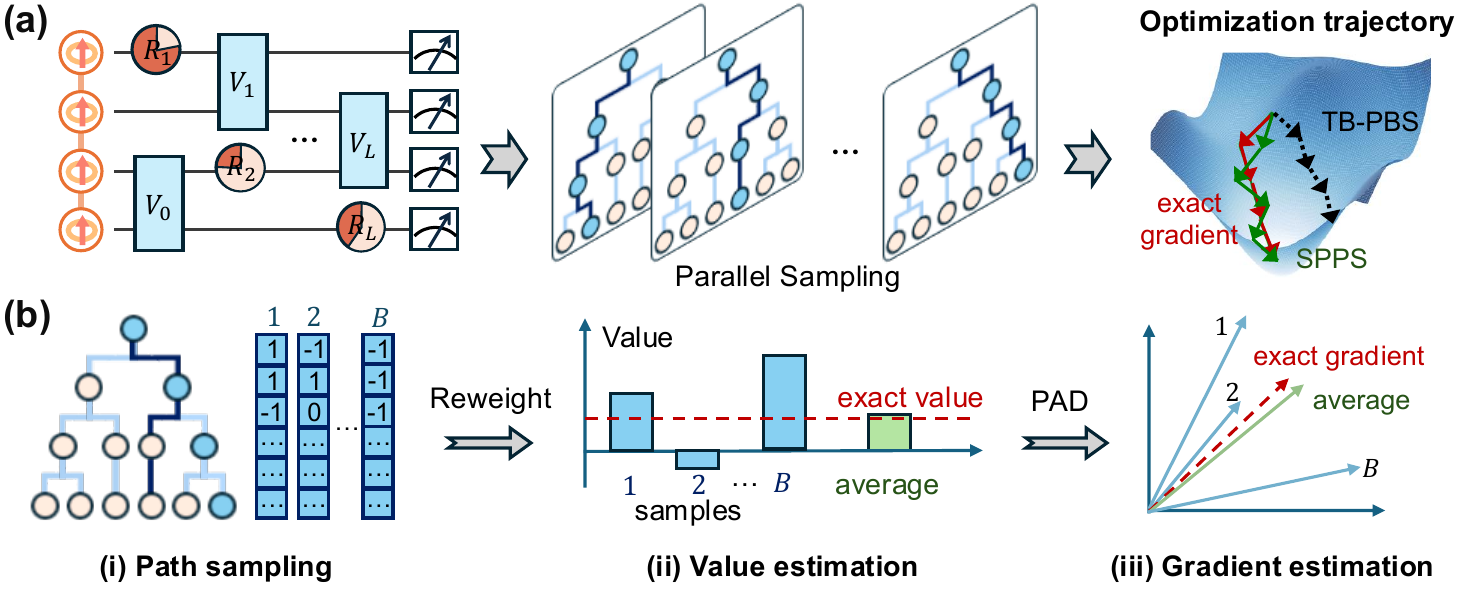}
\caption{ 
{\textbf{Stochastic Pauli-path simulator for gradient estimation.}
\textbf{(a)} Overview of the \texttt{SPPS} framework. \texttt{SPPS} samples propagation paths from the full path space and combines them into unbiased estimators of both expectation values and gradients. The sampled paths are generated independently and can therefore be evaluated in parallel.
\textbf{(b)} Core algorithm of \texttt{SPPS}. A group of sampled paths is first converted into an unbiased estimator of the target expectation value by importance reweighting. The same path-sampling mechanism is then coupled with path automatic differentiation to obtain unbiased stochastic gradient estimators. } }
    \label{dpp_fig2_main_fig}
\end{figure*}

\textbf{Classical simulation of quantum optimization.} 
Many quantum optimization tasks, e.g., VQE~\cite{Peruzzo2014,PhysRevX.6.031007,Kandala2017}, circuit compiling~\cite{Khatri2019quantumassisted,Jones2022robustquantum,wang2023symmetric}, quantum state preparation~\cite{PhysRevResearch.4.023136,PhysRevA.109.052423,zhang2026aqer}, and quantum machine learning (QML)~\cite{PhysRevA.98.032309,Havlicek2019,PhysRevA.101.032308,cheng2023offline,10138016,du2025quantummachinelearninghandson}, can be unified through objectives constructed from expectation values, i.e.,
\begin{align}
C(\bmt)
&=\sum \nolimits_{k=1}^{K}\ell_k\left(f(\bmt;O_k, \rho_k)\right) =\sum \nolimits_{k=1}^{K}\ell_k\left(\Tr\left[O_k U(\bmt)\rho_k U(\bmt)^\dag\right]\right),
\label{dypp_f_loss_eq}
\end{align}
where the input states $\rho_k$, observables $O_k$, and scalar functions $\ell_k$ are specified by the concrete task.
The optimization of Eq.~(\ref{dypp_f_loss_eq}) is typically performed using gradient-based classical optimizers, such as gradient descent (GD) $\bmt^{(t+1)}=\bmt^{(t)}-\eta \nabla C(\bmt^{(t)})$ with learning rate $\eta$.
Classical simulation of these optimization tasks therefore requires access to both expectation values and their gradients, since $\nabla C(\bmt)$ is obtained by applying the chain rule to functions $f(\bmt;O_k,\rho_k)$.
Apart from \texttt{PBS} methods analyzed in this work, common classical simulation approaches include exact state-vector~\cite{bergholm2018pennylane,javadiabhari2024quantumcomputingqiskit} and tensor-network methods~\cite{pastaq,broughton2021tensorflowquantumsoftwareframework}. State-vector simulation supports exact auto-differentiation through an explicit state-vector representation, but scales exponentially with the number of qubits. Tensor-network methods can scale to large systems when states and circuit evolution have limited entanglement, but their cost can grow rapidly beyond such regimes.

\textbf{Variational quantum algorithms.} Variational quantum algorithms (VQAs)~\cite{Cerezo2021} are hybrid quantum--classical algorithms that optimize the parameters $\bmt$ in objectives defined in Eq.~\eqref{dypp_f_loss_eq}. Unlike classical simulators that obtain gradients $\nabla C(\bmt)$ by automatic differentiation, VQAs estimate such gradients \textit{through repeated circuit executions and measurements on real quantum hardware}. For example, for gates satisfying the parameter-shift rule~\cite{PhysRevA.99.032331,crooks2019gradients,Wierichs2022generalparameter}, each partial derivative can be obtained from shifted circuits, e.g., $\partial_j f(\bmt;O,\rho)=\frac{1}{2}\left[f(\bmt+\frac{\pi}{2}\bm{e}_j;O,\rho)-f(\bmt-\frac{\pi}{2}\bm{e}_j;O,\rho)\right]$, where $\bm{e}_j$ is the computational basis vector whose $j$-th entry is $1$. Thus, full-gradient estimation requires shifted circuit evaluations for each trainable parameter, leading to a quadratic total gate-operation cost. Despite their potential advantages in learning performance~\cite{Abbas2021,Caro2022,wang2024separable}, the quadratic scaling of gradient evaluation makes VQA training resource-intensive~\cite{liu2023can} and, in many cases, prohibitively expensive.

\textbf{\texttt{PBS} methods.} 
The function $f(\bmt)$ in Eq.~(\ref{dypp_f_loss_eq}) can be reformulated by propagating the observable in the Heisenberg picture~\cite{gottesman1998heisenberg} as $O \mapsto U(\bmt)^\dag O U(\bmt)$. As shown in Fig.~\ref{dpp_fig2_main_fig}, when $O\in\mathcal{P}_n$ and $U(\bmt)= \prod_{j=1}^L R_{P_j}(\bmt_j)$, $f(\bmt)$ admits the following Pauli expansion
\begin{equation}
\label{dypp_f_loss_pp_eq}
f(\bmt)= \sum\nolimits_{\bmo\in\Omega} \Psi_{\bmo}(\bmt) \Tr[P_{\bmo} (O) \rho] := \sum\nolimits_{\bmo\in\Omega} h_{\bmo}(\bmt).
\end{equation}
Here, $\Omega\subseteq\{0,\pm1\}^L$ denotes the set of Pauli paths $\bmo$. For each path $\bmo$, $\Psi_{\bmo}(\bmt)$ is the accumulated trigonometric weight and $P_{\bmo}(O)\in\mathcal{P}_n$ is the resulting Pauli operator. These paths are generated by sequential propagation over $j=1,\ldots,L$. For example, let $P$ be the Pauli operator after the first $j-1$ steps. If $P$ commutes with $P_j$, the propagation gives a single branch $P\mapsto P$ labeled by $\bmo_j=0$. Otherwise, it splits as $P \mapsto \cos(\bmt_j)P+\sin(\bmt_j)\imath P_jP$, where $\bmo_j=1$ and $-1$ label the cosine and sine branches, respectively. Eq.~(\ref{dypp_f_loss_pp_eq}) provides the foundation of \texttt{Tb-PBS} methods~\cite{xct1-7kf2,PRXQuantum.6.020302,Fontana2025}, which truncate the set $\Omega$ to make the simulation practical. Refer to App.~\ref{app_PP_notion} for the details about how to generalize Eq.~(\ref{dypp_f_loss_pp_eq}) to arbitrary observables and circuits interleaved with Clifford gates~\cite{PhysRevA.71.042315}.

\textbf{Related works.}
Classical simulators for quantum optimization include state-vector~\cite{bergholm2018pennylane,javadiabhari2024quantumcomputingqiskit}, tensor-network~\cite{pastaq,javadiabhari2024quantumcomputingqiskit}, and Pauli-based methods~\cite{Lin2026utilityscalequantum,monaco2025symbolic,rudolph2025pauli,bermejo2024quantum}, which are complementary in scope. 
Specifically, state-vector simulators provide exact gradients but scale exponentially with qubit number, tensor-network simulators support large systems in low-entanglement regimes, while Pauli-based methods exploit low-magic structure and can handle large-scale and highly entangled quantum systems. 
In the low-magic regime,  \texttt{Tb-PBS} methods have shown promise for simulating optimization tasks such as VQEs~\cite{Lin2026utilityscalequantum,monaco2025symbolic} and QNNs~\cite{bermejo2024quantum}. 
However, a recent empirical study~\cite{4kyq-n8jb} indicated an intrinsic limitation of these approaches, as the gradient information is highly biased. 
Our proposal addresses this fundamental limitation, pushing the frontier of \texttt{PBS} towards various optimization tasks at scale.

\section{Biased gradient estimation for \texttt{Tb-PBS}}
\label{dpp_method_biased_grad}

Here we systematically analyze the gradient bias induced by truncations used in \texttt{Tb-PBS}s. To this end, we first unify different \texttt{Tb-PBS} methods into the same framework. The unified framework allows us to prove that these methods can always encounter non-vanishing gradient errors.

\textbf{A unified framework of \texttt{Tb-PBS}.} Since the path set $\Omega$ in Eq.~(\ref{dypp_f_loss_pp_eq}) can grow exponentially with the number of parameters in $\bmt$, \texttt{Tb-PBS} methods construct a reduced subset $\hat{\Omega}\subsetneq\Omega$ via some threshold-based rules and approximate $f(\bmt)$ in Eq.~(\ref{dypp_f_loss_pp_eq}), i.e.,
\begin{equation}\label{dypp_f_loss_pp_truncate_eq}
\hat{f}(\bmt)=\sum\nolimits_{\bm{\omega}\in\hat{\Omega}} h_{\bmo}(\bmt).
\end{equation}
There are three truncation strategies to control the error $|\hat f(\bmt)-f(\bmt)|$: 
\emph{coefficient truncation} (\texttt{CT}), which keeps paths with coefficient magnitude at least $\tau$ by using 
$\hat{\Omega}=\{\bmo\in\Omega: |\Psi_{\bmo}(\bmt)|\geq \tau\}$~\cite{PRXQuantum.6.020302,Lin2026utilityscalequantum,doi:10.1126/sciadv.adk4321}; 
\emph{frequency truncation} (\texttt{FT}), which retains paths whose Fourier level is at most $\nu$ by using 
$\hat{\Omega}=\{\bmo\in\Omega: \|\bmo\|_0\leq \nu\}$, where $\|\bmo\|_0$ counts nonzero indices in $\bmo$~\cite{Fontana2025,rudolph2023classical,PhysRevA.108.032406}; 
and \emph{Pauli-weight truncation} (\texttt{WT}), which retains paths whose Pauli weights never exceed $\gamma$ during propagation by using 
$\hat{\Omega}=\{\bmo\in\Omega: \|P_{(\omega_1,\ldots,\omega_j,0,\ldots,0)}\|_0\leq\gamma,\ \forall j\in[L]\}$,
where $\|P\|_0$ is the number of non-identity single-qubit Pauli factors in $P$~\cite{xct1-7kf2,PhysRevLett.133.120603,lh6x-7rc3}.  

However, the above \texttt{Tb-PBS} methods do not control the error in the corresponding gradients in general. The following theorem formalizes this separation between function value and gradient accuracy, where the proof is deferred to App.~\ref{app_bias_grad}. 

\begin{theorem}
\label{thm_toy_bias}
We consider the function $f(\bmt)$ in Eq.~(\ref{dypp_f_loss_pp_eq}), where $\rho= ({|0\>}{\<0|})^{\otimes n}$, $O$ is a Pauli observable, and $f$ is non-constant on $\mathbb{R}^L$. Then for any $\epsilon\in(0,1)$, there exist a set $\hat{\Omega}\subsetneq\Omega$ and a region $\Theta\subset\mathbb{R}^L$, such that $\hat{f}(\bmt)$ in Eq.~(\ref{dypp_f_loss_pp_truncate_eq}) satisfies
$| \hat{f}(\bmt)-f(\bmt) | \le \epsilon$ and $\|\nabla_{\bmt}\hat{f}(\bmt)-\nabla_{\bmt}f(\bmt) \|_2 \ge 1-\epsilon$, $\forall \bmt\in\Theta$.
\end{theorem}

Theorem~\ref{thm_toy_bias} shows that a truncated path set can approximate $f(\bmt)$ to arbitrary precision while retaining a non-negligible gradient bias, since paths negligible for the expectation value may still contribute dominantly to the gradient. This result holds in general and does not restrict the circuit structure or the way to generate $\hat{\Omega}$. As explained in App.~\ref{app_worst_case_truncation}, \texttt{CT-}, \texttt{FT-}, or \texttt{WT-PBS} methods can generate $\hat{\Omega}$ in Theorem~\ref{thm_toy_bias} for certain circuits and observables, implying their fundamental limitations.

To reflect how the biased gradients influence the optimization trajectory, we conduct the following analysis. Let the initial parameter be $\bmt^{(0)}$ with learning rate $\eta$, and $\bmt_{\rm ex}^{(t)}$ and $\bmt_{\rm tr}^{(t)}$ separately denote the $t$-th iterations driven by the exact gradient and by the biased gradient from \texttt{Tb-PBS}. The following corollary quantifies the difference between the two optimization trajectories.

\begin{corollary}
\label{cor_opt_gap}
Following the notation of Theorem~\ref{thm_toy_bias}, there exist a circuit $U(\bmt)$, an observable $O$, and the corresponding $\hat{\Omega}\subsetneq\Omega$ produced by either \texttt{CT}, \texttt{FT}, or \texttt{WT}, such that for any $\epsilon\in(0,1)$, there is an initialization $\bmt^{(0)}$ satisfying $|\hat{f}(\bmt_{\rm tr}^{(t)}) - f(\bmt_{\rm tr}^{(t)}) | \le \epsilon$ for arbitrary $t \in \mathbb{N}$, while two GD trajectories using exact and \texttt{Tb-PBS} gradients obey $| f(\bmt_{\rm tr}^{(t)}) - f(\bmt_{\rm ex}^{(t)})| \geq {\Omega}(\eta t \|O\|_2^2)$ for $t \leq \O(1/(\eta\|O\|_2))$. The gap converges to $\Omega(\|O\|_2)$ when $t \rightarrow \infty$.
\end{corollary}

Corollary~\ref{cor_opt_gap} shows that gradient bias can accumulate into a macroscopic optimization error. 
Even when \texttt{Tb-PBS} gives an $\epsilon$-accurate estimate of the objective value at each iterate, its gradient can still move GD in a direction inconsistent with the exact gradient. 
Across $t$ iterations, the biased update direction can accumulate a deviation at the scale of $\eta t \|O\|_2^2$, which can ultimately yield a gap scaling as $\Omega(\|O\|_2)$ at converged points. The proof is deferred to App.~\ref{app_bias_GD}.

\section{\texttt{SPPS}: a faithful simulator for large-scale quantum optimization tasks}
\label{dpp_method_spp}

The optimization gap shown in Corollary~\ref{cor_opt_gap} suggests that simulating quantum optimization requires a design beyond existing \texttt{Tb-PBS} methods. 
To address this issue, here we propose \texttt{SPPS}, a Pauli-propagation-based simulator for estimating gradients faithfully via unbiased sampling of propagation paths from the full path space. For clarity, we first present the implementation details of \texttt{SPPS} in Sec.~\ref{spp_method_implementation}, followed by theoretical analysis in Sec.~\ref{dpp_method_theory}.

\subsection{Algorithm framework of \texttt{SPPS}}
\label{spp_method_implementation}

Instead of constructing a truncated subset $\hat{\Omega}$, \texttt{SPPS} samples propagation paths during Heisenberg evolution and corrects each sampled path by its sampling probability. This construction yields unbiased stochastic gradient estimators, whose error can be systematically reduced by statistical averaging. In the remainder of this subsection, we present the implementation of \texttt{SPPS}.

\textbf{Overview of \texttt{SPPS}.} 
Similar to existing quantum optimization simulators, \texttt{SPPS} is used to simulate optimization trajectories for the objective $C(\bmt)$ in Eq.~(\ref{dypp_f_loss_eq}) by iteratively updating the parameters $\bmt$ with gradient-based optimizers. Without loss of generality, our key focus here is to \textit{exhibit how to use $\mathtt{SPPS}$ to estimate $\nabla_{\bmt} f(\bmt^{(t)})$ at the $t$-th iteration in Eq.~(\ref{dypp_f_loss_pp_eq})}, since $\nabla_{\bmt} C(\bmt^{(t)})$ follows from the gradients $\nabla_{\bmt} f(\bmt)$  via the chain rule. For simplicity, we sometimes write $\bmt:=\bmt^{(t)}$. 

\texttt{SPPS}, as shown in Fig.~\ref{dpp_fig2_main_fig}, consists of three procedures to acquire $\nabla_{\bmt} f(\bmt)$: (i) sampling propagation paths from the full path space through a hierarchical rule; (ii) correcting sampled path contributions by importance reweighting; and (iii) applying path automatic differentiation to obtain unbiased gradient estimators. In the rest of this subsection, we elaborate on these procedures separately.

\textbf{Step (i): Hierarchical path sampling. } The first step of \texttt{SPPS} is to generate propagation paths without explicitly constructing the full propagation tree. 
As illustrated in Fig.~\ref{dpp_fig2_main_fig}(b), this is achieved by sampling from the following distribution {hierarchically}:
\begin{equation}
{\Pr}(\bmo)={\Pr}_1(\bmo_1)\prod\nolimits_{j=2}^{L} {\Pr}_j(\bmo_j\mid \bmo_{1:j-1}).
\label{dypp_method_pr_w}
\end{equation}
Given the unitary $U(\bmt)$ in Eq.~(\ref{dypp_f_loss_pp_eq}), if the current Pauli operator commutes with $P_j$ in $R_{P_j}(\bmt_j)$ and does not branch, the next path variable is deterministic, i.e., ${\Pr}_{j}(\bmo_j=0\mid \bmo_{1:j-1})=1$. Otherwise, \texttt{SPPS} samples branches $\bmo_j=1$ and $\bmo_j=-1$ with probabilities $q_j$ and $1-q_j$, respectively, where $q_j=\frac{|\cos \bmt_j|+a}{|\cos \bmt_j|+|\sin \bmt_j|+2a}$ .
Here, $a>0$ enforces a nonzero sampling probability for both branches, so that branches with small coefficients but non-negligible gradient contributions are not overlooked.

\textbf{Step (ii): Importance reweighting. } After Step (i), \texttt{SPPS} corrects the contribution of each sampled path by its corresponding sampling probability, as illustrated in Fig.~\ref{dpp_fig2_main_fig}(c). The motivation of this step is to convert the sampled propagation paths into an unbiased expectation value estimator, which will subsequently induce unbiased gradients. That is, after $B$ independent sampling rounds in Step (i), \texttt{SPPS} obtains a collection of propagation paths denoted by $\{\bmo^{(b)}\}_{b=1}^{B}$. For each sampled path $\bmo^{(b)}$, we define the calibrated contribution as
$\tilde{h}_{\bmo^{(b)}}(\bmt) = {h_{\bmo^{(b)}}(\bmt)} / {{\Pr}(\bmo^{(b)})}$ with $h_{\bmo^{(b)}}$ in Eq.~(\ref{dypp_f_loss_pp_eq}). Averaging these contributions gives an unbiased estimator $\tilde{f}(\bmt) = \frac{1}{B} \sum_{b=1}^{B} \tilde{h}_{\bmo^{(b)}} (\bmt)$, i.e.,
\begin{align}
\E\!\left[\tilde{f}(\bmt)\right]
&= \E_{\bmo \sim \Pr}\!\left[\tilde{h}_{\bmo}(\bmt)\right] = \sum_{\bmo \in \Omega}\frac{h_{\bmo}(\bmt)}{{\Pr}(\bmo)}{\Pr}(\bmo) = \sum_{\bmo \in \Omega}h_{\bmo}(\bmt) = f(\bmt) .
\label{spp_method_implementation_step2_tildef}
\end{align}
This unbiased expectation value estimator provides the basis for the gradient estimation in Step (iii).

\textbf{Step (iii): Gradient estimation.}
The final step is to calculate gradient estimations based on the sampled paths via a procedure that we named path automatic differentiation (PAD), as visualized in Fig.~\ref{dpp_fig2_main_fig}(d). Since each ${h}_{\bmo^{(b)}}(\bmt)$ in Eq.~(\ref{spp_method_implementation_step2_tildef}) is parameterized by $\bmt$ via the coefficient $\Psi_{\bmo^{(b)}}(\bmt)$, which is a product of trigonometric factors, the corresponding gradient admits the form
\begin{equation}
\tilde{\bm{g}}(\bmt) = \frac{1}{B} \sum\nolimits_{b=1}^{B} \bm{s}(\bmt,\bmo^{(b)}) \tilde{h}_{\bmo^{(b)}}(\bmt),
\label{eq_spp_single_grad}
\end{equation}
where $\bm{s}(\bmt,\bmo^{(b)}) \in \R^L$ and its $j$-th entry is $\cot(\bmt_j) \mathbbm{1}{[\bmo_j^{(b)}=-1]} - \tan(\bmt_j) \mathbbm{1}{[\bmo_j^{(b)}=1]}$. Similar to the case of Eq.~(\ref{spp_method_implementation_step2_tildef}), the estimator $\tilde{\bm g}(\bmt)$ in Eq.~(\ref{eq_spp_single_grad}) is an unbiased estimator of the gradient of $f(\bmt)$.

\noindent\textbf{Remark.}
The derivation of PAD is provided in App.~\ref{app_spp_alg}. 
PAD offers two benefits: (i)~by exploiting the trigonometric structure of $\Psi_{\bmo}(\bmt)$, it obtains all entries of each gradient sample simultaneously, avoiding differentiation through the full propagation tree or parameter-shift evaluations for each entry; (ii) independent PAD samples are naturally parallelizable, enabling acceleration on multi-core CPUs.

\subsection{Theoretical analysis of \texttt{SPPS}}
\label{dpp_method_theory}

We next analyze the theoretical properties of \texttt{SPPS}, including its unbiasedness and sample complexity of the gradient estimation, and the convergence analysis. All proofs are deferred to Apps.~\ref{app_spp_estimator} and \ref{app_spp_sgd}.

\begin{theorem} 
\label{thm_spp_estimator}
$\mathtt{SPPS}$ in Eq.~(\ref{eq_spp_single_grad})  provides unbiased gradient estimation, and the error is bounded by $\epsilon$ with high probability using $\tilde{\mathcal O}((1+2a)^L\kappa(\bmt)L/(a\epsilon^2))$ samples, where $\kappa(\bmt)=\prod_{j=1}^{L}\bigl(1+|\sin(2\bmt_j)|\bigr)$.
\end{theorem}

Theorem~\ref{thm_spp_estimator} indicates that the sample complexity of \texttt{SPPS} is controlled by the effective branching factor $\kappa(\bmt)$, the regularization parameter $a$, and the target accuracy. The bound is meaningful when $a>0$, which prevents derivative-sensitive branches from being overlooked. The factor $\kappa(\bmt)$ aggregates the parameter-induced amplification during propagation and characterizes the intrinsic difficulty of faithfully estimating the corresponding gradient. When $\kappa(\bmt)$ remains moderate, e.g., polynomial in the number of qubits and parameters, a polynomial sampling budget is sufficient. Conversely, rapidly growing $\kappa(\bmt)$ can make the sampling budget exponential in the worst case, which is consistent with the generic hardness to simulate quantum computation.

Since \texttt{SPPS} provides unbiased stochastic estimates of the original gradient, the induced GD dynamics fall within the standard stochastic-gradient framework. This leads to the following corollary.

\begin{corollary}
\label{cor_spp_sgd}
We consider the $\mathtt{SPPS}$-driven GD optimization of the objective $C(\bmt)$ in Eq.~(\ref{dypp_f_loss_eq}) with $K=1$ and $\ell_1(x)=x$. Then, for any $\epsilon>0$, we have $\min_{t\in[T]}\E\|\nabla_{\bmt}f(\bmt^{(t)})\|_2^2\le \epsilon$ within $T=\mathcal{O}(L^3\|O\|_2^4/\epsilon^2)$ iterations by using $\mathcal{O}(N_O L^3\|O\|_2^4\kappa_T/\epsilon^2)$ \texttt{SPPS} path samples, where $\kappa_T=\max_{t\in[T]}\kappa(\bmt^{(t)})$ and $N_O$ is the number of Pauli terms in the observable $O$.
\end{corollary}

To the best of our knowledge, Corollary~\ref{cor_spp_sgd} provides the first convergence guarantee for quantum optimization driven by a classical simulator beyond the exact regime. It shows that \texttt{SPPS}-driven GD finds an approximate stationary point with iteration complexity polynomial in $L$ and $\|O\|_2$. The overall efficiency is governed by the path sample complexity, which is primarily controlled by the trajectory-dependent factor $\kappa_T$. Moderate $\kappa_T$ yields efficient classical simulation, whereas rapidly growing $\kappa_T$ in the worst case reflects the hardness of simulating generic quantum optimization.

\section{Experiments}
\label{spp_experiment}

We conduct systematic experiments on three quantum optimization tasks: pre-training VQEs~\cite{Peruzzo2014,PhysRevResearch.5.043217,khan2023pre} and QML models~\cite{bermejo2024quantum,Benedetti_2019,verdon2019learning}, and preparing quantum encoding circuits~\cite{PhysRevResearch.4.023136,PhysRevA.109.052423}. These tasks have broad downstream applications~\cite{bauer2020quantum,schuld2021effect} and form suitable benchmarks for evaluating the performance of \texttt{SPPS} across different objective functions and observables. 
Here, pre-training uses classical computation to find initial parameters close to high-quality solutions, which can reduce the quantum resource cost of subsequent training on quantum computers~\cite{rudolph2023synergistic,Egger2021warmstartingquantum,electronics12020347,liu2023mitigating,martin2026pre}.
Further implementation details and additional results on these tasks are provided in Apps.~\ref{app_spp_alg} and \ref{app_exp}, respectively.

\begin{figure*}[t]
    \centering
    \includegraphics[width=0.99\linewidth]{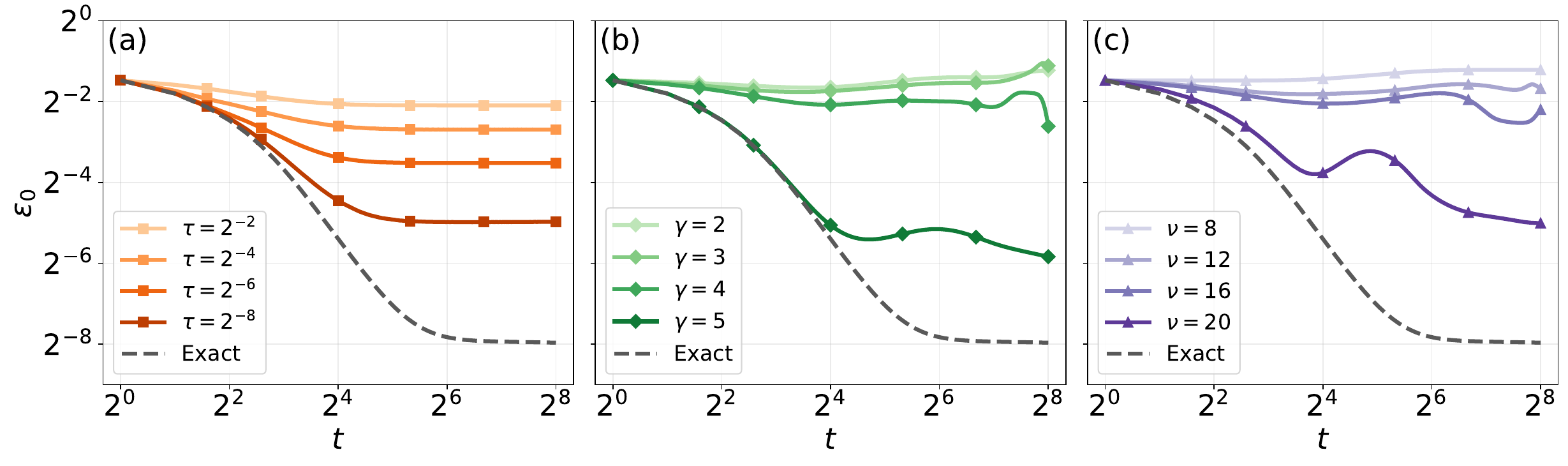}
    \caption{ 
{\textbf{Biased optimization trajectories of \texttt{Tb-PBS} methods.}
Subplots (a), (b), and (c) show the energy error $\epsilon_0$ versus the optimization step $t$ of pre-training VQE for the $15$-qubit TFIM, which are simulated by using \texttt{CT}-, \texttt{WT}-, and \texttt{FT-PBS} with different truncation thresholds $\tau$, $\gamma$, and $\nu$, respectively. The physical meanings of these thresholds are given in Sec.~\ref{dpp_method_biased_grad}. The dashed curve denotes the exact GD baseline obtained using \texttt{PennyLane}. 
}}
    \label{dpp_fig2}
\end{figure*}

\textbf{Pre-training VQE.} 
We benchmark different classical simulators on the pre-training of VQE, where the goal is to find a set of parameters in Eq.~(\ref{dypp_f_loss_eq}) with $K=1$, $\ell_1(x)=x$, and $O=H$ that approximate the ground-state energy of $H$ through classical simulation. Specifically, we consider the one-dimensional transverse-field Ising model (TFIM) with the Hamiltonian
$H=-J\sum_{i=1}^{n-1} Z_i Z_{i+1}-g\sum_{i=1}^{n} X_i$,
which is a prototypical quantum many-body model and a standard benchmark for VQE. We evaluate $g\in\{0.6,0.8,1.0,1.2,1.4\}$ and $J=1.0$ with system sizes up to $100$ qubits, covering both non-critical and critical regimes. The variational ansatz is a $5$-layer circuit consisting of an initial Hadamard layer, repeated $R_Y$ and $R_Z$ rotation layers, and nearest-neighbor CNOT chains.

\textbf{Pre-training QNN.}
We benchmark different simulators for supervised quantum machine learning by pre-training QNNs. In particular, the dataset is $\{(|\psi_i\rangle,y_i)\}_{i=1}^{200}$, where each input $|\psi_i\rangle\in\{|0\rangle,|1\rangle\}^{\otimes n}$ is a computational-basis state with $n=40$ qubits. The label is generated by applying a quantum circuit $V$ to the input and measuring the observable $O=\sum_{j=1}^{n}Z_j/n$, i.e.,
$y_i=\langle\psi_i|V^\dag O V|\psi_i\rangle$. In our experiments, we use $160$ training samples and $40$ test samples. A $4$-layer hardware-efficient VQC $U(\bmt)$ is trained to learn these labels by minimizing the mean-squared error (MSE) loss on the training set, i.e. ${\rm MSE}(\bmt)=\frac{1}{160}\sum_{i=1}^{160}(z_i(\bmt)-y_i)^2$,
where $z_i(\bmt)=\langle\psi_i|U(\bmt)^\dag O U(\bmt)|\psi_i\rangle$.

\subsection{Experimental Settings}

\textbf{Reference methods.}
We compare \texttt{SPPS} with \texttt{CT}-, \texttt{FT}-, and \texttt{WT-PBS} methods introduced in Sec.~\ref{dpp_method_biased_grad}, implemented using \texttt{PauliPropagation.jl}~\cite{rudolph2025pauli}. 
For small-scale VQE tasks with $n\leq 15$, we include a vanilla VQE baseline with exact GD, implemented using \texttt{PennyLane}~\cite{bergholm2018pennylane}, and compute the reference ground-state energy $E_0$ by exact diagonalization. 
For larger systems, we approximate $E_0$ using the thermodynamic-limit expression~\cite{PFEUTY197079}.

\textbf{Evaluation metrics.}
For VQE tasks, we evaluate the simulation quality by using the energy error normalized by the qubit number, $\epsilon_0=|E-E_0|/n$, where $E$ is the pre-trained VQE energy and $E_0$ is the ground-state energy. For small systems with $n\leq 15$, $E$ is obtained exactly by state-vector simulation. For larger systems, $E$ is estimated using \texttt{CT-PBS} with $\tau=2^{-12}$. For QML tasks, we report the MSE$(\cdot)$ on the training set and the $R^2$ score on the test set after training, where $R^2=1-\sum_i(z_i-y_i)^2/\sum_i(y_i-\bar y)^2$. Here, $\bar y$ is the mean target label on the test set, and each prediction $z_i$ is estimated using \texttt{CT-PBS} with $\tau=2^{-12}$. For all tasks, we record the optimization runtime to evaluate the efficiency of different classical simulators.

\textbf{Hyperparameter settings.}
\texttt{SPPS} employs a proxy to estimate the gradient error in practice and to allocate sample budgets adaptively. 
We compare this proxy with a prescribed threshold $\delta\in\{2^{-1},2^{-2},2^{-3}\}$, where a smaller $\delta$ induces larger sample budgets and yields more accurate gradients. For \texttt{Tb-PBS} baselines, we tune the \texttt{CT} threshold $\tau$, \texttt{FT} threshold $\nu$, and \texttt{WT} threshold $\gamma$ according to the task settings. All classical simulation methods share the same optimization settings. Specifically, for pre-training VQEs, we use GD with learning rate $\eta=0.05$ for $T=256$ optimization steps by default. For pre-training QNNs, we use mini-batch GD with learning rate $\eta=5.0$ and batch size $16$ for $T=2000$ optimization steps.
Further implementation details, including the formulation of the gradient-error proxy and the parameter $a$ in the sampling distribution, are provided in App.~\ref{app_spp_alg}.

\subsection{Experimental results}

\begin{figure*}[t]
    \centering
    \includegraphics[width=0.99\linewidth]{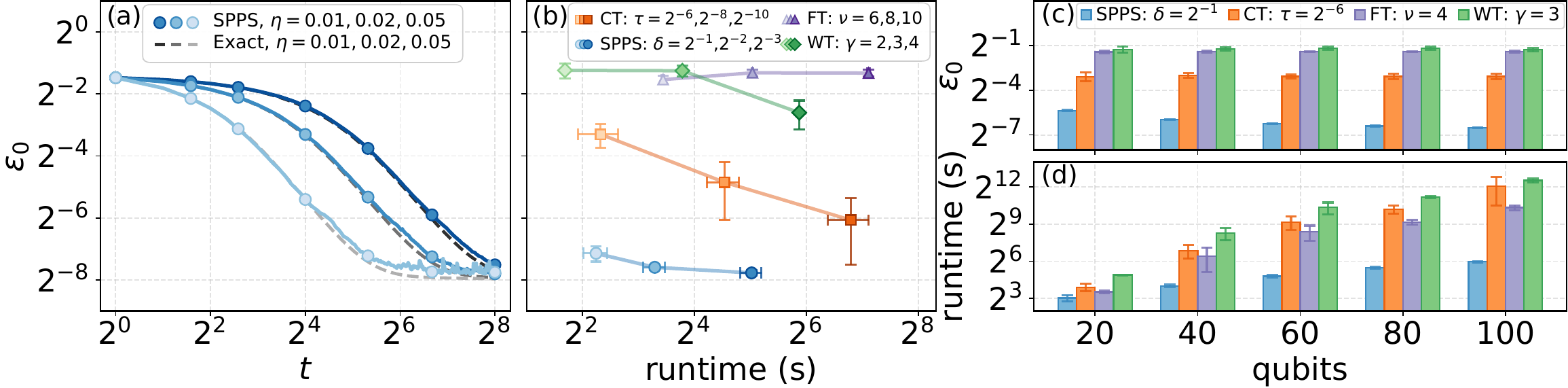}
\caption{
{
\textbf{Performance of different \texttt{PBS} methods for pre-training VQE on the TFIM.}
(a) Energy error $\epsilon_0$ versus optimization step $t$ for \texttt{SPPS}-driven GD and exact GD on the $15$-qubit TFIM with different learning rates.
(b) Final energy error $\epsilon_0$ versus optimization runtime for \texttt{SPPS} and three \texttt{Tb-PBS} methods on the $15$-qubit TFIM with different thresholds.
(c,d) Final energy error and optimization runtime for system sizes from $20$ to $100$ qubits.
Error bars indicate the standard deviation over $10$ independent runs.
} }
    \label{dpp_fig3}
\end{figure*}

\textbf{\texttt{Tb-PBS} leads to biased optimization trajectories.}
We first examine how truncations in \texttt{Tb-PBS}s affect the simulation of quantum optimization. To this end, we pre-train VQE for the $15$-qubit TFIM with $g=1.0$ using \texttt{CT}-, \texttt{WT}-, and \texttt{FT-PBS} under different truncation thresholds, and compare the resulting optimization trajectories with the exact GD trajectory, as shown in Fig.~\ref{dpp_fig2}. Although using less aggressive truncation generally improves the energy accuracy, all three truncation rules can still deviate substantially from the exact trajectory. For example, the optimization driven by \texttt{CT-PBS} with $\tau=2^{-6}$ plateaus above $\epsilon_0=2^{-4}$, while exact GD optimization continues decreasing to around $2^{-8}$. The deviations are more pronounced for \texttt{WT}- and \texttt{FT-PBS}, where the optimization can drive non-monotone energy trajectories and lead to non-converging $\epsilon_0$ far above the exact GD baseline. These deviations empirically confirm that  \texttt{Tb-PBS} can induce structural gradient bias, which distorts the optimization trajectory.

\textbf{\texttt{SPPS} faithfully simulates pre-training VQE trajectories.}
We next evaluate the reliability of \texttt{SPPS} in simulating quantum optimization. Specifically, we pre-train VQE for the $15$-qubit TFIM with $g=1.0$ using GD driven by \texttt{SPPS} gradients and exact gradients with learning rates $\eta\in\{0.01,0.02,0.05\}$ and record the energy error $\epsilon_0$ during the optimization, with the results illustrated in Fig.~\ref{dpp_fig3}(a). Across all learning rates, \texttt{SPPS} closely tracks the corresponding exact-GD baseline. For example, at the largest learning rate $\eta=0.05$, \texttt{SPPS} reduces $\epsilon_0$ from above $2^{-2}$ at initialization to approximately $2^{-8}$ within about $64$ steps, matching the convergence behavior of exact GD. For smaller learning rates, both methods converge more slowly, while the \texttt{SPPS} curves remain aligned with exact GD. These consistent trajectories demonstrate that \texttt{SPPS} provides sufficiently accurate stochastic gradients to pre-train VQE faithfully.

\textbf{\texttt{SPPS} outperforms all reference \texttt{Tb-PBS} methods with superior accuracy-runtime trade-offs.}
We further compare the performance of \texttt{SPPS} with reference \texttt{Tb-PBS} methods, measured by the final energy error $\epsilon_0$ after optimization and the optimization runtime. In particular, on the task of pre-training VQE for the $15$-qubit TFIM with $g=1.0$, we sweep the gradient-error threshold $\delta\in\{2^{-1},2^{-2},2^{-3}\}$ for \texttt{SPPS}, and the truncation thresholds $\tau\in\{2^{-6},2^{-8},2^{-10}\}$, $\nu\in\{6,8,10\}$, and $\gamma\in\{2,3,4\}$ for \texttt{CT}-, \texttt{FT}-, and \texttt{WT-PBS}, respectively. The resulting accuracy-runtime trade-offs are summarized in Fig.~\ref{dpp_fig3}(b). \texttt{SPPS} achieves lower energy errors than all reference \texttt{PBS} methods under comparable or smaller runtime budgets. As $\delta$ is tightened from $2^{-1}$ to $2^{-3}$, \texttt{SPPS} progressively reduces the energy error from around $2^{-7}$ to $2^{-8}$, while keeping the runtime from a few seconds to less than one minute. By contrast, \texttt{Tb-PBS}s incur much larger runtime while still remaining less accurate than \texttt{SPPS}. For example, the most accurate \texttt{CT} setting with $\tau=2^{-10}$ takes roughly $2^7$--$2^8$ seconds but only reaches an error on the order of $2^{-6}$.

\textbf{Scalability of \texttt{SPPS}.}
We then test the scalability of \texttt{SPPS} by increasing the TFIM system size from $20$ to $100$ qubits with $g=1.0$, and record the final energy error and optimization runtime in Figs.~\ref{dpp_fig3}(c) and~\ref{dpp_fig3}(d), respectively. Across this range, the error of \texttt{SPPS} stays between $2^{-5}$ and $2^{-7}$ and decreases as the system size grows. This trend is consistent with using the thermodynamic-limit expression as the reference ground-state energy, which introduces a finite-size discrepancy that decreases for larger systems. In contrast, \texttt{Tb-PBS} methods remain substantially less accurate, with errors higher than $2^{-4}$. The runtime advantage of \texttt{SPPS} also becomes more pronounced for larger systems. For example, \texttt{SPPS} is about twice as fast as \texttt{CT-PBS} at $20$ qubits, while the speedup increases to more than eightfold at $100$ qubits. Overall, \texttt{SPPS} shows more favorable scaling behaviors than existing \texttt{Tb-PBS} methods, in terms of both accuracy and runtime.

\begin{table*}[t]
\caption{ { 
{Energy error $\epsilon_0$ ($\downarrow$) and runtime ($\downarrow$) of different \texttt{PBS} methods on the pre-training of VQE for the $15$-qubit TFIM. We compare \texttt{SPPS} with $\delta=2^{-1}$ against reference \texttt{PBS} methods based on \texttt{CT} ($\tau=2^{-8}$), \texttt{FT} ($\nu=10$), and \texttt{WT} ($\gamma=4$), under different field strengths $g$. Each entry reports the error on the first line and the runtime on the second line with the mean and standard deviation over 10 independent runs.
The \blue{best} and \orange{second-best} results are highlighted in \blue{blue} and \orange{orange}, respectively.
 }} }
\label{tab_relerr_vs_g}
\centering
\begingroup
\fontsize{9pt}{10.5pt}\selectfont
\setlength{\tabcolsep}{4pt}
\renewcommand{\arraystretch}{1.12}
\newcommand{\errtime}[2]{\makecell[c]{#1\\#2}}
\begin{tabular}{l|ccccc}
\toprule
 & $g=0.6$ & $g=0.8$ & $g=1.0$ & $g=1.2$ & $g=1.4$ \\
\midrule
\texttt{CT-PBS} 
& \errtime{\orange{$0.077 \pm 0.035$}}{\orange{$(48.4 \pm 28.8)$ s}}
& \errtime{\orange{$0.051 \pm 0.025$}}{\orange{$(34.3 \pm 10.7)$ s}}
& \errtime{\orange{$0.035 \pm 0.020$}}{\orange{$(23.2 \pm 4.5)$ s}}
& \errtime{\orange{$0.024 \pm 0.015$}}{\orange{$(18.3 \pm 4.8)$ s}}
& \errtime{\orange{$0.021 \pm 0.009$}}{\orange{$(15.9 \pm 2.2)$ s}} \\
\midrule
\texttt{FT-PBS}
& \errtime{$0.401 \pm 0.027$}{$(150.1 \pm 4.6)$ s}
& \errtime{$0.383 \pm 0.028$}{$(153.9 \pm 4.9)$ s}
& \errtime{$0.399 \pm 0.031$}{$(138.5 \pm 2.1)$ s}
& \errtime{$0.412 \pm 0.058$}{$(136.6 \pm 0.6)$ s}
& \errtime{$0.438 \pm 0.063$}{$(143.4 \pm 4.0)$ s} \\
\midrule
\texttt{WT-PBS}
& \errtime{$0.298 \pm 0.068$}{$(60.0 \pm 1.2)$ s}
& \errtime{$0.213 \pm 0.055$}{$(61.0 \pm 1.7)$ s}
& \errtime{$0.165 \pm 0.051$}{$(58.7 \pm 0.7)$ s}
& \errtime{$0.200 \pm 0.062$}{$(57.8 \pm 0.4)$ s}
& \errtime{$0.250 \pm 0.108$}{$(59.7 \pm 1.0)$ s} \\
\midrule
\texttt{SPPS} (ours)
& \errtime{\blue{$0.004 \pm 0.001$}}{\blue{$(10.0 \pm 3.7)$ s}}
& \errtime{\blue{$0.011 \pm 0.001$}}{\blue{$(6.3 \pm 1.5)$ s}}
& \errtime{\blue{$0.007 \pm 0.001$}}{\blue{$(4.7 \pm 0.7)$ s}}
& \errtime{\blue{$0.004 \pm 0.001$}}{\blue{$(3.8 \pm 0.8)$ s}}
& \errtime{\blue{$0.004 \pm 0.001$}}{\blue{$(3.6 \pm 0.7)$ s}} \\
\bottomrule
\end{tabular}
\endgroup
\end{table*}

\textbf{Larger gains of \texttt{SPPS} on harder optimization instances.}
We further evaluate the performance of \texttt{SPPS} on harder quantum optimization instances. Here, we pre-train the VQE on the $15$-qubit TFIM with $g\in\{0.6,0.8,1.0,1.2,1.4\}$ and compare \texttt{SPPS} with \texttt{CT}-, \texttt{FT}-, and \texttt{WT-PBS} methods, with results summarized in Tab.~\ref{tab_relerr_vs_g}. Across all field strengths, \texttt{SPPS} achieves the lowest energy error and shortest runtime. The advantage becomes larger when \texttt{Tb-PBS} methods incur stronger optimization bias. For example, moving from $g=1.0$ to $g=0.6$ roughly doubles the runtime of \texttt{SPPS} and the best \texttt{Tb-PBS}, i.e., \texttt{CT-PBS}. However, the error of \texttt{SPPS} decreases from $0.007$ to $0.004$, whereas the error of \texttt{CT-PBS} doubles. As a result, the error reduction of \texttt{SPPS} over \texttt{CT-PBS} increases from about $5\times$ to about $20\times$, which shows that \texttt{SPPS} provides larger benefits than \texttt{Tb-PBS} methods on harder cases.

\begin{figure}[t]
  \centering
  \includegraphics[width=0.6\linewidth]{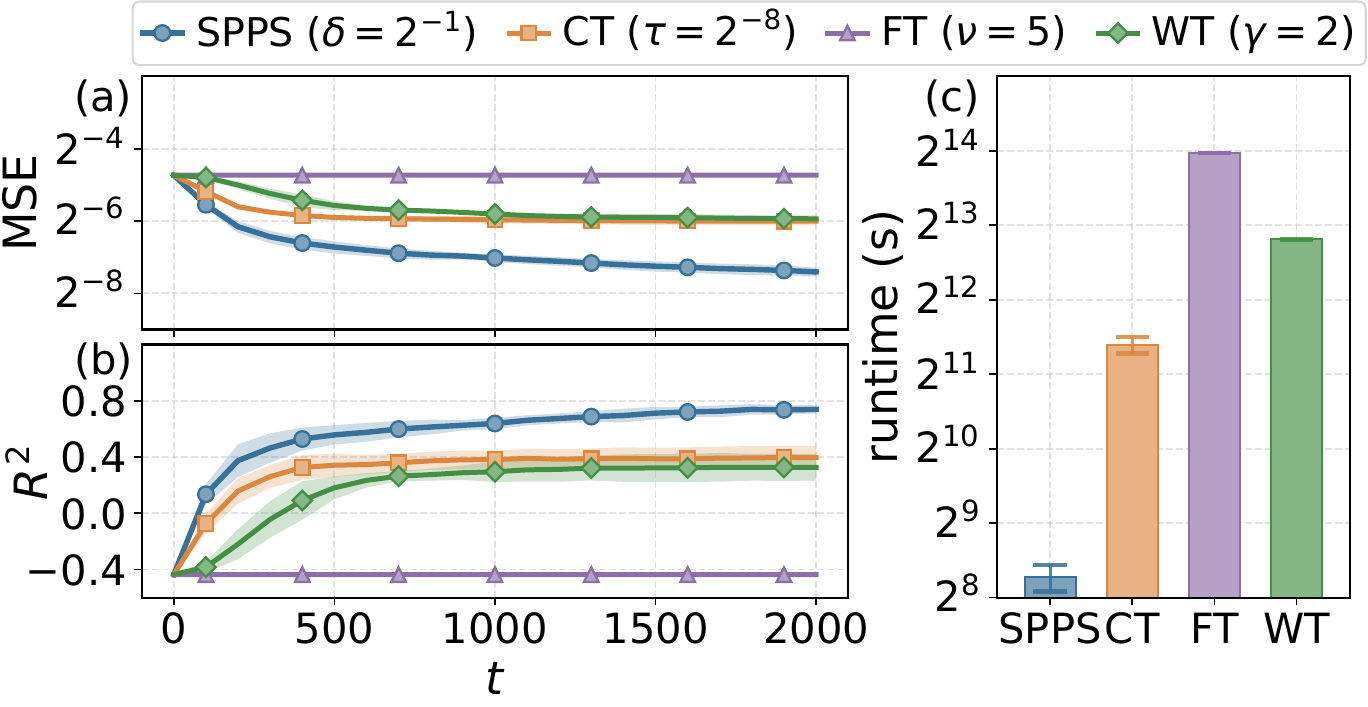}
  \caption{
{
\textbf{Performance of different \texttt{PBS} methods for pre-training QNN on the $40$-QML task.}
(a,b) Training MSE and test $R^2$ score versus optimization step $t$ for \texttt{SPPS}, \texttt{CT}-, \texttt{FT}-, and \texttt{WT-PBS}, respectively. (c) Optimization runtime of the four methods. We use $\delta=2^{-1}$ for \texttt{SPPS}, $\tau=2^{-8}$ for \texttt{CT}, $\nu=5$ for \texttt{FT}, and $\gamma=2$ for \texttt{WT}. Error bars show the standard deviation over $5$ independent runs.}
}
  \label{dpp_fig4}
\end{figure}
\textbf{\texttt{SPPS} achieves better convergence points in pre-training QNN.}
We finally evaluate the applicability of \texttt{SPPS} to QML by pre-training QNNs on the $40$-qubit supervised learning task introduced above, where computational-basis input states are labeled by a given quantum circuit. 
We compare \texttt{SPPS} with reference \texttt{Tb-PBS} methods, and record the training MSE, test $R^2$ score, and optimization runtime, with the results summarized in Fig.~\ref{dpp_fig4}. \texttt{SPPS} converges to a substantially lower training loss and a higher test score than all \texttt{Tb-PBS} methods. In particular, \texttt{SPPS} reduces the training MSE from around $2^{-5}$ to below $2^{-7}$ and reaches a test $R^2$ score around $0.8$. By contrast, \texttt{CT} and \texttt{WT} saturate around the $2^{-6}$ MSE level with test $R^2$ scores around $0.4$, while \texttt{FT} fails to obtain a positive $R^2$ score. The runtime comparison further shows that \texttt{SPPS} is the most efficient method, requiring less than $2^9$ seconds, compared with runtimes above $2^{11}$ seconds for \texttt{Tb-PBS} methods. These comparisons demonstrate that the advantages of \texttt{SPPS} extend to QML tasks, yielding better convergence points at substantially lower simulation cost than \texttt{Tb-PBS} methods.

\section{Conclusion}

In this work, we studied Pauli-propagation simulation for quantum optimization and identified a key limitation of \texttt{Tb-PBS} methods: accurate expectation estimation does not necessarily imply faithful gradient estimation. To address this issue, we proposed \texttt{SPPS}, a stochastic Pauli-path simulator that samples from the full propagation path space and combines importance reweighting with path automatic differentiation to obtain unbiased stochastic gradients. We established gradient accuracy and convergence guarantees for \texttt{SPPS}-driven optimization, showing that faithful classical simulation is possible when the trajectory-dependent path complexity remains moderate. Extensive experiments on pre-training VQE and QNN show that \texttt{SPPS} tracks exact optimization dynamics, improves accuracy-runtime trade-offs, and scales to VQE simulations with up to $100$ qubits. These results provide a practical route to reliable classical simulation for large-scale quantum optimization.

\bibliography{ref}
\bibliographystyle{unsrt}

\clearpage
\onecolumngrid
\appendix

\tableofcontents

\section{More preliminaries and related work}

\label{spp_app_pre}

In this appendix, we first introduce notations used in this work and the basics of quantum computing. Subsequently, we briefly outline the implementation of 
the \texttt{PBS} method. Finally, we provide a brief review of the literature on classical simulation of quantum circuits using \texttt{PBS} method.

\subsection{Notations}
\label{app_notations}

We summarize the notation used throughout this work.
For a positive integer $N$, we denote $[N]=\{1,\ldots,N\}$.
Vectors are written in bold font, for example, $\bm{a}_j$ denotes the $j$-th component of a vector $\bm{a}$.
The tensor product is denoted by $\otimes$, the conjugate transpose of $A$ by $A^\dagger$, and the trace by $\Tr[A]$.
We use $\E[\cdot]$ and $\Var[\cdot]$ for expectation and variance.
The notation $\O(\cdot)$ and $\tilde{\O}(\cdot)$ is used for asymptotic complexity, with $\tilde{\O}(\cdot)$ suppressing logarithmic factors.
For vectors, $\|\cdot\|_2$ denotes the Euclidean norm.
For matrices or observables, $\|\cdot\|_2$ denotes the spectral norm, i.e., the largest singular value.

\subsection{Basics of quantum computing}

\textbf{Basics of quantum computation.} The elementary unit of quantum computation is the qubit (or quantum bit), which is the quantum mechanical analog of a classical bit. A qubit is a two-level quantum-mechanical system described by a unit vector in the Hilbert space $\mathbb{C}^2$. In Dirac notation, a qubit state is defined as $\ket{\phi}=c_0\ket{0}+c_1\ket{1}\in \mathbb{C}^2$ where $\ket{0}=[1,0]^{\top}$ and $\ket{1}=[0,1]^T$ specify two unit bases and the coefficients $c_0,c_1\in\mathbb{C}$ yield $|c_0|^2+|c_1|^2=1$. Similarly, the \textit{quantum state} of $n$ qubits is defined as a unit vector in $\mathbb{C}^{2^n}$, i.e., $\ket{\psi}=\sum_{j=1}^{2^n}c_j\ket{e_j}$, where $\ket{e_j}\in \mathbb{R}^{2^n}$ is the computational basis whose $j$-th entry is $1$ and other entries are $0$, and $\sum_{j=1}^{2^n}|c_j|^2=1$ with $c_j \in \mathbb{C}$. Besides Dirac notation, the density matrix can be used to describe more general qubit states. For example, the density matrix of the state $\ket{\psi}$ is $\rho=\ket{\psi}\bra{\psi} \in \mathbb{C}^{2^n \times 2^n}$, where $\bra{\psi}=\ket{\psi}^{\dagger}$ refers to the complex conjugate transpose of $\ket{\psi}$. For a set of qubit states $\{p_j, \ket{\psi_j}\}_{j=1}^m$ with $p_j>0$, $\sum_{j=1}^m p_j=1$, and $\ket{\psi_j}\in \mathbb{C}^{2^n}$ for $j \in [m]$, its density matrix is $\rho=\sum_{j=1}^m p_j \rho_j$ with $\rho_j=\ket{\psi_j}\bra{\psi_j}$ and $\Tr(\rho)=1$.
	
	A \textit{quantum gate} is a unitary operator that can evolve a quantum state $\rho$ to another quantum state $\rho^{\prime}$. Namely, an $n$-qubit gate $U\in\mathcal{U}({2^n})$ obeys $UU^{\dagger}=U^{\dagger}U=I_{2^n}$, where $\mathcal{U}({2^n})$ refers to the unitary group in  dimension $2^n$. Typical single-qubit quantum gates include the Pauli gates, which can be written as Pauli matrices:
		\begin{equation}
			X = \left[ \begin{array}{ccc}
				0 & 1 \\
				1 & 0 \\
			\end{array}
			\right], \quad 
			Y = \left[ \begin{array}{ccc}
				0 & -\imath \\
				\imath & 0 \\
			\end{array}
			\right], \quad 
			Z = \left[ \begin{array}{ccc}
				1 & 0 \\
				0 & -1 \\
			\end{array}
			\right]. \quad  \label{eq_pauli}
		\end{equation}
		The more general quantum gates are their corresponding rotation gates $R_X(\theta)=e^{-\imath \frac{\theta}{2}X}, R_Y(\theta)=e^{-\imath \frac{\theta}{2}Y}$, and $R_Z(\theta)=e^{-\imath \frac{\theta}{2}Z}$ with a tunable parameter $\theta$, which can be written in the matrix form as
		\begin{equation}
			R_X(\theta)=\left[\begin{array}{cc}
				\cos \frac{\theta}{2} & -\imath \sin \frac{\theta}{2} \\
				-\imath \sin \frac{\theta}{2} & \cos \frac{\theta}{2}
			\end{array}\right], 
			R_Y(\theta)=\left[\begin{array}{cc}
				\cos \frac{\theta}{2} & -\sin \frac{\theta}{2} \\
				\sin \frac{\theta}{2} & \cos \frac{\theta}{2}
			\end{array}\right],  
			R_Z(\theta)=\left[\begin{array}{cc}
				e^{-\imath \frac{\theta}{2}} & 0 \\
				0 & e^{\imath \frac{\theta}{2}}
			\end{array}\right]. \label{eq_pauli_rot}
		\end{equation}
		They are equivalent to rotating a tunable angle $\theta$ around $x$, $y$, and $z$ axes of the Bloch sphere, and recovering the Pauli gates $X$, $Y$, and $Z$ when $\theta=\pi$. Moreover, a multi-qubit gate can be either an individual gate (e.g., CNOT gate) or a tensor product of multiple single-qubit gates.  
	
	The \textit{quantum measurement} refers to the procedure of extracting classical information from the quantum state. It is mathematically specified by a Hermitian matrix $H$ called the \textit{observable}. Applying the observable $H$ to the quantum state $\ket{\psi}$ yields a random variable whose expectation value is $\bra{\psi}H\ket{\psi}$. 

\textbf{Hamiltonian and ground state}. 
	In quantum computation, a \textit{Hamiltonian} is a Hermitian matrix that is used to characterize the evolution of a quantum system or as an observable to extract the classical information from the quantum system. Specifically, under the Schr\"odinger equation, a quantum gate has the mathematical form of $U=e^{-itH}$, where $H$ is a Hermitian matrix, called the Hamiltonian of the quantum system, and $t$ refers to the evolution time of the Hamiltonian. Typical single-qubit Hamiltonians include the Pauli matrices defined in Eq.~(\ref{eq_pauli}). As a result,  the evolution time $t$ refers to the tunable parameter $\theta$ in Eq.~(\ref{eq_pauli_rot}). Any single-qubit Hamiltonian can be decomposed as the linear combination of Pauli matrices, i.e., $H=a_1I+a_2X+a_3Y+a_4Z$ with $a_j \in \mathbb{C}$. In the same way, a multi-qubit Hamiltonian is denoted by $H=\sum_{j=1}^{4^n}a_jP_j$, where $P_j\in\{I,X,Y,Z\}^{\otimes n}$ is the tensor product of Pauli matrices. In quantum chemistry and quantum many-body physics, the Hermitian matrix that describes the quantum system to be solved is denoted as the \textit{problem Hamiltonian} $H_C$. 
	
	When taking the problem Hamiltonian as the observable, the quantum state $\ket{\psi^*}$ is said to be the \textit{ground state} of problem Hamiltonian $H$ if the expectation value $\bra{\psi^*}H\ket{\psi^*}$ takes the minimum eigenvalue of $H$, which is called the \textit{ground energy}. The ground states encode much essential information about the problem Hamiltonian, such as the critical behavior of quantum many-body systems, or the optimal solution of an optimization problem related to the problem Hamiltonian.

\subsection{Pauli path propagation in the Heisenberg picture}
\label{app_PP_notion}

Let $\mathcal{P}_n= \{sP: s \in \{\pm 1\}, P \in \{I,X,Y,Z\}^{\otimes n}\}$ be the set of $n$-qubit Pauli strings with an overall sign $\pm 1$. Then, any Hermitian observable admits a Pauli expansion
\begin{equation}
O=\sum_{k=1}^{N_O} c_k O_k,\qquad O_k\in\mathcal{P}_n,\ c_k >0.
\end{equation}
We consider a parameterized circuit consisting solely of Pauli rotations $\{R_{G_j}\}_{j=1}^{L}$ and Clifford gates $\{V_{j}\}_{j=1}^{L}$:
\begin{equation}\label{app_PP_U_bmt_eq}
U(\bmt)=\prod_{j=1}^{L} R_{G_j}(\bmt_j) V_j,
\qquad
R_{G_j}(\bmt_j):=\exp\left(-\imath \bmt_j G_j/2\right),\ \ G_j\in\mathcal{P}_n .
\end{equation}
For an input state $\rho$, the expectation value is
\begin{equation}\label{dpp_eq_f_bmt}
f(\bmt)
=\Tr \left[ O U(\bmt)\rho U(\bmt)^\dagger \right]
=\sum_{k=1}^{N_O} c_k \Tr \left[\rho O_{k,L}(\bmt)\right],
\end{equation}
where we define the Heisenberg evolution of each Pauli term by the recursion
\begin{equation}\label{app_PP_PkL_eq}
O_{k,0}:=O_k,\qquad
O_{k,j}(\bmt)
:= V_j^\dagger R_{G_j}(\bmt_j)^\dagger O_{k,j-1}(\bmt) R_{G_j}(\bmt_j) V_j,
\qquad j=1,\dots,L.
\end{equation}

Next, we derive the update rule in Eq.~(\ref{app_PP_PkL_eq}). Since conjugation is linear, it suffices to describe the update of a single Pauli operator $P$ appearing in the expansion of $O_{k,j-1}(\bmt)$ under the $j$-th evolution.
A Pauli operator either commutes or anti-commutes with the generator $G_j$. Thus, the conjugation takes the form
\begin{equation}\label{dpp_eq_rot-map-only}
V_j^\dag R_{G_j}(\bmt_j)^\dagger P R_{G_j}(\bmt_j) V_j =
\begin{cases}
V_j^\dag P V_j , & [P,G_j]=0,\\
\cos(\bmt_j) V_j^\dag P V_j + \sin(\bmt_j)  (\imath V_j^\dag G_j P V_j), & \{P,G_j\}=0,
\end{cases}
\end{equation}
where $V_j^\dag P V_j$ and $\imath V_j^\dag G_j P V_j$ are Pauli strings with a sign $\pm 1$. Applying Eq.~(\ref{dpp_eq_rot-map-only}) repeatedly to Eq.~(\ref{app_PP_PkL_eq}) for $j=1,\cdots,L$ yields that, for every $k$, the Heisenberg-evolved operator admits an explicit pattern formulation
\begin{equation}\label{dpp_eq_Pk-expansion-pattern}
O_{k,L}(\bmt)
=\sum_{\bm{\omega}\in\Omega_{k}}
P_{\bm{\omega}} (O_k) \Psi_{\bm{\omega}}(\bmt)
\end{equation}
with 
\begin{align}
\Psi_{\bmo}(\bmt) ={}& \prod_{j=1}^{L} (\cos \bmt_j)^{\mathbb{I}[\bmo_j=1]} (\sin \bmt_j)^{\mathbb{I}[\bmo_j=-1]}, \label{dypp_pp_eq_psi_app} \\
P_{\bmo}(O_k) ={}& \left( \prod_{j=L}^{1} V_j^\dag (\imath G_j)^{\mathbb{I}[\bmo_j=-1]}  \right) O_k \left( \prod_{j=1}^{L} V_j \right), \label{dypp_pp_eq_Pomega_app}
\end{align}
where the pattern $\bm{\omega}=(\bm{\omega}_1,\dots,\bm{\omega}_L)\in\{0,\pm1\}^{L}$ records the propagation mode. Concretely, for each layer $j\in\{1,\dots,L\}$,
\begin{equation}
\bm{\omega}_j =
\begin{cases}
0, & \text{if the current Pauli string commutes with } G_j,\\
+1, & \text{if it anti-commutes with } G_j \text{ and the }\cos\text{-branch is chosen},\\
-1, & \text{if it anti-commutes with } G_j \text{ and the }\sin\text{-branch is chosen}.
\end{cases}
\end{equation}
Here $\Omega_{k}\subseteq\{0,\pm1\}^L$ denotes the set of legal patterns starting from $O_{k,0}=O_k$, i.e., those that are consistent with whether each intermediate Pauli operator commutes or anti-commutes with the corresponding generator. 
Applying Eq.~(\ref{dpp_eq_Pk-expansion-pattern}) in Eq.~(\ref{dpp_eq_f_bmt}) yields the Pauli expansion of the expectation value:
\begin{equation}\label{dpp_eq_app_f_bmt_in_pp}
f(\bmt) =\sum_{k=1}^{N_O} \sum_{\bm{\omega}\in\Omega_{k}} c_k \Psi_{\bm{\omega}}(\bmt) \Tr \left[ P_{\bm{\omega}} (O_k) \rho \right] = \sum_{k=1}^{N_O} \sum_{\bm{\omega}\in\Omega_{k}} c_k f_{\bmo} (\bmt; O_k) ,
\end{equation}
where we denote 
\begin{equation}\label{dpp_eq_f_bmo}
f_{\bmo} (\bmt; O_k) := \Psi_{\bm{\omega}}(\bmt) \Tr \left[ P_{\bm{\omega}} (O_k) \rho \right] .
\end{equation}

\subsection{Related work}

We review prior \texttt{PBS} methods that are most relevant to truncation-based \texttt{PBS} (\texttt{Tb-PBS}). \texttt{PBS} provides a Heisenberg-picture framework for estimating expectation values by expanding the evolved observable in the Pauli basis~\cite{PhysRevA.99.062337,PRXQuantum.6.020302,10.1063/5.0269149,Fontana2025,GonzalezGarcia2025paulipath,PhysRevLett.133.120603,Lin2026utilityscalequantum}. Since the number of propagation paths can grow exponentially with the number of non-Clifford gates, practical simulators typically control the expansion by discarding a subset of paths or Pauli strings during propagation. Existing \texttt{Tb-PBS} methods can be broadly categorized according to their truncation criterion: frequency truncation (\texttt{FT-PBS}), coefficient truncation (\texttt{CT-PBS}), and weight truncation (\texttt{WT-PBS}).

\textbf{\texttt{FT-PBS} method.}
One line of work truncates the Pauli-path expansion by restricting the Fourier level of the propagated observable, equivalently retaining only terms with a bounded number of nonzero sine or cosine factors~\cite{10.1063/5.0269149,Fontana2025,rudolph2023classical,PhysRevA.108.032406,lerch2024efficient,Du2025,liao2025sample,Liao2026}. Such approaches are well motivated in near-Clifford or noisy regimes, where high-frequency components are suppressed and low-frequency surrogates can approximate expectation values with controlled error~\cite{10.1063/5.0269149,Fontana2025}. LOWESA-type methods, for example, exploit noise-induced damping to obtain efficient classical surrogates for noisy variational circuits~\cite{Fontana2025,rudolph2023classical}.
Learning-based surrogate models based on frequency truncations have also
been theoretically characterized for noiseless bounded-gate circuits~\cite{
Du2025} and experimentally demonstrated on quantum processors~\cite{liao2025sample,Liao2026}.
However, these guarantees are primarily designed for expectation-value estimation. For optimization, the relevant object is the gradient field, and Fourier analyses of variational quantum circuits suggest that derivative information can depend on modes that are different from, and often less compressed than, those dominating the function value~\cite{PhysRevA.108.032406}. Consequently, a frequency cutoff that accurately approximates the objective value may still remove paths that make non-negligible contributions to the gradient.

\textbf{\texttt{CT-PBS} method.}
A second class of methods, often described as sparse Pauli dynamics or Pauli-path simulation, prunes Pauli strings or paths whose instantaneous coefficients fall below a prescribed threshold~\cite{PRXQuantum.6.020302,Lin2026utilityscalequantum,doi:10.1126/sciadv.adk4321,gharibyan2025practical,shao2025pauli}. These methods are simple and effective when the Heisenberg-evolved observable remains sparse in the Pauli basis, and they have been successfully applied to spin dynamics, utility-scale state-preparation benchmarks, and practical Pauli-path simulations with empirical convergence diagnostics~\cite{PRXQuantum.6.020302,Lin2026utilityscalequantum,doi:10.1126/sciadv.adk4321}. Nevertheless, the coefficient magnitude is not a reliable proxy for gradient contribution. A path with a small coefficient can have a large derivative with respect to a circuit parameter; for instance, when a final branching angle is close to zero, the sine branch has coefficient $\sin(\bmt_j)$ but derivative $\cos(\bmt_j)$, which remains order one. Thus, coefficient truncation can preserve the dominant value contributions while discarding terms that determine the local descent direction. This mismatch becomes particularly problematic in iterative optimization, where the relevant coefficients change after every parameter update.

\textbf{\texttt{WT-PBS} method.}
A third family of methods restricts the operator space by retaining only Pauli strings or Pauli paths with bounded Pauli weight, namely those supported on at most a prescribed number of non-identity single-qubit factors~\cite{xct1-7kf2,GonzalezGarcia2025paulipath,j1gg-s6zb,PhysRevLett.133.120603,lh6x-7rc3,angrisani2025simulating}. The motivation is that local noise, scrambling, or average-case randomness can suppress high-weight components of the Heisenberg-evolved observable. In noisy circuits, recent results establish polynomial or quasi-polynomial simulation guarantees by showing that high-weight Pauli components are exponentially damped under suitable local noise models~\cite{GonzalezGarcia2025paulipath,PhysRevLett.133.120603,angrisani2025simulating}. Related noiseless average-case analyses show that low-weight Pauli propagation can approximate expectation values for certain locally scrambling circuit ensembles~\cite{xct1-7kf2,lh6x-7rc3}. However, these results typically concern expectation estimation under structural assumptions such as local noise, sparse observables, average-case circuit distributions, or restricted non-Clifford resources. They do not directly provide guarantees for gradient-based optimization of fixed variational circuits with non-uniform and correlated parameters, such as QAOA, Hamiltonian variational ansätze, or trained quantum neural networks. Empirical evidence further indicates that low-weight Pauli propagation can yield biased gradients in variational optimization~\cite{4kyq-n8jb}.

Despite their different motivations, the above methods share a common principle: they replace the full Pauli-path expansion by a deterministic truncated subset chosen mainly to control expectation-value error. This design is suitable for forward simulation, but it does not in general control the bias of the gradient estimator or the optimization trajectory induced by that bias. Our work addresses this gap by avoiding deterministic path truncation. \texttt{SPPS} samples from the full legal Pauli-path space and uses importance reweighting together with path automatic differentiation to construct unbiased stochastic estimators of both expectation values and gradients. This shifts \texttt{PBS} from static expectation estimation to faithful simulation of gradient-based quantum optimization.

\section{\texttt{Tb-PBS} leads to biased gradient (Proof of Theorem~\ref{thm_toy_bias})}
\label{app_bias_grad}

In this appendix, we prove that \texttt{Tb-PBS} approximate the expectation value with arbitrarily small error while inducing a non-vanishing gradient bias. In particular, Theorem~\ref{thm_toy_bias_app} establishes this separation by showing that removing a single value-small but derivative-large path can yield an $\epsilon$-accurate value estimator with order-one gradient error. Sec.~\ref{app_worst_case_truncation} then shows that \texttt{CT}-, \texttt{FT}-, and \texttt{WT-PBS} can realize this removal in explicit circuits.

\begin{theorem}[Formal statement of Theorem~\ref{thm_toy_bias}]
\label{thm_toy_bias_app}
Using the notation introduced in Section~\ref{app_PP_notion}, consider the function
\begin{equation}
\label{thm_toy_bias_app_eq1}
f(\bmt)
= \Tr \big[ O U(\bmt)\rho U^\dagger(\bmt)\big]
= \sum_{\bmo\in\Omega}\Psi_{\bmo}(\bmt)
\Tr \big[P_{\bmo}(O)\rho\big],
\end{equation}
where $\rho= ({|0\>}{\<0|})^{\otimes n}$, $O$ is a Pauli observable, and $f$ is non-constant on $\mathbb{R}^L$. Then, for any $\epsilon\in(0,1)$, there exist a set $\hat{\Omega}\subsetneq\Omega$ and a region $\Theta\subset\mathbb{R}^L$ such that, for any $\bmt\in\Theta$,
\begin{align}
\big| \hat{f}(\bmt)-f(\bmt) \big|
&\le \epsilon,
\label{thm_toy_bias_app_eq2} \\
\big\|\nabla_{\bmt}\hat{f}(\bmt)-\nabla_{\bmt}f(\bmt)\big\|_2
&\ge 1-\epsilon,
\label{thm_toy_bias_app_eq3}
\end{align}
where
\begin{equation}
\hat{f}(\bmt)
=
\sum_{\bmo\in\hat{\Omega}}
\Psi_{\bmo}(\bmt)
\Tr \big[P_{\bmo}(O)\rho\big].
\end{equation}
\end{theorem}

\begin{proof}
Since $f$ is non-constant, there exists at least one coordinate $j\in[L]$ such that $\partial_{\bmt_j} f$ is not identically zero. Differentiating Eq.~\eqref{thm_toy_bias_app_eq1} with respect to $\bmt_j$ gives
\begin{equation}
\partial_{\bmt_j} f(\bmt)
=
\sum_{\bmo\in\Omega:\,\bmo_j\neq 0}
\partial_{\bmt_j}\Psi_{\bmo}(\bmt)\,
\Tr[P_{\bmo}(O)\rho].
\end{equation}
If $\Tr[P_{\bmo}(O)\rho]=0$ for every path with $\bmo_j\neq 0$, then $\partial_{\bmt_j} f$ would vanish identically, contradicting the choice of $j$. Hence, there exists a path $\bmo^*\in\Omega$ such that $\bmo_j^*\in\{\pm1\}$ and
\begin{equation}
\label{eq_nonzero_coeff}
\Tr \big[P_{\bmo^*}(O)\rho\big]\neq 0.
\end{equation}

We construct the truncated path set by removing this single path,
\begin{equation}
\label{thm_toy_bias_app_truncated_set}
\hat{\Omega}=\Omega\setminus\{\bmo^*\}.
\end{equation}
Let
\begin{equation}
r(\bmt):=f(\bmt)-\hat f(\bmt)
=
\Psi_{\bmo^*}(\bmt)
\Tr \big[P_{\bmo^*}(O)\rho\big].
\label{thm_toy_bias_app_r_theta}
\end{equation}
Since $\rho= ({|0\>}{\<0|})^{\otimes n}$ and $P_{\bmo^*}(O)$ is a Pauli operator, Eq.~\eqref{eq_nonzero_coeff} implies
\begin{equation}
\label{thm_toy_bias_app_trace_pm1}
\Tr \big[P_{\bmo^*}(O)\rho\big]=\pm1.
\end{equation}
Indeed, the expectation is nonzero only when $P_{\bmo^*}(O)\in\{\pm1\}\times\{I,Z\}^{\otimes n}$. Therefore,
\begin{align}
\left| r(\bmt) \right|
&=
\big|\Psi_{\bmo^*}(\bmt)\big|,
\label{thm_toy_bias_app_abs_r_theta} \\
\left\| \nabla_{\bmt}r(\bmt) \right\|_2
&=
\left\| \nabla_{\bmt}\Psi_{\bmo^*}(\bmt) \right\|_2 .
\label{thm_toy_bias_app_norm_nabla_r_theta}
\end{align}

We next construct a region where the removed path has a small value contribution but a large derivative contribution. Let
\begin{equation}
\Theta=\Theta_1\times\Theta_2\times\cdots\times\Theta_L,
\label{thm_toy_bias_app_theta_region}
\end{equation}
where, for each $k\in[L]$,
\begin{equation}
\label{thm_toy_bias_app_Theta_k}
\Theta_k=
\begin{cases}
[-\pi,\pi],
& \bmo_k^*=0,\\[4pt]
\left[-\sqrt{\frac{\epsilon}{L}}, \sqrt{\frac{\epsilon}{L}}\right],
& k\neq j \text{ and } \bmo_k^*=1,\\[8pt]
\left[\frac{\pi}{2}-\sqrt{\frac{\epsilon}{L}}, 
\frac{\pi}{2}+\sqrt{\frac{\epsilon}{L}}\right],
& k\neq j \text{ and } \bmo_k^*=-1,\\[8pt]
\left[\frac{\pi}{2}-\epsilon, 
\frac{\pi}{2}+\epsilon\right],
& k=j \text{ and } \bmo_k^*=1,\\[8pt]
[-\epsilon,\epsilon],
& k=j \text{ and } \bmo_k^*=-1.
\end{cases}
\end{equation}

We first prove the expectation value error bound. Suppose $\bmo_j^*=1$. Then, for any $\bmt\in\Theta$, we have $\bmt_j\in[\pi/2-\epsilon,\pi/2+\epsilon]$, and hence
\begin{equation}
|\cos\bmt_j|\le \sin\epsilon\le \epsilon.
\end{equation}
Since all other trigonometric factors in $\Psi_{\bmo^*}(\bmt)$ have magnitude at most one,
\begin{align}
|r(\bmt)|
&=
\prod_{k=1}^{L}
\left| \cos\bmt_k \right|^{\mathbbm{1}[\bmo_k^*=1]}
\left| \sin\bmt_k \right|^{\mathbbm{1}[\bmo_k^*=-1]}
\le |\cos\bmt_j|
\le \epsilon.
\label{thm_toy_bias_app_bound_r}
\end{align}
If $\bmo_j^*=-1$, then $\bmt_j\in[-\epsilon,\epsilon]$, and the same argument gives $|r(\bmt)|\le |\sin\bmt_j|\le\epsilon$. Therefore,
\begin{equation}
|\hat f(\bmt)-f(\bmt)|=|r(\bmt)|\le\epsilon.
\end{equation}

It remains to prove the gradient error bound. Again consider the case $\bmo_j^*=1$. From Eq.~\eqref{thm_toy_bias_app_norm_nabla_r_theta},
\begin{align}
\| \nabla_{\bmt}r(\bmt)\|_2
&\ge
\left|
\frac{\partial}{\partial \bmt_j}\Psi_{\bmo^*}(\bmt)
\right|
\notag\\
&=
|\sin\bmt_j|
\prod_{k\neq j}
|\cos\bmt_k|^{\mathbbm{1}[\bmo_k^*=1]}
|\sin\bmt_k|^{\mathbbm{1}[\bmo_k^*=-1]}
\notag\\
&\ge
\cos\epsilon
\left(\cos\sqrt{\frac{\epsilon}{L}}\right)^L
\notag\\
&=
\sqrt{
(1-\sin^2\epsilon)
\left(1-\sin^2\sqrt{\frac{\epsilon}{L}}\right)^L
}
\notag\\
&\ge
\sqrt{
(1-\epsilon^2)
\left(1-\frac{\epsilon}{L}\right)^L
}
\notag\\
&\ge
\sqrt{(1-\epsilon^2)(1-\epsilon)}
\ge 1-\epsilon.
\label{thm_toy_bias_app_nabla_r_theta_value_3}
\end{align}
Here we used $|\sin x|\le x$ for $x\in\{\epsilon,\sqrt{\epsilon/L}\}$ and $(1-\epsilon/L)^L\ge 1-\epsilon$. The case $\bmo_j^*=-1$ is analogous. Since
\begin{equation}
\nabla_{\bmt}r(\bmt)
=
\nabla_{\bmt}f(\bmt)-\nabla_{\bmt}\hat f(\bmt),
\end{equation}
we obtain
\begin{equation}
\big\|\nabla_{\bmt}\hat{f}(\bmt)-\nabla_{\bmt}f(\bmt)\big\|_2
\ge 1-\epsilon.
\end{equation}
\end{proof}

\subsection{Explicit case with standard truncation rules}
\label{app_worst_case_truncation}

We next provide an explicit $n$-qubit construction showing that \texttt{CT}, \texttt{FT}, and \texttt{WT} can realize the value-gradient separation in Theorem~\ref{thm_toy_bias_app}. 
The construction uses a single-qubit observable, while a Clifford fanout layer generates high-weight Pauli strings during Heisenberg propagation. 
Moreover, the truncated estimator retains a nonzero contribution, so the separation is not caused by removing all paths. 

Let $n\ge 3$, $\rho=(|0\>\<0|)^{\otimes n}$, and consider the single-qubit observable
\begin{equation}
\label{app_worst_case_single_observable}
O=Z_1 .
\end{equation}
Let
\begin{equation}
\label{app_worst_case_cz_fanout}
\mathcal C=\prod_{j=2}^{n} CZ_{1j}
\end{equation}
be a CZ fanout layer. We consider the circuit
\begin{equation}
\label{app_worst_case_single_observable_circuit}
U(\bmt)
=
R_{X_1}(\bmt_1)\,
\mathcal C\,
\left(\prod_{j=2}^{n}R_{X_j}(\bmt_j)\right)
R_{X_1}(\bmt_{n+1}),
\end{equation}
where $\bmt=(\bmt_1,\ldots,\bmt_{n+1})$.

Using
\begin{equation}
R_{X_j}^\dagger(x)Z_jR_{X_j}(x)=\cos x\, Z_j+\sin x\,Y_j,
\qquad
R_{X_j}^\dagger(x)Y_jR_{X_j}(x)=\cos x\,Y_j-\sin x\,Z_j,
\end{equation}
and
\begin{equation}
\mathcal C^\dagger Z_1\mathcal C=Z_1,
\qquad
\mathcal C^\dagger Y_1\mathcal C=Y_1Z_2\cdots Z_n,
\end{equation}
the expectation-value function is
\begin{equation}
\label{app_worst_case_single_exact_function}
f(\bmt)
=
\Tr[U^\dagger(\bmt)OU(\bmt)\rho]
=
F_{\rm keep}(\bmt)+F_{\rm drop}(\bmt),
\end{equation}
where
\begin{equation}
\label{app_worst_case_single_terms}
F_{\rm keep}(\bmt)
=
\cos\bmt_1\cos\bmt_{n+1},
\qquad
F_{\rm drop}(\bmt)
=
-\sin\bmt_1
\left(\prod_{j=2}^{n}\cos\bmt_j\right)
\sin\bmt_{n+1}.
\end{equation}
All other branches contain at least one $X$ or $Y$ factor and hence have zero expectation under $\rho$.

We consider the parameter region
\begin{equation}
\label{app_worst_case_single_region}
|\bmt_1|\le \frac{\epsilon}{64},
\qquad
|\bmt_j|\le \frac{\epsilon}{64n},\quad j=2,\ldots,n,
\qquad
\cos\bmt_{n+1}=\frac{\sqrt{\epsilon}}{4},
\end{equation}
where $\epsilon\in(0,1)$. On this region, the contribution $F_{\rm drop}$ is small in value:
\begin{equation}
\label{app_worst_case_single_value}
|F_{\rm drop}(\bmt)|
\le
|\sin\bmt_1|
\le
\frac{\epsilon}{64}
\le
\epsilon .
\end{equation}
However, its derivative with respect to $\bmt_1$ remains order one:
\begin{equation}
\label{app_worst_case_single_grad}
|\partial_{\bmt_1}F_{\rm drop}(\bmt)|
=
|\cos\bmt_1|
\left(\prod_{j=2}^{n}|\cos\bmt_j|\right)
|\sin\bmt_{n+1}|
\ge
1-\epsilon .
\end{equation}
The last inequality follows from Eq.~\eqref{app_worst_case_single_region} and
$\sin\bmt_{n+1}=\sqrt{1-\epsilon/16}$.

We now show that standard truncation rules used in \texttt{Tb-PBS} remove $F_{\rm drop}$ while retaining the nonzero contribution $F_{\rm keep}$.

\textbf{Biased gradients of \texttt{CT-PBS} method.}
The coefficient magnitude of the path contributing to $F_{\rm drop}$ is
\begin{equation}
|\sin\bmt_1|
\left(\prod_{j=2}^{n}|\cos\bmt_j|\right)
|\sin\bmt_{n+1}|
\le
\frac{\epsilon}{64}.
\end{equation}
In contrast, the coefficient magnitude of the path contributing to $F_{\rm keep}$ is
\begin{equation}
|\cos\bmt_1\cos\bmt_{n+1}|
\ge
\frac{\sqrt{\epsilon}}{8}.
\end{equation}
Choose the coefficient threshold
\begin{equation}
\label{app_worst_case_single_ct_threshold}
\tau=\frac{\sqrt{\epsilon}}{16}.
\end{equation}
Then \texttt{CT-PBS}  removes $F_{\rm drop}$ and retains $F_{\rm keep}$, yielding
\begin{equation}
\label{app_worst_case_single_ct_fhat}
\hat f_{\texttt{CT}}(\bmt)=F_{\rm keep}(\bmt).
\end{equation}

\textbf{Biased gradients of \texttt{FT-PBS} method.}
The path contributing to $F_{\rm keep}$ has Fourier level $2$, corresponding to the two nonzero factors $\cos\bmt_1$ and $\cos\bmt_{n+1}$. 
The path contributing to $F_{\rm drop}$ has Fourier level $n+1$, corresponding to the factors
\begin{equation}
\sin\bmt_1,\quad
\cos\bmt_2,\ldots,\cos\bmt_n,\quad
\sin\bmt_{n+1}.
\end{equation}
Choose the cutoff
\begin{equation}
\label{app_worst_case_single_ft_threshold}
\nu=n .
\end{equation}
Then \texttt{FT-PBS} removes $F_{\rm drop}$ and retains $F_{\rm keep}$, yielding
\begin{equation}
\label{app_worst_case_single_ft_fhat}
\hat f_{\texttt{FT}}(\bmt)=F_{\rm keep}(\bmt).
\end{equation}

\textbf{Biased gradients of \texttt{WT-PBS} method.}
The path contributing to $F_{\rm keep}$ stays supported on qubit $1$, so its Pauli weight is at most $1$ throughout the propagation. 
In contrast, after the CZ fanout layer, the path contributing to $F_{\rm drop}$ contains the Pauli string
\begin{equation}
Y_1Z_2\cdots Z_n,
\end{equation}
which has Pauli weight $n$. Choose the Pauli-weight cutoff
\begin{equation}
\label{app_worst_case_single_wt_threshold}
\gamma=n-1 .
\end{equation}
Then \texttt{WT} removes $F_{\rm drop}$ and retains $F_{\rm keep}$, yielding
\begin{equation}
\label{app_worst_case_single_wt_fhat}
\hat f_{\texttt{WT}}(\bmt)=F_{\rm keep}(\bmt).
\end{equation}

Therefore, even with a single-qubit observable, \texttt{CT}, \texttt{FT}, and \texttt{WT} can retain a nonzero truncated estimate while discarding a value-small but gradient-dominant propagation path. This establishes an explicit realization of the value-gradient separation in Theorem~\ref{thm_toy_bias_app}.

\section{\texttt{Tb-PBS} leads to sub-optimal convergence (Proof of Corollary~\ref{cor_opt_gap})}
\label{app_bias_GD}

In this appendix, we prove that truncation-induced gradient bias can accumulate along optimization trajectories, leading to a macroscopic gap between the objective values reached by exact-gradient descent and \texttt{Tb-PBS}-driven descent.

\begin{theorem}[Formal statement of Corollary~\ref{cor_opt_gap}]
\label{cor_opt_gap_app}
There exist a circuit $U(\bmt)$, an observable $O$, and a truncated path set $\hat{\Omega}\subsetneq\Omega$ produced by \texttt{CT}, \texttt{FT}, or \texttt{WT} in suitable truncation regimes, such that the following holds. 
For any $\epsilon\in(0,1)$, there is an initialization $\bmt^{(0)}$ satisfying
\begin{equation}
\big|\hat f(\bmt_{\rm tr}^{(t)})-f(\bmt_{\rm tr}^{(t)})\big|\le \epsilon, 
\qquad 
\forall t\in\mathbb{N},
\end{equation}
while the exact-gradient trajectory $\bmt_{\rm ex}^{(t)}$ and the \texttt{Tb-PBS}-gradient trajectory $\bmt_{\rm tr}^{(t)}$ obey
\begin{equation}
\big|f(\bmt_{\rm tr}^{(t)})-f(\bmt_{\rm ex}^{(t)})\big|
\ge 
\Omega(\eta t\|O\|_2^2), 
\qquad 
1\le t\le \O\!\left(\frac{1}{\eta\|O\|_2}\right).
\end{equation}
Moreover, the asymptotic gap satisfies
\begin{equation}
\liminf_{t\rightarrow\infty}
\big|f(\bmt_{\rm tr}^{(t)})-f(\bmt_{\rm ex}^{(t)})\big|
\ge 
\Omega(\|O\|_2).
\end{equation}
\end{theorem}

\begin{proof}
Let $n\ge 3$, let $\alpha\ge 1$, and set
\begin{equation}
\rho=(|0\>\<0|)^{\otimes n},
\qquad
O=\alpha Z_1 .
\end{equation}
Then $O$ is a scaled single-qubit Pauli observable and
\begin{equation}
\label{cor_opt_gap_rich_norm_O}
\|O\|_2=\alpha .
\end{equation}

Let
\begin{equation}
\label{cor_opt_gap_rich_cz_fanout}
\mathcal C=\prod_{j=2}^{n}CZ_{1j}
\end{equation}
be a CZ fanout layer. Consider the circuit
\begin{equation}
\label{cor_opt_gap_rich_circuit}
U(\bmt)
=
R_{X_1}(\bmt_1)\,
\mathcal C\,
\left(\prod_{j=2}^{n}R_{X_j}(\bmt_j)\right)
R_{X_1}\!\left(\frac{\pi}{4}\right),
\qquad
\bmt=(\bmt_1,\ldots,\bmt_n),
\end{equation}
where the last rotation is fixed.

Using
\begin{equation}
R_{X_j}^\dagger(x)Z_jR_{X_j}(x)=\cos x\,Z_j+\sin x\,Y_j,
\qquad
R_{X_j}^\dagger(x)Y_jR_{X_j}(x)=\cos x\,Y_j-\sin x\,Z_j,
\end{equation}
and
\begin{equation}
\mathcal C^\dagger Z_1\mathcal C=Z_1,
\qquad
\mathcal C^\dagger Y_1\mathcal C=Y_1Z_2\cdots Z_n,
\end{equation}
the expectation value function is
\begin{equation}
\label{cor_opt_gap_rich_exact_f}
f(\bmt)
=
F_{\rm keep}(\bmt)+F_{\rm drop}(\bmt),
\end{equation}
where
\begin{equation}
\label{cor_opt_gap_rich_terms}
F_{\rm keep}(\bmt)
=
\frac{\alpha}{\sqrt 2}\cos\bmt_1,
\qquad
F_{\rm drop}(\bmt)
=
-\frac{\alpha}{\sqrt 2}\sin\bmt_1
\left(\prod_{j=2}^{n}\cos\bmt_j\right).
\end{equation}
All other branches contain at least one $X$ or $Y$ factor and therefore have zero expectation under $\rho$.

We choose truncation regimes that remove $F_{\rm drop}$ while retaining $F_{\rm keep}$. 
For \texttt{FT}, the retained branch has Fourier level $1$, whereas the dropped branch has Fourier level $n$; choose
\begin{equation}
\label{cor_opt_gap_rich_ft_threshold}
\nu=n-1 .
\end{equation}
For \texttt{WT}, the retained branch stays supported on qubit $1$, whereas the dropped branch contains the Pauli string
\begin{equation}
Y_1Z_2\cdots Z_n
\end{equation}
after the CZ fanout layer and therefore has Pauli weight $n$; choose
\begin{equation}
\label{cor_opt_gap_rich_wt_threshold}
\gamma=n-1 .
\end{equation}
For \texttt{CT}, we will initialize the truncated trajectory in a region where the coefficient of $F_{\rm drop}$ is uniformly small while the coefficient of $F_{\rm keep}$ is order one. 
Specifically, let
\begin{equation}
\label{cor_opt_gap_rich_delta_def}
\delta=\min\left\{\frac{\sqrt 2\,\epsilon}{2\alpha},\frac{1}{16}\right\},
\end{equation}
and choose any fixed threshold
\begin{equation}
\label{cor_opt_gap_rich_ct_threshold}
\frac{1}{8}<\tau<\frac{1}{2}.
\end{equation}
Along the \texttt{Tb-PBS} trajectory constructed below, $|\sin\bmt_1|\le \delta$ and $|\cos\bmt_1|\ge \cos\delta>1/2$, so \texttt{CT} removes $F_{\rm drop}$ and retains $F_{\rm keep}$.

Thus, for \texttt{CT}, \texttt{FT}, and \texttt{WT}, the truncated objective along the \texttt{Tb-PBS} trajectory is
\begin{equation}
\label{cor_opt_gap_rich_truncated_f}
\hat f(\bmt)
=
F_{\rm keep}(\bmt)
=
\frac{\alpha}{\sqrt 2}\cos\bmt_1 .
\end{equation}

Initialize both trajectories at
\begin{equation}
\label{cor_opt_gap_rich_initialization}
\bmt_{\rm ex}^{(0)}
=
\bmt_{\rm tr}^{(0)}
=
(\pi-\delta,0,\ldots,0).
\end{equation}
Since $\bmt_j^{(0)}=0$ for $j=2,\ldots,n$ and
\begin{equation}
\partial_{\bmt_j}f(\bmt)=0,\qquad
\partial_{\bmt_j}\hat f(\bmt)=0
\qquad
\text{whenever } \bmt_j=0,\quad j=2,\ldots,n,
\end{equation}
both trajectories remain in the one-dimensional invariant subspace
\begin{equation}
\bmt_{{\rm ex},j}^{(t)}=\bmt_{{\rm tr},j}^{(t)}=0,
\qquad
j=2,\ldots,n .
\end{equation}
Therefore, it suffices to analyze the first coordinate.

The \texttt{Tb-PBS}-driven update is
\begin{equation}
\label{cor_opt_gap_rich_tr_update}
\bmt_{{\rm tr},1}^{(t+1)}
=
\bmt_{{\rm tr},1}^{(t)}
+
\frac{\eta\alpha}{\sqrt 2}\sin\bmt_{{\rm tr},1}^{(t)} .
\end{equation}
Let
\begin{equation}
u_t=\pi-\bmt_{{\rm tr},1}^{(t)} .
\end{equation}
Then $u_0=\delta$ and
\begin{equation}
u_{t+1}
=
u_t-\frac{\eta\alpha}{\sqrt 2}\sin u_t .
\end{equation}
For a sufficiently small learning rate, $u_t$ decreases monotonically to $0$. Hence
\begin{equation}
0\le u_t\le \delta,
\qquad
\bmt_{{\rm tr},1}^{(t)}\rightarrow \pi .
\end{equation}
In particular, the \texttt{Tb-PBS} trajectory changes with $t$ and converges to the nonzero truncated value
\begin{equation}
\label{cor_opt_gap_rich_tr_convergence}
\lim_{t\to\infty}\hat f(\bmt_{\rm tr}^{(t)})
=
-\frac{\alpha}{\sqrt 2}.
\end{equation}

Along this trajectory, the dropped contribution remains uniformly small:
\begin{equation}
\label{cor_opt_gap_rich_value_error}
\big|
\hat f(\bmt_{\rm tr}^{(t)})
-
f(\bmt_{\rm tr}^{(t)})
\big|
=
|F_{\rm drop}(\bmt_{\rm tr}^{(t)})|
=
\frac{\alpha}{\sqrt 2}
|\sin\bmt_{{\rm tr},1}^{(t)}|
\le
\frac{\alpha}{\sqrt 2}\delta
\le
\epsilon,
\qquad
\forall t\in\mathbb N .
\end{equation}

We now compare it with the exact-gradient trajectory. 
On the invariant subspace, the exact objective reduces to
\begin{equation}
\label{cor_opt_gap_rich_one_dim_exact_f}
f(\bmt_{\rm ex}^{(t)})
=
\frac{\alpha}{\sqrt 2}
\left(
\cos\bmt_{{\rm ex},1}^{(t)}
-
\sin\bmt_{{\rm ex},1}^{(t)}
\right).
\end{equation}
Thus the exact-gradient update is
\begin{equation}
\label{cor_opt_gap_rich_exact_update}
\bmt_{{\rm ex},1}^{(t+1)}
=
\bmt_{{\rm ex},1}^{(t)}
+
\frac{\eta\alpha}{\sqrt 2}
\left(
\sin\bmt_{{\rm ex},1}^{(t)}
+
\cos\bmt_{{\rm ex},1}^{(t)}
\right).
\end{equation}
For $\bmt_{{\rm ex},1}^{(t)}\in[7\pi/8,\pi]$, we have
\begin{equation}
\sin\bmt_{{\rm ex},1}^{(t)}+\cos\bmt_{{\rm ex},1}^{(t)}
\le
-c_0
\end{equation}
for a universal constant $c_0>0$. 
Hence, for $1\le t\le c/(\eta\alpha)$ with a sufficiently small universal constant $c>0$, the exact trajectory remains in $[7\pi/8,\pi]$ and satisfies
\begin{equation}
\label{cor_opt_gap_rich_exact_move}
\bmt_{{\rm ex},1}^{(t)}
\le
\pi-\delta
-
c_1\eta\alpha t
\end{equation}
for another universal constant $c_1>0$. 
Meanwhile, the truncated trajectory satisfies
\begin{equation}
\bmt_{{\rm tr},1}^{(t)}\ge \pi-\delta .
\end{equation}
Therefore,
\begin{equation}
\label{cor_opt_gap_rich_param_sep}
\bmt_{{\rm tr},1}^{(t)}-\bmt_{{\rm ex},1}^{(t)}
\ge
c_1\eta\alpha t .
\end{equation}

The one-dimensional function
\begin{equation}
g(x)=\frac{\alpha}{\sqrt 2}(\cos x-\sin x)
\end{equation}
has derivative
\begin{equation}
g'(x)=\frac{\alpha}{\sqrt 2}(-\sin x-\cos x).
\end{equation}
On $[7\pi/8,\pi]$, $g'(x)\ge c_2\alpha$ for a universal constant $c_2>0$. 
Using Eq.~\eqref{cor_opt_gap_rich_param_sep}, we obtain
\begin{equation}
\label{cor_opt_gap_rich_early_gap}
\big|
f(\bmt_{\rm tr}^{(t)})
-
f(\bmt_{\rm ex}^{(t)})
\big|
\ge
c_3\eta t\alpha^2
\end{equation}
for a universal constant $c_3>0$ and all $1\le t\le c/(\eta\alpha)$. 
Since $\|O\|_2=\alpha$ and $\alpha\ge1$, this implies
\begin{equation}
\big|
f(\bmt_{\rm tr}^{(t)})
-
f(\bmt_{\rm ex}^{(t)})
\big|
\ge
\Omega(\eta t\|O\|_2^2),
\qquad
1\le t\le \O\!\left(\frac{1}{\eta\|O\|_2}\right).
\end{equation}

Finally, for a sufficiently small learning rate, Eq.~\eqref{cor_opt_gap_rich_tr_update} converges to the local minimizer $\bmt_{{\rm tr},1}=\pi$ of $\hat f$, while Eq.~\eqref{cor_opt_gap_rich_exact_update} converges to the local minimizer $\bmt_{{\rm ex},1}=3\pi/4$ of the exact objective in Eq.~\eqref{cor_opt_gap_rich_one_dim_exact_f}. Consequently,
\begin{equation}
\lim_{t\to\infty} f(\bmt_{\rm tr}^{(t)})
=
-\frac{\alpha}{\sqrt 2},
\qquad
\lim_{t\to\infty} f(\bmt_{\rm ex}^{(t)})
=
-\alpha .
\end{equation}
Therefore,
\begin{equation}
\liminf_{t\to\infty}
\big|
f(\bmt_{\rm tr}^{(t)})
-
f(\bmt_{\rm ex}^{(t)})
\big|
=
\alpha\left(1-\frac{1}{\sqrt 2}\right)
=
\Omega(\|O\|_2).
\end{equation}
This completes the proof.
\end{proof}

\section{Implementation details of \texttt{SPPS}}
\label{app_spp_alg}

In this appendix, we practical ingredients used in the implementation of \texttt{SPPS}. In particular, we introduce the sequential sampling rule and the importance reweighting procedure, the choice of the smoothing parameter $a$ in the sampling distribution, the derivation and numerically stable implementation of path automatic differentiation (PAD), and the adaptive gradient-error proxy used to determine the number of samples. Throughout this appendix, we describe the estimator for a Pauli observable $O\in\mathcal P_n$ by default; for a general observable $O=\sum_{m=1}^{N_O} \bm{c}_m O_m$, \texttt{SPPS} applies the same procedure to each Pauli term $O_m$ and combines the resulting estimators linearly.

\subsection{Sequential path sampling and importance reweighting}
\label{app_spp_estimator_sample_rule}

We follow the notation introduced in Section~\ref{app_PP_notion}. For a Pauli observable $O\in\mathcal P_n$, the Heisenberg-picture propagation gives
\begin{equation}
\label{app_spp_value_path_expansion}
f(\bmt)
=
\Tr \left[ O U(\bmt)\rho U(\bmt)^\dagger \right]
=
\sum_{\bmo\in\Omega}
\Psi_{\bmo}(\bmt)\Tr\left[P_{\bmo}(O)\rho\right],
\end{equation}
where $\bmo=(\bmo_1,\ldots,\bmo_L)\in\Omega\subseteq\{0,\pm1\}^L$ denotes a legal propagation path. In \texttt{SPPS}, paths are sampled sequentially during propagation instead of being enumerated explicitly. Therefore, the sampling probability factorizes as
\begin{equation}
\label{app_spp_sampling_prob_factorization}
\Pr(\bmo)
=
\prod_{j=1}^{L}
{\Pr}_j(\bmo_j\mid \bmo_{1:j-1}).
\end{equation}

At the $j$-th parametrized gate, the conditional distribution is determined by the commutation relation between the currently propagated Pauli operator and the generator of $R_{P_j}(\bmt_j)$. If they commute, no branching occurs and the path variable is deterministic:
\begin{equation}
\label{app_spp_sampling_rule_commute}
{\Pr}_j(\bmo_j=0\mid \bmo_{1:j-1})=1.
\end{equation}
If they anti-commute, the propagation splits into the cosine and sine branches. \texttt{SPPS} samples these two branches according to
\begin{equation}
\label{app_spp_sampling_rule_local}
{\Pr}_j(\bmo_j=1\mid \bmo_{1:j-1})
=
q_j(\bmt),
\qquad
{\Pr}_j(\bmo_j=-1\mid \bmo_{1:j-1})
=
1-q_j(\bmt),
\end{equation}
where
\begin{equation}
\label{app_spp_q_def}
q_j(\bmt)
=
\frac{|\cos\bmt_j|+a}{|\cos\bmt_j|+|\sin\bmt_j|+2a}.
\end{equation}
Here, $a$ is a smoothing parameter. When $a=0$, the distribution reduces to magnitude-proportional sampling. When $a>0$, it prevents either branch from having vanishing probability when $|\cos\bmt_j|$ or $|\sin\bmt_j|$ is close to zero. This is important for gradient estimation, since a branch with a small value coefficient can still have a large derivative.

For a sampled path $\bmo$, \texttt{SPPS} uses importance reweighting to generate unbiased estimation. Define
\begin{equation}
\label{app_spp_single_path_value_estimator}
\tilde h_{\bmo}(\bmt)
=
\frac{\Psi_{\bmo}(\bmt)\Tr[P_{\bmo} (O) \rho]}{\Pr(\bmo)}.
\end{equation}
Then
\begin{equation}
\label{app_spp_value_unbiased_impl}
\E_{\bmo\sim\Pr}\left[\tilde h_{\bmo}(\bmt)\right]
=
\sum_{\bmo\in\Omega}
\Pr(\bmo)
\frac{\Psi_{\bmo}(\bmt)\Tr[P_{\bmo}(O)\rho]}{\Pr(\bmo)}
=
f(\bmt).
\end{equation}
For $B$ independent samples $\{\bmo^{(b)}\}_{b=1}^{B}$, the expectation value estimator is
\begin{equation}
\label{app_spp_value_estimator_impl}
\tilde f(\bmt)
=
\frac{1}{B}
\sum_{b=1}^{B}
\tilde h_{\bmo^{(b)}}(\bmt).
\end{equation}

\subsection{Choice of the smoothing parameter}
\label{app_spp_smoothing_parameter}

The parameter $a$ in Eq.~(\ref{app_spp_q_def}) controls the trade-off between magnitude-proportional sampling and exploration of derivative-sensitive branches. In practice, \texttt{SPPS} supports two choices of $a$.

\textbf{Fixed smoothing.}
The simplest choice is to use a constant smoothing value
\begin{equation}
\label{app_spp_fixed_a}
a=a_0,
\end{equation}
where $a_0>0$ is a small constant. This choice keeps the sampling distribution fixed across observable terms and optimization steps, and is convenient when a uniform amount of exploration is desired.

\textbf{Adaptive smoothing.}
For observables with multiple Pauli terms during optimization, \texttt{SPPS} can also allows a term-wise adaptive smoothing schedule. Suppose the observable is decomposed as
\begin{equation}
\label{app_spp_observable_decomp_impl}
O=\sum_{m=1}^{N_O} \bm{c}_m O_m,
\qquad
O_m\in\mathcal P_n.
\end{equation}
For the $m$-th Pauli term at optimization step $t$, the local sampling probability is
\begin{equation}
\label{app_spp_termwise_q}
q_{m,j}^{(t)}
=
\frac{|\cos\bmt_j^{(t)}|+a_m^{(t)}}{|\cos\bmt_j^{(t)}|+|\sin\bmt_j^{(t)}|+2a_m^{(t)}}.
\end{equation}
The initial value is set as
\begin{equation}
\label{app_spp_a_init}
a_m^{(0)}=a_{\mathrm{init}},
\end{equation}
and is updated after each estimation step according to the empirical average number of nonzero branch variables in sampled nonzero paths:
\begin{equation}
\label{app_spp_a_update}
a_m^{(t+1)}
=
\frac{b_a}{\max\{\bar r_m^{(t)},\epsilon_a\}}.
\end{equation}
Here, $\bar r_m^{(t)}$ denotes the empirical average number of nonzero entries in the sampled paths for $O_m$, $b_a>0$ is a scaling constant, and $\epsilon_a>0$ is a small numerical floor. This rule increases smoothing when sampled paths are sparse and decreases smoothing when the propagation already explores many active branches. Thus, \texttt{SPPS} avoids overly deterministic sampling in sparse regimes while approaching magnitude-proportional sampling when many nonzero branches are naturally activated.

\subsection{Path automatic differentiation}
\label{app_spp_pad_derivation}

We now derive the PAD estimator used to obtain gradients from the sampled paths. For a fixed path $\bmo$, its contribution after importance reweighting as
\begin{equation}
\label{app_spp_phi_product}
\tilde{h}_{\bmo}(\bmt)= \frac{\Psi_{\bmo}(\bmt)}{\Pr(\bmo)} \Tr[P_{\bmo}(O)\rho] =
\frac{1}{\Pr(\bmo)} \Tr[P_{\bmo}(O)\rho]
\prod_{j\in\mathcal A(\bmo)}
\Psi_{\bmo_j}(\bmt_j),
\end{equation}
where $\mathcal A(\bmo)=\{j:\bmo_j\neq 0\}$ is the set of nontrivial branching positions. For $j\in\mathcal A(\bmo)$,
\begin{equation}
\label{app_spp_phi_def}
\Psi_{\bmo_j}(\bmt_j)
=
\begin{cases}
\cos\bmt_j, & \bmo_j=1,\\[1.0em]
\sin\bmt_j, & \bmo_j=-1.
\end{cases}
\end{equation}
PAD differentiates the path contribution with respect to the trigonometric coefficient. For $j\notin\mathcal A(\bmo)$, the path contribution does not depend on $\bmt_j$ and the derivative is zero. For $j\in\mathcal A(\bmo)$,
\begin{equation}
\label{app_spp_phi_derivative}
\partial_{\bmt_j}
\Psi_{\bmo_j}(\bmt_j)
=
\begin{cases}
-\sin\bmt_j, & \bmo_j=1,\\[1.0em]
\cos\bmt_j, & \bmo_j=-1.
\end{cases}
\end{equation}
Therefore,
\begin{align}
\label{app_spp_pad_direct_derivative}
\partial_{\bmt_j}\tilde h_{\bmo}(\bmt)
=
\frac{1}{\Pr(\bmo)} 
\Tr[P_{\bmo}(O)\rho] 
\partial_{\bmt_j}
\Psi_{\bmo_j}(\bmt_j)
\prod_{k\in\mathcal A(\bmo)\setminus\{j\}}
\Psi_{\bmo_k}(\bmt_k) ,
\qquad j\in\mathcal A(\bmo).
\end{align}
Equivalently, this derivative can be written in a score-function form:
\begin{equation}
\label{app_spp_pad_score_form}
\partial_{\bmt_j}\tilde h_{\bmo}(\bmt)
=
s_j(\bmt,\bmo)\tilde h_{\bmo}(\bmt),
\end{equation}
where
\begin{equation}
\label{app_spp_pad_score}
s_j(\bmt,\bmo)
=
-\tan\bmt_j\,\mathbbm{1}[\bmo_j=1]
+
\cot\bmt_j\,\mathbbm{1}[\bmo_j=-1].
\end{equation}
Combining all components, one \texttt{SPPS} gradient sample is
\begin{equation}
\label{app_spp_single_sample_grad_impl}
\tilde{\bm g}_{\bmo}(\bmt)
=
\left(
s_1(\bmt,\bmo),\ldots,s_L(\bmt,\bmo)
\right)^T
\tilde h_{\bmo}(\bmt),
\end{equation}
and the empirical PAD estimator is
\begin{equation}
\label{app_spp_pad_estimator}
\tilde{\bm g}(\bmt)
=
\frac{1}{B}
\sum_{b=1}^{B}
\tilde{\bm g}_{\bmo^{(b)}}(\bmt).
\end{equation}
Unbiasedness follows by linearity:
\begin{align}
\label{app_spp_pad_unbiased}
\E_{\bmo\sim\Pr}
\left[
\tilde{\bm g}_{\bmo}(\bmt)
\right]
&=
\sum_{\bmo\in\Omega}
\Pr(\bmo)
\nabla_{\bmt}
\left(
\frac{\Psi_{\bmo}(\bmt)\Tr[P_{\bmo}(O)\rho]}{\Pr(\bmo)}
\right) \notag\\
&=
\sum_{\bmo\in\Omega}
\nabla_{\bmt}
\left(
\Psi_{\bmo}(\bmt)\Tr[P_{\bmo}(O)\rho]
\right)
=
\nabla_{\bmt} f(\bmt).
\end{align}

\subsection{Numerically stable PAD implementation}
\label{app_spp_pad_numerical_stability}

The score form in Eq.~\eqref{app_spp_pad_score_form} is algebraically convenient but can be numerically unstable when $|\cos\bmt_j|$ or $|\sin\bmt_j|$ is close to zero. For example, if $\bmo_j=1$ and $|\cos\bmt_j|\approx 0$, the path value $\tilde h_{\bmo}(\bmt)$ can be very small while $-\tan\bmt_j$ is very large. Their product is finite, but evaluating it directly may suffer from overflow or cancellation.

To avoid this issue, the practical implementation of PAD uses the following stable equivalent of Eq.~\eqref{app_spp_pad_direct_derivative}. Define the product excluding the $j$-th active factor as
\begin{equation}
\label{app_spp_prod_excl}
\Pi_{\bmo}^{(-j)}
=
\prod_{k\in\mathcal A(\bmo)\setminus\{j\}}
\Psi_{\bmo_k}(\bmt_k).
\end{equation}
Then
\begin{equation}
\label{app_spp_stable_pad}
\partial_{\bmt_j}\tilde h_{\bmo}(\bmt)
=
\frac{1}{\Pr(\bmo)}
\Tr[P_{\bmo}(O)\rho]\,
\Pi_{\bmo}^{(-j)}
\partial_{\bmt_j} \Psi_{\bmo_j}(\bmt_j).
\end{equation}
In practice, \texttt{SPPS} computes $\Pi_{\bmo}^{(-j)}$ by prefix and suffix products over the active factors. Let $\mathcal A(\bmo)=\{j_1,\ldots,j_{|\mathcal A(\bmo)|}\}$. Then
\begin{equation}
\label{app_spp_prefix_suffix}
\Pi_{\bmo}^{(-j_\ell)}
=
\left(\prod_{r<\ell}
\Psi_{\bmo_{j_r}}(\bmt_{j_r})
\right)
\left(\prod_{r>\ell}
\Psi_{\bmo_{j_r}}(\bmt_{j_r})
\right).
\end{equation}
When the active factor is safely away from zero, the implementation uses the faster score form
\begin{equation}
\label{app_spp_fast_score_impl}
\partial_{\bmt_j}\tilde h_{\bmo}(\bmt)
=
\tilde h_{\bmo}(\bmt)s_j(\bmt,\bmo).
\end{equation}
When the corresponding trigonometric factor or active factor is below a numerical threshold, it switches to the prefix--suffix form in Eq.~\eqref{app_spp_stable_pad}. This switch does not change the estimator; it only replaces an unstable algebraic representation by an equivalent stable one.

\subsection{Adaptive gradient-error proxy}
\label{app_spp_gradient_error_proxy}

The number of samples required by \texttt{SPPS} varies across optimization steps and observable terms. To avoid using a fixed overly conservative budget, the implementation uses an adaptive absolute gradient-error proxy based on two independent macro-replicates.

For a given Pauli term $O_m$, let $\tilde{\bm g}_{m,A}$ and $\tilde{\bm g}_{m,B}$ be two independent \texttt{SPPS} gradient estimates computed with the same sample budget. Both are unbiased estimates of the same term gradient $\bm g_m=\nabla_{\bmt} f_m(\bmt)$. We define
\begin{equation}
\label{app_spp_ab_proxy}
\hat\Delta_m
=
\frac{1}{2}
\left\|
\tilde{\bm g}_{m,A}
-
\tilde{\bm g}_{m,B}
\right\|_2.
\end{equation}
This quantity estimates the stochastic scale of the averaged estimator
\begin{equation}
\label{app_spp_ab_average}
\tilde{\bm g}_m
=
\frac{1}{2}
\left(
\tilde{\bm g}_{m,A}
+
\tilde{\bm g}_{m,B}
\right).
\end{equation}
Indeed, since $\tilde{\bm g}_{m,A}$ and $\tilde{\bm g}_{m,B}$ are independent and unbiased,
\begin{align}
\label{app_spp_proxy_reasoning}
\E
\left[
\left\|
\tilde{\bm g}_{m,A}
-
\tilde{\bm g}_{m,B}
\right\|_2^2
\right]
&=
\E
\left[
\left\|
(\tilde{\bm g}_{m,A}-\bm g_m)
-
(\tilde{\bm g}_{m,B}-\bm g_m)
\right\|_2^2
\right] \notag\\
&=
\E\left[
\left\|
\tilde{\bm g}_{m,A}-\bm g_m
\right\|_2^2
\right]
+
\E\left[
\left\|
\tilde{\bm g}_{m,B}-\bm g_m
\right\|_2^2
\right],
\end{align}
where the cross term vanishes by independence and unbiasedness. If the two replicates use the same budget, then the two terms on the right-hand side are equal, and hence
\begin{equation}
\label{app_spp_proxy_average_relation}
\E\left[\hat\Delta_m^2\right]
=
\frac{1}{4}
\E
\left[
\left\|
\tilde{\bm g}_{m,A}
-
\tilde{\bm g}_{m,B}
\right\|_2^2
\right]
=
\frac{1}{2}
\E\left[
\left\|
\tilde{\bm g}_{m,A}-\bm g_m
\right\|_2^2
\right].
\end{equation}
Moreover, the averaged estimator in Eq.~\eqref{app_spp_ab_average} satisfies
\begin{align}
\label{app_spp_proxy_average_error}
\E
\left[
\left\|
\tilde{\bm g}_m-\bm g_m
\right\|_2^2
\right]
&=
\E
\left[
\left\|
\frac{1}{2}
(\tilde{\bm g}_{m,A}-\bm g_m)
+
\frac{1}{2}
(\tilde{\bm g}_{m,B}-\bm g_m)
\right\|_2^2
\right] \notag\\
&=
\frac{1}{2}
\E\left[
\left\|
\tilde{\bm g}_{m,A}-\bm g_m
\right\|_2^2
\right]
=
\E\left[\hat\Delta_m^2\right].
\end{align}
Thus, $\hat\Delta_m$ has the same second-moment scale as the error of the averaged estimator used by \texttt{SPPS}. This motivates using Eq.~\eqref{app_spp_ab_proxy} as an empirical stopping proxy. The proxy is not used as a theorem-level confidence bound; it is a practical adaptive rule for allocating samples.

For a general observable
\begin{equation}
\label{app_spp_proxy_observable_decomp}
O=\sum_{m=1}^{N_O} \bm{c}_m O_m,
\qquad
O_m\in\mathcal P_n,
\end{equation}
\texttt{SPPS} controls the proxies term by term and then combines the resulting estimates linearly. The implementation supports two stopping mechanisms.

\textbf{Coefficient-normalized proxy.}
When different Pauli terms have different coefficient magnitudes, \texttt{SPPS} can apply the stopping rule to a coefficient-normalized proxy
\begin{equation}
\label{app_spp_term_proxy_unweighted}
\hat\Delta_m^{\rm norm}
=
\frac{\hat\Delta_m}{c_{\rm safe}},
\qquad
c_{\rm safe}=\max\{|\bm{c}_m|,\epsilon_c\},
\end{equation}
where $\epsilon_c>0$ is a small numerical floor. Sampling for term $m$ stops once
\begin{equation}
\label{app_spp_coeff_normalized_stop_rule}
\hat\Delta_m^{\rm norm}
\le
\delta,
\end{equation}
where $\delta$ is the prescribed absolute proxy threshold. If the condition is not satisfied, the sample budget for that term is doubled until the proxy passes the threshold. This rule controls the stochastic error at the level of the unweighted Pauli-term estimator before multiplying by the coefficient $\bm{c}_m$.

\textbf{Root-sum-square proxy allocation.}
Alternatively, when a single overall proxy tolerance $\delta$ is assigned to the full observable, \texttt{SPPS} can distribute it uniformly across Pauli terms as
\begin{equation}
\label{app_spp_term_tol}
\delta_m
=
\frac{\delta}{\sqrt{N_O}},
\qquad
m=1,\ldots,N_O.
\end{equation}
This allocation gives
\begin{equation}
\label{app_spp_global_proxy}
\left(
\sum_{m=1}^{N_O}
\delta_m^2
\right)^{1/2}
=
\delta.
\end{equation}
Each Pauli term doubles its sample budget until
\begin{equation}
\label{app_spp_rss_stop_rule}
\hat\Delta_m
\le
\delta_m
\end{equation}
or the maximum budget is reached. The final observable-level proxy is then reported as
\begin{equation}
\label{app_spp_input_proxy}
\hat\Delta_{\rm obs}
=
\left(
\sum_{m=1}^{N_O}
\hat\Delta_m^2
\right)^{1/2}.
\end{equation}
This mechanism directly controls the aggregate proxy scale of the multi-term observable.

\subsection{Implementation summary}
\label{app_spp_impl_summary}

For each optimization step, \texttt{SPPS} proceeds as follows. First, it constructs the local sampling probabilities in Eq.~\eqref{app_spp_termwise_q} from the current parameters. Second, for each Pauli term, it draws two independent groups of propagation paths and evaluates both value and PAD gradient estimates from the same sampled paths. Third, it computes the A/B proxy in Eq.~\eqref{app_spp_ab_proxy}. If the proxy is above the prescribed tolerance, the sample budget is doubled and the term is re-estimated. Finally, after all active terms either pass the proxy check or reach the maximum budget, the two macro-replicates are averaged and the Pauli-term estimators are combined linearly.

This implementation has three practical consequences. First, \texttt{SPPS} never constructs the full propagation tree. Second, each sampled path contributes all active gradient components through PAD. Third, the A/B proxy allocates more samples only to terms and steps with large stochastic gradient fluctuations, which substantially reduces runtime compared with a fixed worst-case sampling budget.

\section{\texttt{SPPS} generates unbiased gradients (Proof of Theorem~\ref{thm_spp_estimator})}
\label{app_spp_estimator}

In this appendix, we prove the unbiasedness and concentration properties of the \texttt{SPPS} gradient estimator stated in Theorem~\ref{thm_spp_estimator}. We first establish the unbiasedness and variance bound of the single-sample gradient estimator in Section~\ref{app_spp_estimator_grad}. We then average independent samples and apply Bernstein's inequality to obtain a high-probability gradient error bound in Section~\ref{app_spp_estimator_concentration}.

\subsection{Single-sample variance for Gradient estimation}
\label{app_spp_estimator_grad}

For a path $\bmo\in\Omega$, the single-sample gradient estimator for the value function $f(\bmt)$ in Eq.~(\ref{app_spp_value_path_expansion}) is defined as
\begin{equation}\label{app_def_single_grad_estimator}
\tilde{\bm g}_{\bmo}(\bmt) := \frac{\nabla_{\bmt}\Psi_{\bmo}(\bmt)}{\Pr(\bmo)} \Tr \left[P_{\bmo}(O)\rho\right].
\end{equation}
We denote its $k$-th coordinate by $\tilde{\bm g}_{\bmo,k}(\bmt)$. The empirical gradient estimator is then given by
\begin{equation}\label{lem_spp_value_tildeg}
\tilde{\bm g}(\bmt)
=
\frac{1}{B}\sum_{b=1}^{B}\tilde{\bm g}_{\bmo^{(b)}}(\bmt),
\quad \text{where }
\bmo^{(1)},\dots,\bmo^{(B)} \stackrel{\rm i.i.d.}{\sim} \Pr .
\end{equation}

The following lemma establishes the unbiasedness of the single-sample gradient estimator and bounds its variance.

\begin{lemma}\label{lem_spp_grad_unbiased_variance}
Consider the value function $f(\bmt)$ defined in Eq.~(\ref{app_spp_value_path_expansion}) and let $\kappa(\bmt)=\prod_{j=1}^{L}(1+\left| \sin 2\bmt_j \right|)$. Then for every $k\in[L]$, the $k$-th coordinate of the single-sample gradient estimator satisfies
\begin{align}
\E_{\bmo\sim \Pr} \left[\tilde{\bm g}_{\bmo,k}(\bmt)\right]
={}&
\partial_{\bmt_k}f(\bmt), \label{app_grad_unbiased} \\
\Var_{\bmo\sim \Pr} \left[\tilde{\bm g}_{\bmo,k}(\bmt)\right]
\leq{}&
\frac{1}{a}(1+2a)^L \kappa(\bmt).
\label{app_grad_variance_bound}
\end{align}
\end{lemma}

\begin{proof}
We first prove unbiasedness. Since $\Omega$ is finite, differentiation commutes with summation. Differentiating Eq.~(\ref{app_spp_value_path_expansion}) with respect to $\bmt_k$ yields
\begin{equation}\label{app_grad_path_expansion}
\partial_{\bmt_k} f(\bmt)
=
\sum_{\bmo\in\Omega}
\partial_{\bmt_k}\Psi_{\bmo}(\bmt)\Tr \left[P_{\bmo}(O)\rho\right].
\end{equation}
Inserting $\Pr(\bmo)/\Pr(\bmo)=1$ termwise, we obtain
\begin{align}
\partial_{\bmt_k} f(\bmt)
={}&
\sum_{\bmo\in\Omega}
\Pr(\bmo)
\frac{\partial_{\bmt_k}\Psi_{\bmo}(\bmt)}{\Pr(\bmo)}
\Tr \left[P_{\bmo}(O)\rho\right]
\notag\\
={}&
\E_{\bmo\sim\Pr}\left[
\frac{\partial_{\bmt_k}\Psi_{\bmo}(\bmt)}{\Pr(\bmo)}
\Tr \left[P_{\bmo}(O)\rho\right]
\right]
\notag\\
={}&
\E_{\bmo\sim\Pr}\left[\tilde{\bm g}_{\bmo,k}(\bmt)\right].
\end{align}
This proves Eq.~(\ref{app_grad_unbiased}).

Next we prove the variance bound. We have
\begin{align}
\Var_{\bmo\sim \Pr} \left[\tilde{\bm g}_{\bmo,k}(\bmt)\right]
={}&
\E_{\bmo\sim\Pr}\left[\tilde{\bm g}_{\bmo,k}(\bmt)^2\right]
-
\left(
\E_{\bmo\sim\Pr}\left[\tilde{\bm g}_{\bmo,k}(\bmt)\right]
\right)^2
\label{lem_spp_grad_unbiased_variance_2_1}\\
\leq{}&
\E_{\bmo\sim\Pr}\left[\tilde{\bm g}_{\bmo,k}(\bmt)^2\right]
\notag\\
={}&
\sum_{\bmo\in\Omega}
\Pr(\bmo)
\left(
\frac{\partial_{\bmt_k}\Psi_{\bmo}(\bmt)}{\Pr(\bmo)}
\Tr \left[P_{\bmo}(O)\rho\right]
\right)^2
\label{lem_spp_grad_unbiased_variance_2_3}\\
\leq{}&
\sum_{\bmo\in\Omega}
\frac{\left(\partial_{\bmt_k}\Psi_{\bmo}(\bmt)\right)^2}{\Pr(\bmo)},
\label{lem_spp_grad_unbiased_variance_2_4}
\end{align}
where Eq.~(\ref{lem_spp_grad_unbiased_variance_2_1}) follows from the definition of statistical variance, Eq.~(\ref{lem_spp_grad_unbiased_variance_2_3}) follows from the definition of the single-sample gradient estimator in Eq.~(\ref{app_def_single_grad_estimator}), and Eq.~(\ref{lem_spp_grad_unbiased_variance_2_4}) is derived by noticing
\begin{equation}
\left|\Tr\left[P_{\bmo}(O)\rho\right]\right|\le 1.
\end{equation}

We now bound the sum on the right-hand side of Eq.~(\ref{lem_spp_grad_unbiased_variance_2_4}). Using the product decomposition from Eq.~(\ref{dypp_pp_eq_psi_app}),
\begin{equation}\label{app_grad_psi}
\Psi_{\bmo}(\bmt)=\prod_{j=1}^{L}\Psi_{\bmo_j}(\bmt_j),
\quad
\Psi_{\bmo_j}(\bmt_j)
=
(\cos \bmt_j)^{\mathbb{I}[\bmo_j=1]}
(\sin \bmt_j)^{\mathbb{I}[\bmo_j=-1]},
\end{equation}
we obtain
\begin{equation}
\partial_{\bmt_k}\Psi_{\bmo}(\bmt)
=
\left(\prod_{j\neq k}\Psi_{\bmo_j}(\bmt_j)\right)
\partial_{\bmt_k}\Psi_{\bmo_k}(\bmt_k).
\end{equation}
Together with Eq.~(\ref{app_spp_sampling_prob_factorization}), this yields
\begin{equation}\label{app_grad_factorization}
\frac{\left(\partial_{\bmt_k}\Psi_{\bmo}(\bmt)\right)^2}{\Pr(\bmo)}
=
\left(
\prod_{j\neq k}
\frac{\Psi_{\bmo_j}(\bmt_j)^2}
{\Pr_j(\bmo_j\mid \bmo_{j-1},\dots,\bmo_1)}
\right)
\frac{\left(\partial_{\bmt_k}\Psi_{\bmo_k}(\bmt_k)\right)^2}
{\Pr_k(\bmo_k\mid \bmo_{k-1},\dots,\bmo_1)}.
\end{equation}

Let $u_j=|\cos \bmt_j|+|\sin \bmt_j|$. Then
\begin{equation}
u_j^2=1+|\sin 2\bmt_j|.
\end{equation}
We first bound the contribution from the distinguished layer $k$. If the currently propagated Pauli operator commutes with the $k$-th parametrized gate, then only the branch $\bmo_k=0$ is present. Since $\Psi_0(\bmt_k)=1$ is independent of $\bmt_k$, we have
\begin{equation}\label{lem_spp_grad_unbiased_variance_3_0}
\sum_{\bmo_k\in\{0\}}
\frac{\left(\partial_{\bmt_k}\Psi_{\bmo_k}(\bmt_k)\right)^2}
{\Pr_k(\bmo_k\mid \bmo_{k-1},\dots,\bmo_1)}
=0.
\end{equation}
If it anti-commutes, then only the branches $\bmo_k=\pm 1$ contribute, and
\begin{align}
\sum_{\bmo_k\in\{\pm 1\}}
\frac{\left(\partial_{\bmt_k}\Psi_{\bmo_k}(\bmt_k)\right)^2}
{\Pr_k(\bmo_k\mid \bmo_{k-1},\dots,\bmo_1)}
={}&
\frac{\sin^2\bmt_k}{q(\bmt_k)}
+
\frac{\cos^2\bmt_k}{1-q(\bmt_k)}
\label{lem_spp_grad_unbiased_variance_3_1}\\
\leq{}&
\frac{|\sin\bmt_k|}{a}(u_k+2a)
+
\frac{|\cos\bmt_k|}{a}(u_k+2a)
\label{lem_spp_grad_unbiased_variance_3_2}\\
={}&
\frac{u_k(u_k+2a)}{a}
\label{lem_spp_grad_unbiased_variance_3_3}\\
\leq{}&
\frac{(1+2a)u_k^2}{a}.
\label{lem_spp_grad_unbiased_variance_3_4}
\end{align}
Here Eq.~(\ref{lem_spp_grad_unbiased_variance_3_1}) follows from Eqs.~(\ref{app_spp_sampling_rule_local}) and (\ref{app_grad_psi}), while Eq.~(\ref{lem_spp_grad_unbiased_variance_3_2}) follows from
\begin{equation}
q(\bmt_k)
=
\frac{|\cos\bmt_k|+a}{u_k+2a}
\ge
\frac{a}{u_k+2a},
\qquad
1-q(\bmt_k)
=
\frac{|\sin\bmt_k|+a}{u_k+2a}
\ge
\frac{a}{u_k+2a},
\end{equation}
which follow directly from Eq.~(\ref{app_spp_q_def}).

Next we bound the contribution from every other layer $j\neq k$. If the currently propagated Pauli operator commutes with the $j$-th parametrized gate, then
\begin{equation}\label{lem_spp_grad_unbiased_variance_4_0}
\sum_{\bmo_j\in\{0\}}
\frac{\Psi_{\bmo_j}(\bmt_j)^2}
{\Pr_j(\bmo_j\mid \bmo_{j-1},\dots,\bmo_1)}
=1.
\end{equation}
If it anti-commutes, then by the sampling rule in Eq.~\eqref{app_spp_q_def},
\begin{align}
\sum_{\bmo_j\in\{\pm 1\}}
\frac{\Psi_{\bmo_j}(\bmt_j)^2}
{\Pr_j(\bmo_j\mid \bmo_{j-1},\dots,\bmo_1)}
&\le
(1+2a)u_j^2,
\label{lem_spp_grad_unbiased_variance_4_1}
\end{align}
which follows from
\begin{equation}
\frac{\cos^2\bmt_j}{q(\bmt_j)}+\frac{\sin^2\bmt_j}{1-q(\bmt_j)}
\le |\cos\bmt_j|(u_j+2a)+|\sin\bmt_j|(u_j+2a)
\le (1+2a)u_j^2,
\end{equation}
where $u_j=|\cos\bmt_j|+|\sin\bmt_j|$ and $u_j\ge 1$.

Applying Eqs.~(\ref{lem_spp_grad_unbiased_variance_3_0}), (\ref{lem_spp_grad_unbiased_variance_3_4}), (\ref{lem_spp_grad_unbiased_variance_4_0}), and (\ref{lem_spp_grad_unbiased_variance_4_1}) to Eq.~(\ref{app_grad_factorization}) layer by layer yields
\begin{align}
\sum_{\bmo\in\Omega}
\frac{\left(\partial_{\bmt_k}\Psi_{\bmo}(\bmt)\right)^2}{\Pr(\bmo)}
\leq{}&
\frac{1}{a}\prod_{j=1}^{L}(1+2a)u_j^2
\notag\\
={}&
\frac{1}{a}(1+2a)^L
\prod_{j=1}^{L}\left(1+|\sin 2\bmt_j|\right)
\notag\\
={}&
\frac{1}{a}(1+2a)^L\kappa(\bmt).
\label{app_grad_iterated_bound}
\end{align}
Combining Eqs.~(\ref{lem_spp_grad_unbiased_variance_2_4}) and (\ref{app_grad_iterated_bound}) proves Eq.~(\ref{app_grad_variance_bound}).
This completes the proof.
\end{proof}

As an immediate corollary, since $\tilde{\bm g}(\bmt)$ is the average of $B$ i.i.d.\
copies of $\tilde{\bm g}_{\bmo}(\bmt)$, we have, for every $k\in[L]$,
\begin{equation}\label{app_grad_avg_variance}
\Var_{\bmo^{(1)}, \cdots, \bmo^{(B)} \sim \Pr}\left[\tilde{\bm g}_k(\bmt)\right]
=
\frac{1}{B}\Var_{\bmo\sim \Pr}\left[\tilde{\bm g}_{\bmo,k}(\bmt)\right]
\le
\frac{1}{aB}(1+2a)^L\kappa(\bmt).
\end{equation}

\subsection{High-probability bounds for gradient estimation}
\label{app_spp_estimator_concentration}

We now combine the variance bound established in Section~\ref{app_spp_estimator_grad} with Bernstein's inequality to obtain a high-probability error bound for the empirical gradient estimator. We first recall the Bernstein inequality used below and then prove the formal gradient-estimation statement of Theorem~\ref{thm_spp_estimator}.

\begin{lemma}[Bernstein's inequality]\label{lem_bernstein}
Let $X_1,\dots,X_B$ be independent zero-mean random variables. Suppose that $|X_i| \leq M$ for all $i$ and $ \sum_{i=1}^{B} \Var[X_i] \leq v$. Then, for any $t>0$,
\begin{equation}\label{app_bernstein_tail}
\Pr\left(
\bigg|\sum_{i=1}^{B}X_i\bigg|\ge t
\right)
\le
2\exp\left(
-\frac{\frac{1}{2} t^2}{ v +\frac{Mt}{3} }
\right).
\end{equation}
\end{lemma}

\begin{theorem}[Formal statement of Theorem~\ref{thm_spp_estimator}]\label{thm_spp_estimator_grad}
Consider the value function $f(\bmt)$ defined in Eq.~(\ref{app_spp_value_path_expansion}) and let $\kappa(\bmt)=\prod_{j=1}^{L}\left(1+\left| \sin 2\bmt_j \right|\right)$. Then the gradient estimation $\tilde{\bm g}(\bmt)$ in Eq.~(\ref{lem_spp_value_tildeg}) is unbiased. 
Moreover, for any $\delta\in(0,1)$ and $\epsilon \in (0, \sqrt{L\kappa(\bmt)})$, 
\begin{equation}
B = \left\lceil \frac{4L\kappa(\bmt)}{a\epsilon^2}(1+2a)^L\ln\frac{2L}{\delta} \right\rceil
\end{equation}
independent samples suffice to guarantee that, with probability at least $1-\delta$,
\begin{equation}\label{app_grad_concentration_full}
\big\| \tilde{\bm{g}}(\bmt)- \nabla_{\bmt} f (\bmt) \big\|_2 \leq \epsilon.
\end{equation}
\end{theorem}

\begin{proof}

The unbiasedness of $\tilde{\bm g}(\bmt)$ follows from Lemma~\ref{lem_spp_grad_unbiased_variance}. Next, we prove the gradient estimation error bound in Eq.~(\ref{app_grad_concentration_full}). 

We prove the gradient estimation error bound in Eq.~\eqref{app_grad_concentration_full} by applying Bernstein's inequality in Lemma~\ref{lem_bernstein} to each coordinate of $\tilde{\bm g}(\bmt)$.
Here, the variance term has already been bounded in Eq.~(\ref{app_grad_avg_variance}), so it remains to establish an upper bound for the absolute value of each single-sample gradient estimator $\tilde{\bm g}_{\bmo,k}(\bmt)$. We do so by combining the factorized form of the gradient coefficient with the sequential decomposition of the sampling probability, and then bounding the resulting expression layer by layer. 

Since $P_{\bmo}(O)$ is a Pauli operator and $\rho$ is a quantum state, we have
\begin{equation}\label{app_grad_trace_bound}
\left|\Tr\left[P_{\bmo}(O)\rho\right]\right|\le 1.
\end{equation}
Using Eq.~\eqref{app_grad_trace_bound} and the formulation of the single-sample gradient estimator in Eq.~(\ref{app_def_single_grad_estimator}), the absolute value of each single-sample gradient estimator is bounded as
\begin{equation}\label{app_grad_concentration_pf_2}
\left|\tilde{\bm g}_{\bmo,k}(\bmt)\right|
\le
\frac{\left|\partial_{\bmt_k}\Psi_{\bmo}(\bmt)\right|}{\Pr(\bmo)}.
\end{equation}
Since the gradient coefficient admits the factorized form in Eq.~(\ref{app_grad_factorization}), Eq.~(\ref{app_grad_concentration_pf_2}) can be further formulated as
\begin{equation}\label{app_grad_concentration_pf_3}
\frac{\left|\partial_{\bmt_k}\Psi_{\bmo}(\bmt)\right|}{\Pr(\bmo)}
=
\left(
\prod_{j\neq k}
\frac{|\Psi_{\bmo_j}(\bmt_j)|}{\Pr_j(\bmo_j\mid \bmo_{j-1},\dots,\bmo_1)}
\right)
\frac{\left|\partial_{\bmt_k}\Psi_{\bmo_k}(\bmt_k)\right|}{\Pr_k(\bmo_k\mid \bmo_{k-1},\dots,\bmo_1)}.
\end{equation}

We now bound each factor in Eq.~(\ref{app_grad_concentration_pf_3}) layer by layer. 
For every layer $j\neq k$, if the propagated Pauli operator commutes with the $j$-th parametrized gate, then the only branch is $\bmo_j=0$ and the factor equals one. If it anti-commutes, then
\begin{equation}
\frac{|\Psi_{1}(\bmt_j)|}{q(\bmt_j)}=\frac{|\cos\bmt_j|}{q(\bmt_j)}\le u_j+2a,\qquad
\frac{|\Psi_{-1}(\bmt_j)|}{1-q(\bmt_j)}=\frac{|\sin\bmt_j|}{1-q(\bmt_j)}\le u_j+2a.
\end{equation}
Therefore,
\begin{equation}\label{app_grad_concentration_pf_4}
\frac{|\Psi_{\bmo_j}(\bmt_j)|}{\Pr_j(\bmo_j\mid \bmo_{j-1},\dots,\bmo_1)}\le u_j+2a .
\end{equation}

At the distinguished layer $k$, if the propagated Pauli operator commutes with the $k$-th parametrized gate, then only the branch $\bmo_k=0$ is allowed, and since $\Psi_0(\bmt_k)=1$ is independent of $\bmt_k$, the corresponding derivative vanishes. If it anti-commutes, then Eq.~(\ref{app_spp_q_def}) gives
\begin{align}
\frac{\left|\partial_{\bmt_k}\Psi_{1}(\bmt_k)\right|}{q(\bmt_k)}
={}&
\frac{|\sin\bmt_k|}{q(\bmt_k)}
=
|\sin\bmt_k|\frac{u_k+2a}{|\cos\bmt_k|+a}
\le
\frac{u_k+2a}{a},
\label{app_grad_concentration_pf_5}
\\
\frac{\left|\partial_{\bmt_k}\Psi_{-1}(\bmt_k)\right|}{1-q(\bmt_k)}
={}&
\frac{|\cos\bmt_k|}{1-q(\bmt_k)}
=
|\cos\bmt_k|\frac{u_k+2a}{|\sin\bmt_k|+a}
\le
\frac{u_k+2a}{a},
\label{app_grad_concentration_pf_6}
\end{align}
where the last inequalities follow from $|\cos\bmt_k|+a\ge a$ and $|\sin\bmt_k|+a\ge a$, respectively. Combining Eq.~(\ref{app_grad_concentration_pf_2})--(\ref{app_grad_concentration_pf_6}) yields
\begin{equation}\label{app_grad_concentration_pf_7}
\left|\tilde{\bm g}_{\bmo,k}(\bmt)\right|
\le
\frac{1}{a}\prod_{j=1}^{L}(u_j+2a).
\end{equation}
Since $u_j=|\cos\bmt_j|+|\sin\bmt_j|=\sqrt{1+|\sin 2\bmt_j|}\ge 1$, we have $u_j+2a\le (1+2a)u_j$. Hence
\begin{equation}\label{app_grad_concentration_pf_9}
\left|\tilde{\bm g}_{\bmo,k}(\bmt)\right|
\le
\frac{(1+2a)^L}{a}\prod_{j=1}^{L}u_j = \frac{(1+2a)^L}{a}\sqrt{\kappa(\bmt)}.
\end{equation}

Next, for each fixed $k\in[L]$, we denote the zero-mean random variables
\begin{equation}\label{app_grad_concentration_pf_10}
X_{b,k}:=\tilde{\bm g}_{\bmo^{(b)},k}(\bmt)-\partial_{\bmt_k}f(\bmt).
\end{equation}
Here, the exact gradient is bounded by the parameter-shift rule:
\begin{equation}\label{app_grad_concentration_pf_11}
\left|\partial_{\bmt_k}f(\bmt)\right|
=
\frac{1}{2}\left|f(\bmt_+)-f(\bmt_-)\right|
\le
\frac{1}{2}\left|f(\bmt_+)\right|
+
\frac{1}{2}\left|f(\bmt_-)\right|
\le
1,
\end{equation}
where $\bmt_+$ and $\bmt_-$ are obtained from $\bmt$ by shifting its $k$-th entry by $+\frac{\pi}{2}$ and $-\frac{\pi}{2}$, respectively. Combining Eqs.~(\ref{app_grad_concentration_pf_9}) and (\ref{app_grad_concentration_pf_11}) gives
\begin{equation}\label{app_grad_concentration_pf_12}
|X_{b,k}|
\le
\left|\tilde{\bm g}_{\bmo^{(b)},k}(\bmt)\right|
+
\left|\partial_{\bmt_k}f(\bmt)\right|
\le
\frac{2(1+2a)^L}{a}\sqrt{\kappa(\bmt)}.
\end{equation}
On the other hand, by using Eq.~(\ref{app_grad_variance_bound}), we have
\begin{equation}\label{app_grad_concentration_pf_13}
\sum_{b=1}^{B}\E_{\bmo^{(b)}\sim\Pr}[X_{b,k}^2]
=
\sum_{b=1}^{B}\Var_{\bmo^{(b)}\sim\Pr}\left[\tilde{\bm g}_{\bmo^{(b)},k}(\bmt)\right]
\le
\frac{B}{a}(1+2a)^L\kappa(\bmt).
\end{equation}

We now apply Lemma~\ref{lem_bernstein} to the variables $X_{1,k},\dots,X_{B,k}$. Since
\begin{equation}\label{app_grad_concentration_pf_14}
\frac{1}{B}\sum_{b=1}^{B}X_{b,k}
=
\frac{1}{B}\sum_{b=1}^{B}\tilde{\bm g}_{\bmo^{(b)},k}(\bmt)-\partial_{\bmt_k}f(\bmt)
=
\tilde{\bm g}_k(\bmt)-\partial_{\bmt_k}f(\bmt),
\end{equation}
one may take
\begin{equation}\label{app_grad_concentration_pf_15}
M=\frac{2(1+2a)^L}{a}\sqrt{\kappa(\bmt)},
\qquad
v=\frac{B}{a}(1+2a)^L\kappa(\bmt),
\qquad
t=\frac{B\epsilon}{\sqrt{L}}
\end{equation}
in Lemma~\ref{lem_bernstein}, which yields
\begin{align}
\Pr\left(
\left|\tilde{\bm g}_k(\bmt)-\partial_{\bmt_k}f(\bmt)\right|\ge \frac{\epsilon}{\sqrt{L}}
\right)
&=
\Pr\left(
\left|\sum_{b=1}^{B}X_{b,k}\right|\ge \frac{B\epsilon}{\sqrt{L}}
\right)
\notag\\
&\le
2\exp\left(
-\frac{\frac{1}{2}B^2\epsilon^2/L}{
\frac{B}{a}(1+2a)^L\kappa(\bmt)
+
\frac{2(1+2a)^L}{3a}\sqrt{\kappa(\bmt)}\cdot \frac{B\epsilon}{\sqrt{L}}
}
\right)
\notag\\
&=
2\exp\left(
-\frac{aB\epsilon^2}{
2L(1+2a)^L\kappa(\bmt)
+\frac{4}{3}(1+2a)^L\epsilon\sqrt{L\kappa(\bmt)}
}
\right).
\label{app_grad_concentration_pf_16}
\end{align}
After some straightforward algebra, when $\epsilon \le \sqrt{L\kappa(\bmt)}$, it suffices to choose
\begin{equation}\label{app_grad_concentration_pf_17}
B = \left\lceil
\frac{4L\kappa(\bmt)}{a\epsilon^2}(1+2a)^L\ln\frac{2L}{\delta}
\right\rceil,
\end{equation}
which guarantees, for every $k\in[L]$,
\begin{equation}\label{app_grad_concentration_pf_18}
\Pr\left(
\left|\tilde{\bm g}_k(\bmt)-\partial_{\bmt_k}f(\bmt)\right|\ge \frac{\epsilon}{\sqrt{L}}
\right)
\le
\frac{\delta}{L}.
\end{equation}

Finally, if
\begin{equation}\label{app_grad_concentration_pf_19}
\left|\tilde{\bm g}_k(\bmt)-\partial_{\bmt_k}f(\bmt)\right|
\le
\frac{\epsilon}{\sqrt{L}}
\qquad
\text{for all } k\in[L],
\end{equation}
then
\begin{equation}\label{app_grad_concentration_pf_20}
\left\|\tilde{\bm g}(\bmt)-\nabla_{\bmt}f(\bmt)\right\|_2^2
=
\sum_{k=1}^{L}\left|\tilde{\bm g}_k(\bmt)-\partial_{\bmt_k}f(\bmt)\right|^2
\le
L\cdot \frac{\epsilon^2}{L}
=
\epsilon^2,
\end{equation}
which implies
\begin{equation}\label{app_grad_concentration_pf_21}
\left\|\tilde{\bm g}(\bmt)-\nabla_{\bmt}f(\bmt)\right\|_2
\le
\epsilon.
\end{equation}
Therefore, by a union bound over all $k\in[L]$ and using Eq.~(\ref{app_grad_concentration_pf_18}), we obtain
\begin{equation}\label{app_grad_concentration_pf_22}
\Pr\left(
\left\|\tilde{\bm g}(\bmt)-\nabla_{\bmt}f(\bmt)\right\|_2
\ge
\epsilon
\right)
\le
\sum_{k=1}^{L}
\Pr\left(
\left|\tilde{\bm g}_k(\bmt)-\partial_{\bmt_k}f(\bmt)\right|
\ge
\frac{\epsilon}{\sqrt{L}}
\right)
\le
\delta.
\end{equation}
This completes the proof.

\end{proof}

\section{Convergence guarantee for \texttt{SPPS}-driven gradient descent optimization (Proof of Corollary~\ref{cor_spp_sgd})}
\label{app_spp_sgd}

In this appendix, we establish the convergence guarantee for stochastic gradient descent (SGD) applied to the quantum objective function
\begin{equation}\label{app_spp_sgd_quantum_obj}
f(\bmt)=\Tr\left[O\,U(\bmt)\rho U(\bmt)^\dagger\right]
\end{equation}
for arbitrary quantum observable $O$ and the input state $\rho$, where the gradient is estimated by \texttt{SPPS}. The proof is divided into three steps. We first derive a gradient-variance-based first-order convergence guarantee for stochastic gradient descent (SGD) applied to Eq.~(\ref{app_spp_sgd_quantum_obj}) in Section~\ref{app_spp_sgd_generic}. We then bound the variance of the full \texttt{SPPS} gradient estimator for arbitrary observables under a general Pauli-basis decomposition in Section~\ref{app_spp_sgd_variance}. Finally, we combine these two ingredients to obtain an explicit convergence result in Section~\ref{app_spp_sgd_corollary}, from which Corollary~\ref{cor_spp_sgd} in the main text follows as an informal statement.

\subsection{A variance-based first-order convergence guarantee for quantum optimization}
\label{app_spp_sgd_generic}

We first derive a generic first-order convergence guarantee for SGD when optimizing the objective function in Eq.~(\ref{app_spp_sgd_quantum_obj}). This result is stated in terms of a variance bound on the stochastic gradient, and relies on the global smoothness of the objective function, which we prove below.

\begin{lemma}
\label{lem_spp_sgd_smoothness}
Let the quantum circuit in Eq.~(\ref{app_spp_sgd_quantum_obj}) be $U(\bmt)= V_0 \prod_{j=1}^{L} R_{j}(\bmt_j)V_j$, where $R_{j}(\bmt_j)=\exp\left(-i\bmt_j G_j/2\right)$ with $G_j\in\mathcal P_n$, and where each $V_j$ is a fixed unitary. Then $f(\bmt)$ is globally $L\|O\|_2$-smooth. Equivalently, for any $\bmt,\bmt'\in\mathbb R^L$,
\begin{equation}\label{app_spp_sgd_generic_smooth}
f(\bmt') \le f(\bmt) + (\bmt'-\bmt)^T\nabla_{\bmt}f(\bmt) + \frac{L\|O\|_2}{2}\|\bmt'-\bmt\|_2^2.
\end{equation}
\end{lemma}

\begin{proof}

We begin with a uniform bound on the objective value.
By Eq.~(\ref{app_spp_sgd_quantum_obj}),
\begin{equation}\label{app_spp_sgd_smooth_pf_3}
f(\bmt)=\Tr\left[O\,U(\bmt)\rho U(\bmt)^\dagger\right].
\end{equation}
Since $U(\bmt)\rho U(\bmt)^\dagger$ is a density matrix for every $\bmt$, Hölder's inequality gives
\begin{equation}\label{app_spp_sgd_smooth_pf_4}
|f(\bmt)|
=
\left|\Tr\left[O U(\bmt)\rho U(\bmt)^\dagger\right]\right| \leq
\|O\|_2.
\end{equation}

Next we bound the second-order derivatives of $f(\bmt)$. Since each parametrized gate is generated by a Pauli operator, the parameter-shift rule applies to every circuit parameter. Therefore, for any $j\in[L]$,
\begin{equation}\label{app_spp_sgd_smooth_pf_1}
\partial_{\bmt_j}f(\bmt)
=
\frac{1}{2}
f\!\left(\bmt+\frac{\pi}{2}\bm e_j\right)
-
\frac{1}{2} f\!\left(\bmt-\frac{\pi}{2}\bm e_j\right)
,
\end{equation}
where $\bm e_j$ is the $j$-th standard basis vector. Applying the same rule once more to Eq.~(\ref{app_spp_sgd_smooth_pf_1})  with respect to $\bmt_k$ yields
\begin{align}
\partial_{\bmt_j}\partial_{\bmt_k}f(\bmt)
={}&
\frac{1}{4}
f\!\left(\bmt+\frac{\pi}{2}\bm e_j+\frac{\pi}{2}\bm e_k\right)
-
\frac{1}{4} f\!\left(\bmt+\frac{\pi}{2}\bm e_j-\frac{\pi}{2}\bm e_k\right)
\notag \\
-{}& \frac{1}{4} f\!\left(\bmt-\frac{\pi}{2}\bm e_j+\frac{\pi}{2}\bm e_k\right)
+
\frac{1}{4} f\!\left(\bmt-\frac{\pi}{2}\bm e_j-\frac{\pi}{2}\bm e_k\right) .
\label{app_spp_sgd_smooth_pf_2}
\end{align}
Combining Eqs.~(\ref{app_spp_sgd_smooth_pf_4}) and (\ref{app_spp_sgd_smooth_pf_2}), we obtain
\begin{equation}\label{app_spp_sgd_smooth_pf_5}
\left|
\partial_{\bmt_j}\partial_{\bmt_k}f(\bmt)
\right|
\le
\|O\|_2
\qquad
\text{for all } j,k\in[L].
\end{equation}

We now bound the Hessian operator norm. Eq.~(\ref{app_spp_sgd_smooth_pf_5}) implies
\begin{equation}\label{app_spp_sgd_smooth_pf_6}
\|\nabla_{\bmt}^2 f(\bmt)\|_{F}^2
=
\sum_{j,k=1}^{L}
\left(\partial_{\bmt_j}\partial_{\bmt_k}f(\bmt)\right)^2
\le
L^2\|O\|_2^2,
\end{equation}
and hence
\begin{equation}\label{app_spp_sgd_smooth_pf_7}
\|\nabla_{\bmt}^2 f(\bmt)\|_2
\le
\|\nabla_{\bmt}^2 f(\bmt)\|_{F}
\le
L\|O\|_2.
\end{equation}
Therefore, $f(\bmt)$ is globally $L\|O\|_2$-smooth. The inequality in Eq.~(\ref{app_spp_sgd_generic_smooth}) then follows from the second-order Taylor expansion for smooth functions.

\end{proof}

We now use the smoothness property established in Lemma~\ref{lem_spp_sgd_smoothness} to derive a generic first-order convergence guarantee for SGD when optimizing the quantum objective function in Eq.~(\ref{app_spp_sgd_quantum_obj}). Specifically, we consider the iterates
\begin{equation}\label{app_spp_sgd_generic_update}
\bmt^{(t+1)}=\bmt^{(t)}-\eta^{(t)} \tilde{\bm g}^{(t)},
\qquad
t=1,\dots,T,
\end{equation}
where $\eta^{(t)}$ is the learning rate at the $t$-th step, and $\tilde{\bm g}^{(t)}$ is an unbiased estimator of $\nabla_{\bmt}f(\bmt^{(t)})$, namely
\begin{equation}\label{app_spp_sgd_generic_unbiased}
\E \big[\tilde{\bm g}^{(t)} \big]
=
\nabla_{\bmt}f(\bmt^{(t)}).
\end{equation}
The following lemma shows that the convergence rate of the GD optimization in Eq.~(\ref{app_spp_sgd_generic_update}) is controlled by the variance of the gradient noise.

\begin{lemma}\label{lem_spp_sgd_generic}
Suppose that
\begin{equation}\label{app_spp_sgd_generic_var}
\E\left[
\big\|\tilde{\bm g}^{(t)}-\nabla_{\bmt}f(\bmt^{(t)})\big\|_2^2
\right]
\le
G_T^2
\end{equation}
for all $t\in[T]$ and $T \geq \frac{4L \|O\|_2^2}{G_T^2}$. If the learning rate is chosen as a constant $\eta^{(t)}=\eta= \sqrt{\frac{4}{L G_T^2 T}}$ for all $t\in[T]$, then the GD iterates in Eq.~(\ref{app_spp_sgd_generic_update}) satisfy
\begin{equation}\label{app_spp_sgd_generic_result_balanced}
\min_{t\in[T]}
\E\left\|\nabla_{\bmt}f(\bmt^{(t)})\right\|_2^2
\le
4 \|O\|_2 G_T \sqrt{\frac{L}{T}} .
\end{equation}
\end{lemma}

\begin{proof}
Applying Lemma~\ref{lem_spp_sgd_smoothness} with $\bmt=\bmt^{(t)}$ and $\bmt'=\bmt^{(t+1)}$ gives
\begin{align}
f(\bmt^{(t+1)})
\le{}&
f(\bmt^{(t)})
+
(\bmt^{(t+1)}-\bmt^{(t)})^T\nabla_{\bmt}f(\bmt^{(t)})
+
\frac{L\|O\|_2}{2}\|\bmt^{(t+1)}-\bmt^{(t)}\|_2^2.
\label{app_spp_sgd_generic_pf_1}
\end{align}
Substituting the update rule in Eq.~(\ref{app_spp_sgd_generic_update}) into Eq.~(\ref{app_spp_sgd_generic_pf_1}) yields
\begin{align}
f(\bmt^{(t+1)})
\le{}&
f(\bmt^{(t)})
-
\eta\left\langle \tilde{\bm g}^{(t)},\nabla_{\bmt}f(\bmt^{(t)})\right\rangle
+
\frac{L\|O\|_2}{2}\eta^2\|\tilde{\bm g}^{(t)}\|_2^2.
\label{app_spp_sgd_generic_pf_2}
\end{align}
Taking expectation on both sides of Eq.~(\ref{app_spp_sgd_generic_pf_2}) and using Eq.~(\ref{app_spp_sgd_generic_unbiased}), we obtain
\begin{align}
\E\left[f(\bmt^{(t+1)})\right]
\le{}&
\E\left[f(\bmt^{(t)})\right]
-
\eta\E\left[\left\|\nabla_{\bmt}f(\bmt^{(t)})\right\|_2^2\right]
+
\frac{L\|O\|_2}{2}\eta^2\E\left[\|\tilde{\bm g}^{(t)}\|_2^2\right].
\label{app_spp_sgd_generic_pf_3}
\end{align}

Next, by expanding the second moment around the mean, we have
\begin{align}
\E\left[\|\tilde{\bm g}^{(t)}\|_2^2\right]
={}&
\E\left[
\left\|
\tilde{\bm g}^{(t)}-\nabla_{\bmt}f(\bmt^{(t)})+\nabla_{\bmt}f(\bmt^{(t)})
\right\|_2^2
\right]
\notag\\
={}&
\E\left[\left\|\nabla_{\bmt}f(\bmt^{(t)})\right\|_2^2\right]
+
\E\left[
\left\|\tilde{\bm g}^{(t)}-\nabla_{\bmt}f(\bmt^{(t)})\right\|_2^2
\right]
\label{app_spp_sgd_generic_pf_4}\\
\le{}&
\E\left[\left\|\nabla_{\bmt}f(\bmt^{(t)})\right\|_2^2\right]
+
G_T^2,
\notag
\end{align}
where Eq.~(\ref{app_spp_sgd_generic_pf_4}) follows from Eq.~(\ref{app_spp_sgd_generic_unbiased}), and the last inequality follows from Eq.~(\ref{app_spp_sgd_generic_var}). Substituting this bound into Eq.~(\ref{app_spp_sgd_generic_pf_3}) gives
\begin{align}
\E\left[f(\bmt^{(t+1)})\right]
\le{}&
\E\left[f(\bmt^{(t)})\right]
-
\left(\eta-\frac{L\|O\|_2}{2}\eta^2\right)
\E\left[\left\|\nabla_{\bmt}f(\bmt^{(t)})\right\|_2^2\right]
+
\frac{L\|O\|_2}{2}\eta^2 G_T^2.
\label{app_spp_sgd_generic_pf_5}
\end{align}

Next, the condition
\begin{equation}\label{app_spp_sgd_generic_pf_6}
T \ge \frac{4L\|O\|_2^2}{G_T^2}
\end{equation}
together with the choice
\begin{equation}\label{app_spp_sgd_generic_pf_7}
\eta=\sqrt{\frac{4}{L G_T^2 T}}
\end{equation}
implies
\begin{equation}\label{app_spp_sgd_generic_pf_8}
\eta
\le
\frac{1}{L\|O\|_2}.
\end{equation}
Therefore,
\begin{equation}\label{app_spp_sgd_generic_pf_9}
\eta-\frac{L\|O\|_2}{2}\eta^2 = \eta \left( 1 - \frac{L\|O\|_2}{2}\eta \right)
\ge
\frac{\eta}{2}.
\end{equation}
Applying Eq.~(\ref{app_spp_sgd_generic_pf_9}) to Eq.~(\ref{app_spp_sgd_generic_pf_5}) yields
\begin{align}
\E\left[f(\bmt^{(t+1)})\right]
\le{}&
\E\left[f(\bmt^{(t)})\right]
-
\frac{\eta}{2}\E\left[\left\|\nabla_{\bmt}f(\bmt^{(t)})\right\|_2^2\right]
+
\frac{L\|O\|_2}{2}\eta^2 G_T^2.
\label{app_spp_sgd_generic_pf_10}
\end{align}

Summing Eq.~(\ref{app_spp_sgd_generic_pf_10}) over $t=1,\dots,T$ gives
\begin{align}
\frac{\eta}{2}\sum_{t=1}^{T}\E\left[\left\|\nabla_{\bmt}f(\bmt^{(t)})\right\|_2^2\right]
\le{}&
\E\left[f(\bmt^{(1)})-f(\bmt^{(T+1)})\right]
+
\frac{L\|O\|_2}{2}\eta^2 T G_T^2.
\label{app_spp_sgd_generic_pf_11}
\end{align}
On the other hand, by Eq.~(\ref{app_spp_sgd_smooth_pf_4}) in the proof of Lemma~\ref{lem_spp_sgd_smoothness}, we have
\begin{equation}\label{app_spp_sgd_generic_pf_12}
|f(\bmt)|\le \|O\|_2
\qquad
\text{for all } \bmt,
\end{equation}
which implies
\begin{equation}\label{app_spp_sgd_generic_pf_13}
\E\left[f(\bmt^{(1)})-f(\bmt^{(T+1)})\right]
\le
2\|O\|_2.
\end{equation}
Combining Eqs.~(\ref{app_spp_sgd_generic_pf_11}) and (\ref{app_spp_sgd_generic_pf_13}) yields
\begin{equation}\label{app_spp_sgd_generic_pf_14}
\min_{t\in[T]}
\E\left[\left\|\nabla_{\bmt}f(\bmt^{(t)})\right\|_2^2\right]
\le
\frac{4\|O\|_2}{\eta T}
+
\eta L\|O\|_2 G_T^2.
\end{equation}

Finally, substituting Eq.~(\ref{app_spp_sgd_generic_pf_7}) into Eq.~(\ref{app_spp_sgd_generic_pf_14}) gives
\begin{align}
\min_{t\in[T]}
\E\left[\left\|\nabla_{\bmt}f(\bmt^{(t)})\right\|_2^2\right]
\le{}&
\frac{4\|O\|_2}{T}\sqrt{\frac{L G_T^2 T}{4}}
+
L\|O\|_2 G_T^2\sqrt{\frac{4}{L G_T^2 T}}
\notag\\
={}&
2\|O\|_2\sqrt{\frac{L G_T^2}{T}}
+
2\|O\|_2 G_T \sqrt{\frac{L}{T}}
\notag\\
={}&
4 \|O\|_2 G_T \sqrt{\frac{L}{T}}.
\notag
\end{align}
This completes the proof.

\end{proof}

\subsection{Variance of the full \texttt{SPPS} gradient estimator}
\label{app_spp_sgd_variance}

We now bound the variance of the full \texttt{SPPS} gradient estimator for a general observable expressed in the Pauli basis. Specifically, we suppose that $O=\sum_{m=1}^{N_O} c_m O_m$, where each $O_m \in\mathcal P_n$. For each $m\in[N_O]$, define
\begin{equation}\label{app_spp_sgd_variance_fm}
f^{(m)}(\bmt):=\Tr \left[O_m U(\bmt)\rho U(\bmt)^\dagger\right].
\end{equation}
Then, by linearity,
\begin{equation}\label{app_spp_sgd_variance_f_linear}
f(\bmt)=\sum_{m=1}^{N_O} c_m f^{(m)}(\bmt),
\qquad
\nabla_{\bmt}f(\bmt)=\sum_{m=1}^{N_O} c_m \nabla_{\bmt}f^{(m)}(\bmt).
\end{equation}

For each $m\in[N_O]$, let $\tilde{\bm g}^{(m)}(\bmt)$ denote the empirical \texttt{SPPS} gradient estimator for the Pauli observable $O_m$, obtained by averaging $B$ independent \texttt{SPPS} samples as in Lemma~\ref{lem_spp_grad_unbiased_variance}. We remark that, for different $m$, the estimators $\tilde{\bm g}^{(m)}(\bmt)$ are constructed from independent batches of \texttt{SPPS} samples. We define the full \texttt{SPPS} gradient estimator by
\begin{equation}\label{app_spp_sgd_variance_full_grad}
\tilde{\bm g}(\bmt):=\sum_{m=1}^{N_O} c_m \tilde{\bm g}^{(m)}(\bmt).
\end{equation}
Since each $\tilde{\bm g}^{(m)}(\bmt)$ is an unbiased estimator of $\nabla_{\bmt}f^{(m)}(\bmt)$, linearity immediately yields that
\begin{equation}\label{app_spp_sgd_variance_unbiased}
\E \left[\tilde{\bm g}(\bmt)\right]
=
\nabla_{\bmt}f(\bmt).
\end{equation}

\begin{lemma}\label{lem_spp_sgd_full_variance}
For any $\bmt\in\mathbb R^L$, let $\kappa(\bmt)$ be defined as in Lemma~\ref{lem_spp_grad_unbiased_variance}. Then the full \texttt{SPPS} gradient estimator $\tilde{\bm g}(\bmt)$ satisfies
\begin{equation}\label{app_spp_sgd_variance_bound}
\E \left[
\left\|\tilde{\bm g}(\bmt)-\nabla_{\bmt}f(\bmt)\right\|_2^2
\right]
\le
\frac{L\|O\|_2^2}{aB}(1+2a)^L\kappa(\bmt).
\end{equation}
\end{lemma}

\begin{proof}

For each $m\in[N_O]$, we define the estimation error
\[
\bm\Delta^{(m)}(\bmt)
:=
\tilde{\bm g}^{(m)}(\bmt)-\nabla_{\bmt}f^{(m)}(\bmt).
\]
Then, by Eqs.~(\ref{app_spp_sgd_variance_f_linear}) and (\ref{app_spp_sgd_variance_full_grad}),
\begin{align}
\tilde{\bm g}(\bmt)-\nabla_{\bmt}f(\bmt)
&=
\sum_{m=1}^{N_O} c_m\tilde{\bm g}^{(m)}(\bmt)
-
\sum_{m=1}^{N_O} c_m\nabla_{\bmt}f^{(m)}(\bmt) \notag\\
&=
\sum_{m=1}^{N_O} c_m\bm\Delta^{(m)}(\bmt).
\label{app_spp_sgd_variance_pf_1}
\end{align}
Therefore, we have
\begin{align}
\E\left[
\left\|
\tilde{\bm g}(\bmt)-\nabla_{\bmt}f(\bmt)
\right\|_2^2
\right]
&=
\E\left[
\left\|
\sum_{m=1}^{N_O} c_m\bm\Delta^{(m)}(\bmt)
\right\|_2^2
\right] \notag\\
&=
\sum_{m,m'=1}^{N_O}
c_mc^{(m')}
\E\left[
\left\langle
\bm\Delta^{(m)}(\bmt),\bm\Delta^{(m')}(\bmt)
\right\rangle
\right].
\label{app_spp_sgd_variance_pf_2}
\end{align}
Since each $\tilde{\bm g}^{(m)}(\bmt)$ is unbiased, we have
\[
\E\!\left[\bm\Delta^{(m)}(\bmt)\right]=\bm 0
\qquad
\text{for all } m\in[N_O].
\]
Moreover, for $m\neq m'$, the random vectors $\bm\Delta^{(m)}(\bmt)$ and $\bm\Delta^{(m')}(\bmt)$ are independent since the corresponding estimators are constructed from independent batches of \texttt{SPPS} samples. Hence
\begin{align}
\E\left[
\left\langle
\bm\Delta^{(m)}(\bmt),\bm\Delta^{(m')}(\bmt)
\right\rangle
\right]
&=
\left\langle
\E\!\left[\bm\Delta^{(m)}(\bmt)\right],
\E\!\left[\bm\Delta^{(m')}(\bmt)\right]
\right\rangle
=0,
\qquad m\neq m'.
\label{app_spp_sgd_variance_pf_3}
\end{align}
Substituting Eq.~(\ref{app_spp_sgd_variance_pf_3}) into Eq.~(\ref{app_spp_sgd_variance_pf_2}) yields
\begin{align}
\E\left[
\left\|
\tilde{\bm g}(\bmt)-\nabla_{\bmt}f(\bmt)
\right\|_2^2
\right]
&=
\sum_{m=1}^{N_O}(c_m)^2
\E\left[
\left\|
\bm\Delta^{(m)}(\bmt)
\right\|_2^2
\right].
\label{app_spp_sgd_variance_pf_4}
\end{align}
For each $m\in[N_O]$, Lemma~\ref{lem_spp_grad_unbiased_variance} applied to the Pauli observable $O_m$ gives
\begin{equation}\label{app_spp_sgd_variance_pf_5}
\E\left[
\left\|
\bm\Delta^{(m)}(\bmt)
\right\|_2^2
\right]
\le
\E\left[
\left\|
\tilde{\bm{g}}^{(m)}(\bmt)
\right\|_2^2
\right]
\le
\frac{L}{aB}(1+2a)^L\kappa(\bmt).
\end{equation}
Applying Eq.~(\ref{app_spp_sgd_variance_pf_5}) in Eq.~(\ref{app_spp_sgd_variance_pf_4}) yields
\begin{align}
\E\left[
\left\|
\tilde{\bm g}(\bmt)-\nabla_{\bmt}f(\bmt)
\right\|_2^2
\right]
&\le
\frac{L}{aB}(1+2a)^L\kappa(\bmt)
\sum_{m=1}^{N_O}(c_m)^2.
\label{app_spp_sgd_variance_pf_6}
\end{align}

It remains to relate the coefficient norm to $\|O\|_2$. Since the Pauli operators form an orthogonal basis under the Hilbert--Schmidt inner product,
\[
\Tr\!\left(O_m O_{m'}\right)=2^n\delta_{m,m'},
\]
and thus
\begin{align}
\|O\|_F^2
&=
\Tr(O^\dagger O)
=
2^n\sum_{m=1}^{N_O}(c_m)^2.
\label{app_spp_sgd_variance_pf_7}
\end{align}
Hence
\begin{equation}\label{app_spp_sgd_variance_pf_8}
\sum_{m=1}^{N_O}(c_m)^2
=
2^{-n}\|O\|_F^2
\le
2^{-n}\,\mathrm{rank}(O)\,\|O\|_2^2
\le
\|O\|_2^2.
\end{equation}
Combining Eqs.~(\ref{app_spp_sgd_variance_pf_6}) and (\ref{app_spp_sgd_variance_pf_8}) proves Eq.~(\ref{app_spp_sgd_variance_bound}).

\end{proof}

\subsection{Convergence complexity of \texttt{SPPS}-based SGD}
\label{app_spp_sgd_corollary}

We are now ready to combine the generic convergence guarantee in Lemma~\ref{lem_spp_sgd_generic} with the variance bound for the full \texttt{SPPS} gradient estimator in Lemma~\ref{lem_spp_sgd_full_variance} to derive the convergence complexity of \texttt{SPPS}-based SGD. The resulting guarantee is stated explicitly in Theorem~\ref{thm_spp_sgd_explicit}. 

\begin{theorem}[Formal statement of Corollary~\ref{cor_spp_sgd}]\label{thm_spp_sgd_explicit}
We follow the notation in Lemma~\ref{lem_spp_sgd_full_variance}. Consider GD applied to the quantum objective $f(\bmt)$ in Eq.~(\ref{app_spp_sgd_quantum_obj}), where the gradient is given by Eq.~(\ref{app_spp_sgd_variance_full_grad}) for a general observable $O$ with $N_O$ Pauli terms. Let $\bmt^{(t)}\in\mathbb R^L$ denote the parameters at iteration $t$, and define $\kappa_T:=\max_{t\in[T]}\prod_{j=1}^{L}(1+|\sin 2\bmt_j^{(t)}|)$. Then, for any $\epsilon>0$, the guarantee $\min_{t\in[T]}\E\|\nabla_{\bmt}f(\bmt^{(t)})\|_2^2\le \epsilon$ holds within $T=\mathcal{O}(L^3\|O\|_2^4/\epsilon^2)$ iterations, using $\mathcal{O}(N_O L^3\|O\|_2^4\kappa_T/\epsilon^2)$ \texttt{SPPS} samples in total.
\end{theorem}

\begin{proof}
By Lemma~\ref{lem_spp_sgd_full_variance}, the full \texttt{SPPS} gradient estimator in Eq.~(\ref{app_spp_sgd_variance_full_grad}) is unbiased and satisfies
\begin{equation}\label{app_spp_sgd_explicit_pf_1}
\E\left[
\left\|\tilde{\bm g}^{(t)}-\nabla_{\bmt}f(\bmt^{(t)})\right\|_2^2
\right]
\le
\frac{2L\|O\|_2^2}{aB}(1+2a)^L\kappa(\bmt^{(t)})
\le
\frac{2L\|O\|_2^2}{aB}(1+2a)^L\kappa_T
\end{equation}
for all $t\in[T]$, where the last inequality follows from the definition of $\kappa_T$. Choosing $a=\O(1/L)$ and $B= \left\lceil \kappa_T \right\rceil$, Eq.~(\ref{app_spp_sgd_explicit_pf_1}) further yields 
\begin{equation}\label{app_spp_sgd_explicit_pf_1_1}
\E\left[ \left\|\tilde{\bm g}^{(t)}-\nabla_{\bmt}f(\bmt^{(t)})\right\|_2^2 \right] \le
{cL^2 \|O\|_2^2} ,
\end{equation}
where $c>0$ is an absolute constant. Therefore, the variance condition in Lemma~\ref{lem_spp_sgd_generic} holds with
\begin{equation}\label{app_spp_sgd_explicit_pf_2}
G_T^2={cL^2 \|O\|_2^2}.
\end{equation}

Following Eq.~(\ref{app_spp_sgd_explicit_pf_2}), we have $T\ge \frac{4L\|O\|_2^2}{G_T^2}$, which is guaranteed by the condition $T \geq \O(1)$ since 
\begin{equation}
\frac{4L\|O\|_2^2}{G_T^2} = \frac{4L\|O\|_2^2}{{cL^2 \|O\|_2^2} } = \frac{4}{{cL} } \leq \O(1).
\end{equation}
Thus, Lemma~\ref{lem_spp_sgd_generic} implies that, if the learning rate $\eta^{(t)}=\eta=\sqrt{\frac{4}{L G_T^2 T}}$ for all $t\in[T]$, then
\begin{equation}\label{app_spp_sgd_explicit_pf_5}
\min_{t\in[T]}
\E\left\|\nabla_{\bmt}f(\bmt^{(t)})\right\|_2^2
\le
C\sqrt{\frac{L\|O\|_2^2 G_T^2}{T}} = C \sqrt{\frac{cL^3\|O\|_2^4}{T}}
\end{equation}
for an absolute constant $C>0$, where the last equation yields from Eq.~(\ref{app_spp_sgd_explicit_pf_2}).

We now derive a sufficient condition on $T$ for the right-hand side of Eq.~(\ref{app_spp_sgd_explicit_pf_5}) to be at most $\epsilon$. Requiring
\begin{equation}\label{app_spp_sgd_explicit_pf_7}
C \sqrt{\frac{cL^3\|O\|_2^4}{T}} \le \epsilon
\end{equation}
is equivalent to
\begin{equation}\label{app_spp_sgd_explicit_pf_8}
T\ge \frac{c C^2L^3\|O\|_2^4}{\epsilon^2}.
\end{equation}
Therefore, the total \texttt{SPPS} sample count is
\begin{equation}
BN_O \cdot T= \O \left( \frac{ N_O L^3\|O\|_2^4 \kappa_T}{\epsilon^2} \right)
\end{equation}
by considering $B$ samples used for each Pauli operator in the decomposition $O=\sum_{k=1}^{N_O} c_k O_k$ and using $B= \left\lceil \kappa_T \right\rceil$. This completes the proof.

\end{proof}

\section{Experimental settings and additional experimental results}
\label{app_exp}

In this appendix, we provide additional experimental settings and numerical results that complement Sec.~\ref{spp_experiment}. 
We first describe the implementation details of experiments, and then report further results on tensor-network comparisons, sensitivity to the sampling parameter, circuit-depth dependence, and state-encoding circuit preparation. In particular, \texttt{SPPS} uses the implementation described in App.~\ref{app_spp_alg}, including sequential path sampling, importance reweighting, PAD, and the adaptive gradient-error proxy in App.~\ref{app_spp_gradient_error_proxy}. 
For all additional experiments, runtime denotes the cumulative time spent on stochastic gradient estimation and parameter updates.

\subsection{Experimental settings}
\label{app_exp_settings}

\textbf{VQE settings.}
The Hamiltonian is
\begin{equation}
H
=
-J\sum_{i=1}^{n-1} Z_i Z_{i+1}
-
g\sum_{i=1}^{n} X_i ,
\end{equation}
with $J=1.0$ and $g=1.0$ unless otherwise specified. 
The variational circuit is initialized from the plus state by default and uses a one-dimensional hardware-efficient ansatz. 
Each layer applies $R_Z$ and $R_Y$ rotations on all qubits, followed by an open-boundary nearest-neighbor CNOT chain. 
For a circuit with depth $L$, each trainable angle is initialized from $\mathrm{Unif}[-0.25\pi/L,0.25\pi/L]$, following the small-angle initialization strategies~\cite{zhang2022escaping,PhysRevApplied.22.054005,mhiri2025unifying,peng2025titan} designed to mitigate barren plateaus~\cite{McClean2018,Larocca2025}.
For \texttt{SPPS}, $\delta$ denotes the adaptive gradient-error proxy threshold. 
The smoothing value in the branch-sampling distribution is term-dependent and follows the adaptive rule in App.~\ref{app_spp_smoothing_parameter}. 
It is initialized as $a_{\mathrm{init}}=0.25/L$ for a circuit with depth $L$, and then updated using the path-statistics-based rule in Eq.~\eqref{app_spp_a_update}. 
Unless otherwise specified, we use $b_a=0.5$.

\textbf{QNN settings.}
For the QNN pre-training experiment, we use the same dataset, circuit, optimizer, and evaluation metrics as described in Sec.~\ref{spp_experiment}. 
The synthetic labels are generated by a fixed $40$-qubit near-Clifford
circuit $V$, which consists of two random Clifford brick-wall
layers. Each layer applies a single-qubit gate sampled uniformly from
$\{H,S,S^\dagger,X,Y,Z\}$ to every qubit, followed by two staggered
nearest-neighbor entangling sublayers covering the odd and even bonds.
Each two-qubit gate is sampled from $\{\mathrm{CNOT},\mathrm{CZ}\}$,
with the control--target direction of each CNOT sampled uniformly.
Four single-qubit rotations, sampled uniformly from
$\{R_X,R_Y,R_Z\}$, are inserted at random qubits and circuit positions.
Their angles are sampled from $\mathrm{Unif}[-\pi,\pi]$.
Candidate inputs are computational-basis product states $|\psi_i\rangle$ whose bits are
sampled independently from $\mathrm{Bernoulli}(0.5)$.
For each input $|\psi_i\rangle$, the label is
\begin{equation}
y_i
=
\langle\psi_i|V^\dagger O V|\psi_i\rangle,
\qquad
O=\frac{1}{n}\sum_{j=1}^{n}Z_j.
\end{equation}
The labels are evaluated using untruncated Pauli propagation. We retain unique inputs satisfying $|y_i|>0.1$. This filtering prevents the target labels from
concentrating near zero and yields a structured learning task rather than the
sufficiently randomized quantum-data regime, in which training
efficiency and generalization can deteriorate exponentially with the
number of qubits~\cite{zhang2024curse}.

The trainable QNN uses a $4$-layer hardware-efficient ansatz. 
Each layer applies an $R_Z$-$R_Y$-$R_Z$ rotation block to every qubit, followed by a nearest-neighbor CNOT chain. 
Thus, for the $40$-qubit QNN, the total number of trainable parameters is $4\times 40\times 3=480$. 
The circuit parameters are randomly initialized with scale $0.05$. 
For \texttt{SPPS}, we use the fixed-smoothing option in App.~\ref{app_spp_smoothing_parameter}, i.e., the smoothing value is kept as a constant $a=a_0=0.01$ throughout training.

\textbf{Tensor-network baseline.}
For the additional scaling comparison, we compare \texttt{SPPS} against a tensor-network simulator implemented with \texttt{PastaQ.jl}~\cite{pastaq}. 
The comparison is performed on the TFIM VQE benchmark with $n\in\{20,40,60,80,100\}$ qubits and $L=6$ ansatz layers. 
The tensor-network baseline uses a bond dimension $\chi=2$. 
Both methods are evaluated under the same VQE objective, and we report the final normalized energy error and the optimization runtime.

\textbf{State-preparation benchmark.}
We further evaluate \texttt{SPPS} on the preparation of quantum encoding circuits. 
Given a normalized classical vector $\bm{v}\in\mathbb{R}^{2^n}$, the target state is the amplitude-encoded state
\begin{equation}
\ket{v}
=
\sum_{j=0}^{2^n-1} v_j \ket{j}.
\end{equation}
The goal is to optimize a parameterized circuit $U(\bmt)$ such that $U(\bmt)\ket{0}^{\otimes n}$ approximates $\ket{v}$. 
Equivalently, we maximize the fidelity
\begin{equation}
F(\bmt)
=
\left|
\bra{v} U(\bmt)\ket{0}^{\otimes n}
\right|^2
=
\Tr\!\left[
O_v U(\bmt) \rho_0 U(\bmt)^\dag
\right],
\quad
O_v=\ket{v}\!\bra{v},
\quad
\rho_0=\ket{0}\!\bra{0}^{\otimes n},
\end{equation}
or minimize the infidelity
\begin{equation}
\mathcal{I}(\bmt)
=
1-F(\bmt).
\end{equation}
This objective is a direct instance of Eq.~\eqref{dypp_f_loss_eq} with a global projector observable. 
Since $O_v$ is not a single Pauli observable, we expand it in the Pauli basis and keep the largest $100$ non-identity Pauli coefficients by magnitude and include the identity contribution as an offset. 
Each Pauli term is then handled by the same \texttt{SPPS} estimator described in App.~\ref{app_spp_alg}, and the final gradient is obtained by linearly combining term-wise stochastic gradients.

For the reported state-preparation experiment, we use MNIST images as target data. 
Each image is resized to a $4\times4$ grayscale vector, normalized to unit $\ell_2$ norm, and amplitude-encoded into a $4$-qubit target state. 
We optimize a hardware-efficient circuit consisting of repeated $R_ZR_YR_Z$ single-qubit rotation layers and nearest-neighbor CNOT chains. 
We vary the circuit depth as $L\in\{1,2,3,4\}$. The circuit parameters are randomly initialized with scale $0.1$. The optimizer is gradient descent with learning rate $0.5$ for $200$ steps.

\textbf{Hardware environment.}
Unless otherwise specified, experiments were run on a personal computer with an Apple M1 Pro chip. 
The large-scale VQE experiments with $n\in\{20,40,60,80,100\}$ qubits, including the tensor-network scaling comparison, were run on CPU compute nodes equipped with dual-socket AMD EPYC 7713 processors at 2.0GHz. No GPU acceleration is used in the reported runtimes.

\subsection{Additional VQE results}
\label{app_exp_vqe_results}

\textbf{\texttt{SPPS} remains competitive against tensor-network simulation.}
Fig.~\ref{dpp_fig_app_tn} compares \texttt{SPPS} with the tensor-network baseline on the $L=6$ TFIM VQE benchmark. 
Across $n=20$ to $100$ qubits, \texttt{SPPS} attains a favorable error-runtime trade-off. 
The advantage is especially clear at larger system sizes, where the error of tensor-network simulation grows more rapidly. 
This comparison complements Fig.~\ref{dpp_fig3}(c,d) by showing that \texttt{SPPS} is not only competitive against \texttt{Tb-PBS} baselines, but also provides an efficient alternative to tensor-network simulation in this optimization setting.

\begin{figure}[t]
    \centering
    \includegraphics[width=0.7\linewidth]{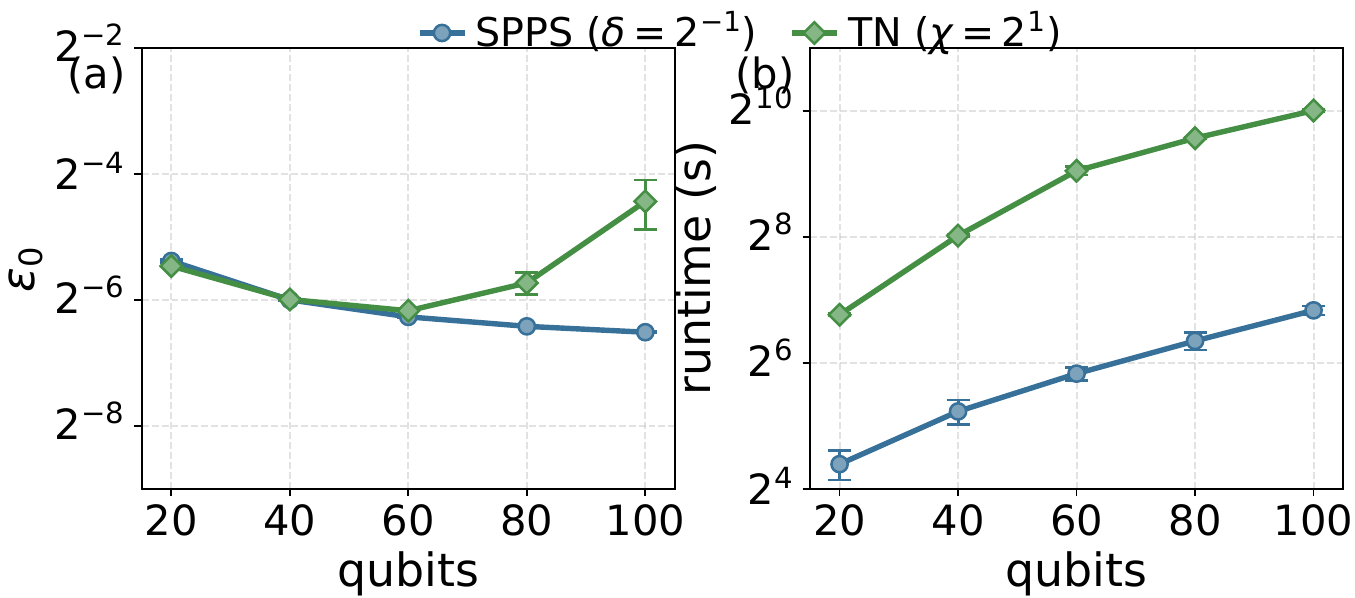}
    \caption{
{
    \textbf{\texttt{SPPS} versus tensor-network simulation on large-qubit VQE pre-training.}
    We compare \texttt{SPPS} and a tensor-network baseline on the $L=6$ TFIM VQE benchmark with $n\in\{20,40,60,80,100\}$ qubits.
    Figures (a) and (b) report the final normalized energy error $\epsilon_0$ and the optimization runtime, respectively.
    Error bars indicate the standard deviation over $10$ independent runs.
    }}
    \label{dpp_fig_app_tn}
\end{figure}

\textbf{Sensitivity to the sampling-smoothing parameter.}
Fig.~\ref{dpp_fig_app_diffa} studies the effect of the smoothing parameter $b_a$ in the \texttt{SPPS} sampling distribution. 
The benchmark uses the $60$-qubit TFIM VQE with ansatz depth $L=8$. 
The results show that \texttt{SPPS} is stable across a broad range of $b_a$ values. 
In particular, moderate values of $b_a$ yield similar final errors, while the runtime varies mildly with the sampling distribution. 
This validates the practical robustness of the smoothed sampling rule introduced in App.~\ref{app_spp_estimator_sample_rule}. 
The result also confirms the role of the smoothing parameter: it should be positive to preserve derivative-sensitive branches, but it does not require delicate fine-tuning in the tested regime.

\begin{figure}[t]
    \centering
    \includegraphics[width=0.7\linewidth]{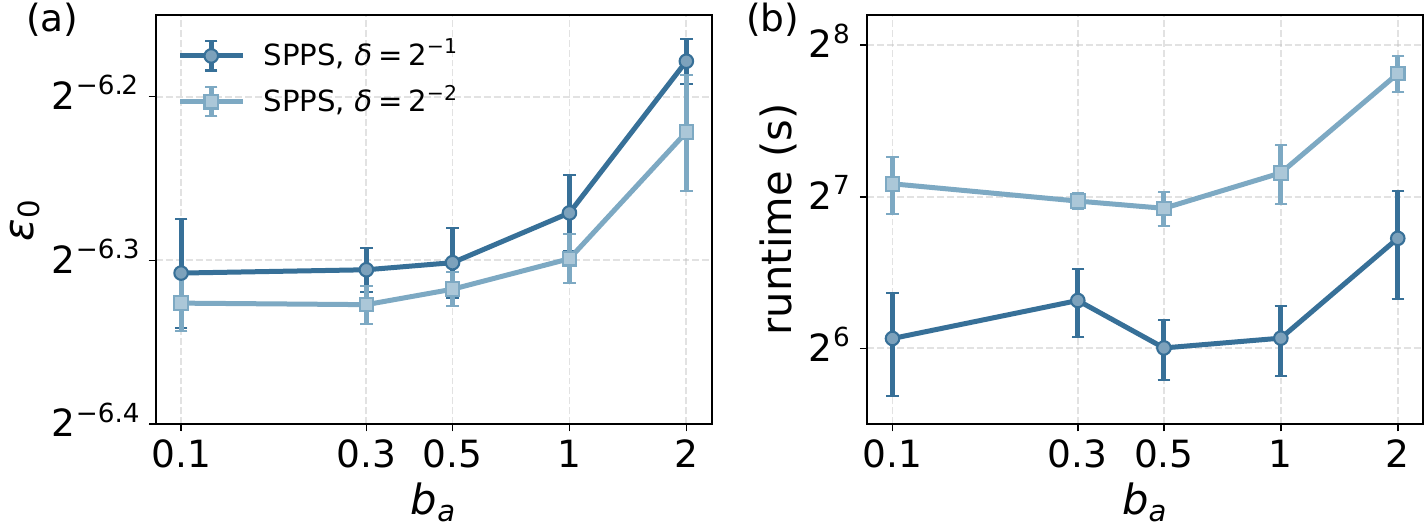}
    \caption{
{
    \textbf{Effect of the sampling-smoothing parameter in \texttt{SPPS}.}
    We evaluate \texttt{SPPS} on the $60$-qubit TFIM VQE benchmark with ansatz depth $L=8$ and vary the smoothing parameter $b_a$ in the path-sampling distribution.
    The two figures show the final normalized energy error $\epsilon_0$ and the optimization runtime.
    The threshold $\delta$ denotes the adaptive gradient-error proxy tolerance described in App.~\ref{app_spp_gradient_error_proxy}.
    Error bars indicate the standard deviation over $5$ independent runs.
    }}
    \label{dpp_fig_app_diffa}
\end{figure}

\textbf{Effect of circuit depth.}
Fig.~\ref{dpp_fig_app_diffL} evaluates \texttt{SPPS} under different ansatz depths for the $60$-qubit TFIM VQE benchmark. 
We fix $b_a=0.5$ and vary $L\in\{4,6,8,10,12\}$. 
The final energy error remains in the same order across different depths, showing that \texttt{SPPS} continues to produce useful stochastic gradients as the circuit becomes deeper. 
Meanwhile, the runtime increases with $L$, which is consistent with the sample-complexity dependence on the path length and the effective branching factor in Theorem~\ref{thm_spp_estimator}. 
This experiment therefore empirically supports the theoretical message that the cost of faithful stochastic Pauli-path simulation is controlled by the trajectory-dependent path complexity.

\begin{figure}[t]
    \centering
    \includegraphics[width=0.7\linewidth]{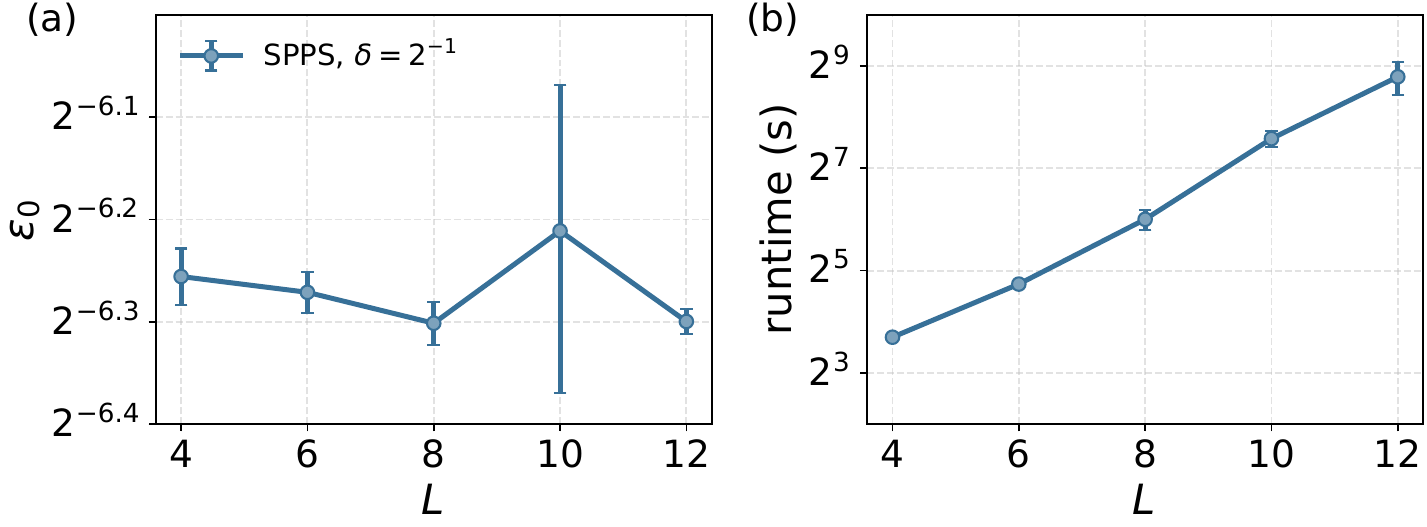}
    \caption{
{
    \textbf{Effect of circuit depth on \texttt{SPPS}-driven VQE pre-training.}
    We evaluate the $60$-qubit TFIM VQE benchmark with fixed sampling-smoothing parameter $b_a=0.5$ and vary the ansatz depth as $L\in\{4,6,8,10,12\}$.
    The figures report the final normalized energy error $\epsilon_0$ and the optimization runtime.
    Error bars indicate the standard deviation over $5$ independent runs.
    }}
    \label{dpp_fig_app_diffL}
\end{figure}

\subsection{Additional state-preparation results}
\label{app_exp_stateprep}

\textbf{\texttt{SPPS} prepares quantum encoding circuits for MNIST amplitude encoding.}
We finally evaluate whether \texttt{SPPS} can be used beyond VQE and QML pre-training by optimizing state-preparation circuits. 
Fig.~\ref{dpp_fig_app_stateprep} summarizes the results. 
Figures (a) and (b) show the endpoint exact infidelity and runtime for depths $L=1,2,3,4$. 
Figure (c) reports the exact infidelity along optimization. 
Increasing $L$ substantially improves the attainable infidelity: shallow circuits with $L=1$ remain far from the target state, whereas depths $L=3$ and $L=4$ reach much lower infidelity after optimization. 
The improvement comes with increased runtime, as deeper circuits contain more parameters and induce more Pauli-path branching. 
The optimization curves further show that \texttt{SPPS} steadily decreases exact infidelity over training, indicating that the stochastic gradients remain informative for global state-preparation objectives.

\begin{figure}[t]
    \centering
    \includegraphics[width=0.98\linewidth]{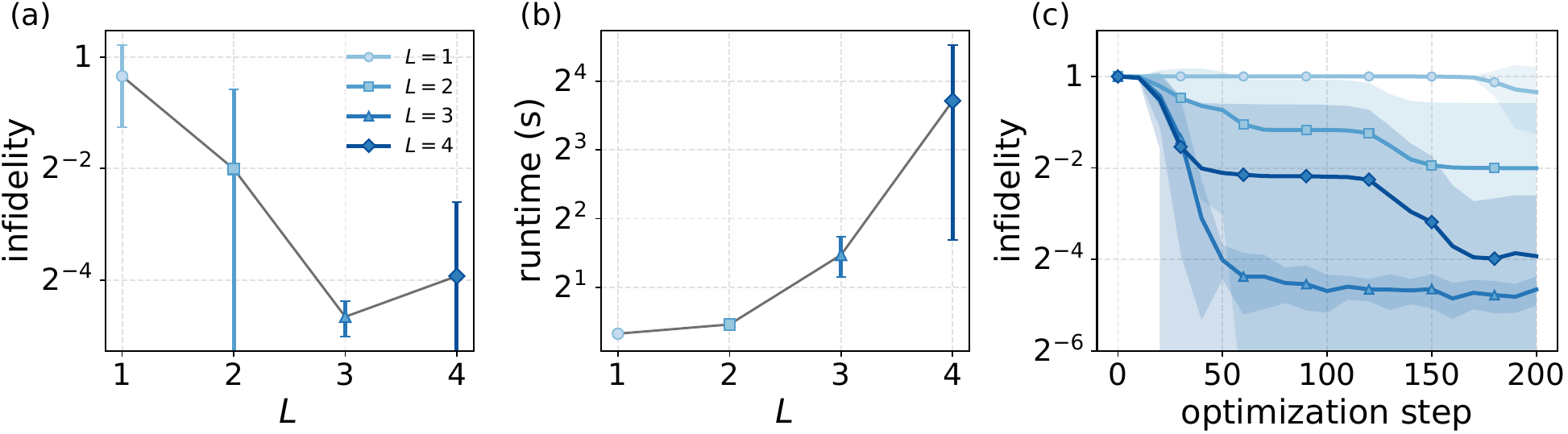}
    \caption{
{
    \textbf{\texttt{SPPS} for preparing MNIST amplitude-encoding circuits.}
    Figures (a) and (b) show the endpoint exact infidelity and cumulative \texttt{SPPS} gradient-estimation runtime, respectively.
    Figure (c) shows the exact infidelity during optimization.
    Error bars and shaded regions indicate the standard deviation over $5$ MNIST images.
    }}
    \label{dpp_fig_app_stateprep}
\end{figure}

\end{document}